\documentclass[pra,superscriptaddress]{revtex4}
\usepackage{amsfonts,amssymb,amsmath}
\usepackage[]{graphics,graphicx,epsfig}
\usepackage{amsthm}
\usepackage{color}
\usepackage{enumerate}
\usepackage{algorithm}
\usepackage{algpseudocode}
\usepackage{pifont}
\usepackage{pst-all} 
\usepackage[color]{changebar}
\usepackage{booktabs}
\usepackage{multirow}
\usepackage{siunitx}

\usepackage{amsmath,amsthm,amsfonts,amssymb,amscd}
\usepackage[utf8]{inputenc}
\usepackage{indentfirst}
\usepackage{graphicx}
\usepackage[T1]{fontenc}
\usepackage{mathpazo}
\usepackage{changes}

\usepackage[USenglish]{babel}
\usepackage[colorlinks=true]{hyperref}
\usepackage{hyperref}
\usepackage{bm}
\usepackage{lscape}
\def\id{\leavevmode\hbox{\small1\kern-3.8pt\normalsize1}}

\def\identity{\leavevmode\hbox{\small1\kern-3.8pt\normalsize1}}

\renewcommand{\epsilon}{\varepsilon}

\newtheorem{claim}{Claim}

\newtheorem{definition}{Definition} 

\newtheorem{prop}[definition]{Proposition}
\newtheorem{lemma}[definition]{Lemma}
\newtheorem{remark}[definition]{Remark}

\newtheorem{Eg}[definition]{Example}

\def\ba#1\ea{\begin{align}#1\end{align}}
\def\ban#1\ean{\begin{align*}#1\end{align*}}

\newcommand{\be}{\begin{equation}}
\newcommand{\ee}{\end{equation}}

\def\squareforqed{\hbox{\rlap{$\sqcap$}$\sqcup$}}
\def\qed{\ifmmode\squareforqed\else{\unskip\nobreak\hfil
\penalty50\hskip1em\null\nobreak\hfil\squareforqed
\parfillskip=0pt\finalhyphendemerits=0\endgraf}\fi}
\def\endenv{\ifmmode\;\else{\unskip\nobreak\hfil
\penalty50\hskip1em\null\nobreak\hfil\;
\parfillskip=0pt\finalhyphendemerits=0\endgraf}\fi}

\newcommand{\<}{\langle}
\renewcommand{\>}{\rangle}

\def\id{{\operatorname{id}}}

\def\be{\begin{equation}}
\def\ee{\end{equation}}
\def\ben{\begin{eqnarray}}
\def\een{\end{eqnarray}}

\def\bei{\begin{itemize}}
\def\eei{\end{itemize}}

\mathchardef\ordinarycolon\mathcode`\:
\mathcode`\:=\string"8000
\def\vcentcolon{\mathrel{\mathop\ordinarycolon}}
\begingroup \catcode`\:=\active
  \lowercase{\endgroup
  \let :\vcentcolon
  }

\newcommand{\nc}{\newcommand}
 \nc{\proj}[1]{|#1\rangle\!\langle #1 |} 
\nc{\avg}[1]{\langle#1\rangle}

\nc{\conv}{\operatorname{conv}}
\nc{\ox}{\otimes} \nc{\dg}{\dagger} \nc{\dn}{\downarrow}
\nc{\lmax}{\lambda_{\text{max}}}
\nc{\lmin}{\lambda_{\text{min}}}

\nc{\csupp}{{\operatorname{csupp}}}
\nc{\qsupp}{{\operatorname{qsupp}}} \nc{\var}{\operatorname{var}}
\nc{\rar}{\rightarrow} \nc{\lrar}{\longrightarrow}
\nc{\poly}{\operatorname{poly}}
\nc{\polylog}{\operatorname{polylog}} \nc{\Lip}{\operatorname{Lip}}
\nc{\Om}{\Omega}
\nc{\wt}[1]{\widetilde{#1}}

\def\>{\rangle}
\def\<{\langle}

\nc{\glneq}{{\raisebox{0.6ex}{$>$}  \hspace*{-1.8ex} \raisebox{-0.6ex}{$<$}}}
\nc{\gleq}{{\raisebox{0.6ex}{$\geq$}\hspace*{-1.8ex} \raisebox{-0.6ex}{$\leq$}}}
\nc{\vholder}[1]{\rule{0pt}{#1}}
\nc{\wh}[1]{\widehat{#1}}
\nc{\h}[1]{\widehat{#1}}
\nc{\ob}[1]{#1}
\def\beq{\begin {equation}}
\def\eeq{\end {equation}}
\def\be{\begin{equation}}
\def\ee{\end{equation}}

\nc{\eq}[1]{(\ref{eq:#1})} 
\nc{\eqs}[2]{\eq{#1} and \eq{#2}}

\nc{\eqn}[1]{Eq.~(\ref{eqn:#1})}
\nc{\eqns}[2]{Eqs.~(\ref{eqn:#1}) and (\ref{eqn:#2})}

\nc{\region}{\cS\cW}

\newenvironment{protocol*}[1]
  {
    \begin{center}
      \hrulefill\\
      \textbf{#1}
  }
  {
    \vspace{-1\baselineskip}
    \hrulefill
    \end{center}
  }

\begin{document}
\title{Relativistic Causality vs. No-Signaling as the limiting paradigm for correlations in physical theories}
\author{Pawe{\l} Horodecki}
\affiliation{Faculty of Applied Physics and Mathematics, National Quantum Information Center, Gda\'nsk University of Technology, 80-233 Gda\'nsk, Poland}
\author{Ravishankar \surname{Ramanathan}}
\affiliation{Institute of Theoretical Physics and Astrophysics, National Quantum Information Centre, Faculty of Mathematics, Physics and Informatics, University of Gda\'nsk, 80-308 Gda\'nsk, Poland} 
\affiliation{Laboratoire d'Information Quantique, Universit\'{e} Libre de Bruxelles, Belgium}

\begin{abstract}
The ubiquitous no-signaling constraints state that the probability distribution of the outputs of any subset of parties is independent of the inputs of the complementary set of parties; here we re-examine these constraints to see how they arise from relativistic causality. We show that while the usual no-signaling constraints are sufficient, in general they are not necessary to ensure that a theory does not violate causality. Depending on the exact space-time coordinates of the measurement events of the parties participating in a Bell experiment, relativistic causality only imposes a subset of the usual no-signaling conditions. We first revisit the derivation of the two-party no-signaling constraints from the viewpoint of relativistic causality and show that the they are both necessary and sufficient to ensure that no causal loops appear. We then consider the three-party Bell scenario and identify a measurement configuration in which a subset of the no-signaling constraints is sufficient to preserve relativistic causality. We proceed to characterize the exact space-time region in the tripartite Bell scenario where this phenomenon occurs.
Secondly, we examine the implications of the new relativistic causality conditions for security of device-independent cryptography against an eavesdropper constrained only by the laws of relativity. We show an explicit attack in certain measurement configurations on a family of randomness amplification protocols based on the $n$-party Mermin inequalities that were previously proven secure under the old no-signaling conditions. We also show that security of two-party key distribution protocols can be compromised when the spacetime configuration of the eavesdropper is not constrained.  Thirdly, we show how the monogamy of non-local correlations that underpin their use in secrecy can also be broken under the relativistically causal constraints in certain measurement configurations. We then inspect the notion of free-choice in the Bell experiment and propose a definition of free-choice in the multi-party Bell scenario. We re-examine the question whether quantum correlations may admit explanations by finite speed superluminal influences propagating between the spacelike separated parties. 
Finally, we study the notion of genuine multiparty non-locality in light of the new considerations and propose a new class of causal bilocal models in the three-party scenario. We propose a new Svetlichny-type inequality that is satisfied by the causal bilocal model and show its violation within quantum theory.    
\end{abstract}

\maketitle

\section{Introduction.}
The recent experimental confirmation of the violation of Bell inequalities \cite{EPR, Bell} in systems of electron spins \cite{Hensen}, entangled photons \cite{Giustina, Shalm}, etc. has made a compelling case for
the ``non-locality" of quantum mechanics. Quantum phenomena
exhibit correlations between space-like separated measurements that appear to be inconsistent with
any local hidden variable explanation. 
The ``spooky action at a distance" of quantum non-locality is now embraced and utilized in fundamental applications such as device-independent cryptography and randomness generation \cite{VV, Pironio, our} and reductions in communication complexity \cite{CC}. Moreover, this non-locality has also been used to show that even a tiny amount of free-randomness can be amplified \cite{CR, our} and that extensions of quantum theory which incorporate a particular notion of free-choice cannot have better predictive power than quantum theory itself \cite{CR2}. The quantum non-local correlations are known to be fully compatible with the no-signaling principle, i.e., the space-like separated parties cannot use the non-local correlations to communicate superluminally \cite{ER89}.   

Since the proposal of Popescu and Rohrlich \cite{PR}, it has been realized that non-local correlations might take on a more fundamental aspect. Not only quantum theory but any future theory that might contain the quantum theory as an approximation is now expected to incorporate non-locality as an essential intrinsic feature. This program has led to the formulation of device-independent information-theoretic principles \cite{Principles2, Principles3} that attempt to \textit{derive} the set of quantum correlations from amongst all correlations obeying the no-signaling principle. In parallel, cryptographic protocols have been devised based on the input-output statistics in Bell tests such that their proof of security only relies on the no-signaling principle. When one considers such post-quantum cryptography \cite{BHK, BCK}, randomness amplification \cite{CR, Acin, our, our2}, etc. the eavesdropper Eve is assumed to be limited to the preparation of boxes (input-output statistics) obeying a set of constraints collectively referred to as the no-signaling constraints. The general properties of no-signaling theories have been investigated \cite{Mas06} in a related program to formulate an information-theoretic axiomatic framework for quantum theory. On the other hand, quantum theory does not provide a mechanism for the non-local correlations. Several theoretical proposals have been put forward to explain the phenomenon of non-local correlations between quantum particles via superluminal communication between them. These models go beyond quantum mechanics but reproduce the experimental statistical predictions of quantum mechanics, the most famous of these models being the de Broglie-Bohm pilot wave theory \cite{Hol93, BH93}.      

In all relativistic theories, ``causality," is imposed i.e., the requirement that causes must precede effects in all space-time rest frames. 
Before going further, two remarks are in order here. First, by an
effect we mean any possible event, even if it has been affected by
other events (causes) indirectly.
Second, we shall use as general a correlation point of
view as possible, \textit{regardless} of the physical
theory from which the correlations arise. In this context, let us also note that given an arbitrary space-time structure, the question of causal order for any two measures
has been formalized using intuitions from optimal transport theory \cite{Eckstein1}.
Furthermore it has been proven that in any hyperbolic space-time, casuality of the 
evolution of measures supported on time-slices is an observer independent concept \cite{Miller}.
Both formalisms have also been applied in the study of the evolution of the 
quantum wavepackets showing in particular the natural consistency of the causality concept 
with the relativistic continuity equation \cite{Eckstein-Miller-models}.
From the perspective of communication, the requirement of relativistic causality strictly demands that no faster-than-light (FTL) transmission of information takes place between a sender and a receiver. The no-signaling principle being in ubiquitous use (in device-independent cryptography, axiomatic formulations, etc. as explained earlier), a natural question is to explore whether the no-signaling constraints that are currently in use precisely capture the constraints imposed by relativistic causality, i.e., to derive the no-signaling constraints from relativistic causality. 
In this paper, we investigate this question and find that, surprisingly, the answer is that the multi-party no-signaling principle requires a modification to capture the notion of causality. In particular, the actual no-signaling constraints that one must impose in a multi-party Bell experiment are dictated by the space-time configuration of the measurement events in the experiment. This modification of the no-signaling constraints naturally leads to a number of implications in tasks of device-independent cryptography against relativistic eavesdroppers and to explanations of quantum correlations via finite-speed superluminal influences. 
This paper is therefore, propaedeutic to a larger project undertaken by the authors to explore the relativistic causality constraints and their implications on the foundations of quantum theory. 

%

The structure of the paper is as follows. We initially establish the setup of the Bell experiment and recall the assumptions in the Bell theorem. We then define the notion of relativistic causality that we use in this paper (and that is commonly accepted, i.e., that there be no causal loops in spacetime) and revisit the derivation of the two-party no-signaling constraints from causality. 
We then show that, in the multi-party scenario, only a restricted subset of the no-signaling constraints is required to ensure that no causal loops appear. We explicitly identify a region of space-time for the measurement events in a Bell scenario where the usual no-signaling constraints fail. 
In this regard, we extend a particular framework of ``jamming" non-local correlations by Grunhaus, Popescu and Rohrlich in \cite{GPR} based upon an earlier suggestion of Shimony in \cite{Shim83}. 
We then examine the implications of the restricted subset of no-signaling constraints for device-independent cryptographic tasks against an eavesdropper constrained only by the laws of relativity. We detail explicit attacks on known protocols for randomness amplification based on the GHZ-Mermin inequalities using boxes that obey the new relativistic causality conditions. We show that from that perspective the security theory needs revision, and - in a way - to the some degree collapses. 
We also explore the implications on some of the known features of no-signaling theories \cite{Mas06}, in particular we find that the phenomenon of monogamy of correlations is significantly weakened in the relativistically causal theories and that the monogamy of CHSH inequality violation \cite{Toner} disappears in certain spacetime configurations. The notions of freedom-of-choice and no-signaling are known to be intimately related \cite{CR}. We re-examine how the notion of free-choice as proposed by Bell and formalized by Colbeck and Renner \cite{Bell2, Renner-Colbeck} can be stated mathematically within the structure of a space-time configuration of measurement events. A breakthrough result in \cite{BPAL+12} was a claim that any finite superluminal speed explanation of quantum correlations could lead to superluminal signaling and must hence be discarded. We re-examine this question in light of the modified relativistic causality and free-will conditions. 
Both non-relativistic quantum theory and relativistic quantum field theory \cite{ER89} are well-known to obey a no superluminal signaling condition, proposals to modify quantum theory by introducing non-linearities have been shown to lead to signaling \cite{Wei89, Cza91, Gis90}. We end with discussion and open questions concerning feasible mechanisms for the 
point-to-region superluminal influences.

\section{Preliminaries.}

\subsection{Notation.}
Let us first establish the setup of a typical Bell experiment. In the Bell scenario denoted $\textbf{B}(n,m,k)$, we have $n$ space-like separated parties, each of whom chooses from among $m$ possible measurement settings and obtains one of $k$ possible outcomes. In general, the number of inputs and outputs for each party may vary, but this will not concern us in this paper. The inputs of the $i$-th party will be denoted by random variable (r.v.) $X_i$ taking values $x_i$ in $[m] = \{1, \dots, m\}$ and the outputs of this party will be denoted by r.v $A_i$ taking values $a_i$ in $[k]$. Accordingly, the conditional probability distribution of the outputs given the inputs will be denoted by
\begin{equation}
P_{A_1, \dots, A_n | X_1, \dots, X_n}(a_1, \dots, a_n | x_1, \dots, x_n) := P(A_1 = a_1, \dots, A_n = a_n | X_1 = x_1, \dots, X_n = x_n).
\end{equation}
Following \cite{Renner-Colbeck, CR}, we also consider the notion of a spacetime random variable (SRV), which is a random variable $R$ together with a set of spacetime coordinates $(t_R, \textbf{r}_R) \in \mathbb{R}^4$ in some inertial reference frame $\mathcal{I}$ at which it is generated. A measurement event $\mathcal{M}_{X,A}$ is thus modeled as an input SRV $X$ together with an output SRV $A$. As in typical studies of Bell experiments, here we consider the measurement process as instantaneous, i.e., $X$ and $A$ share the same spacetime coordinates. Denote a causal order relation between two SRV's $X$ and $A$ by $X \rightarrow A$ if $t_X < t_A$ in all inertial reference frames, i.e., $A$ is in the future light cone of $X$ (so that $X$ may cause $A$). A pair $(A_j, A_k)$ of SRVs is spacelike separated if $\Delta s^2 := | \textbf{r}_{A_j} - \textbf{r}_{A_k} |^2 - c^2 (t_{A_j} - t_{A_k})^2 > 0$. 

In this paper, we will have occasion to distinguish the specific spacetime location at which \textit{correlations} between random variables manifest themselves, i.e., the particular spacetime location at which the correlations are registered, from the spacetime locations at which the random variables themselves are generated. Accordingly, we label by $\mathcal{C}_{A_j,A_k}$ the SRV denoting the correlations between the output SRVs $A_j$ and $A_k$ with its associated spacetime location $(t_{\mathcal{C}_{A_j,A_k}}, \textbf{r}_{\mathcal{C}_{A_j,A_k}})$ being at the earliest (smallest $t$) intersection of the future light cones of $A_j$ and $A_k$ in the reference frame $\mathcal{I}$.  

\subsection{The two-party Bell theorem.}
\label{subsec:two-party-Bell-theorem}

Consider the two-party Bell experiment with two spacelike separated parties Alice and Bob. The inputs of Alice and Bob are denoted by SRV's $X_1, X_2$ and the outputs by $A_1, A_2$. If the measurement process is considered to be instantaneous, $X_1$ and $A_1$ share the same space-time coordinates as do $X_2, A_2$. 
In the two-party Bell scenario, the results of the experiment are described by the set of conditional probability distributions $P_{A_1, A_2 | X_1, X_2}(a_1, a_2 | x_1, x_2)$. Let $\Lambda$ denote a set of underlying variables describing the state of the system under consideration, these could in general be local or non-local. We have evidently
\begin{equation}
\label{eq:two-party-BI}
P_{A_1, A_2 | X_1, X_2}(a_1, a_2 | x_1, x_2) =  \int d\lambda P_{A_1, A_2 | X_1, X_2, \Lambda}(a_1, a_2 | x_1, x_2, \lambda) P_{\Lambda | X_1, X_2}(\lambda | x_1, x_2) 
\end{equation}


In this scenario, the Bell theorem \cite{Bell} is based on the following assumptions \cite{Hall, Valdenebro}: 
\begin{enumerate}

\item \textit{Outcome Independence}: The statistical correlations between the outputs arise from ignorance of the underlying variable $\Lambda$
\begin{equation}
\label{eq:two-party-oi}
P_{A_1, A_2|X_1, X_2, \Lambda}(a_1, a_2 | x_1, x_2, \lambda) = P_{A_1 | X_1, X_2, \Lambda}(a_1|x_1, x_2, \lambda) P_{A_2|X_1, X_2, \Lambda}(a_2 | x_1, x_2, \lambda) 
\end{equation}
This is evidently true for deterministic models (those that output deterministic answers $a_1, a_2$ for inputs $x_1, x_2$) as well as for stochastic models and is motivated by a notion of realism, i.e., that the measurements merely reveal pre-existing outcomes encoded in $\Lambda$. 

\item \textit{Parameter-Independence}: For each so-called ``microstate" $\lambda$, the probability of an outcome on Alice's side is assumed to be (stochastically) independent of the experimental setting (the parameters of the device) on Bob's side, 
\begin{eqnarray}
\label{eq:two-party-pi}
P_{A_1|X_1, X_2, \Lambda}(a_1|x_1, x_2, \lambda) &=& P_{A_1 | X_1, \Lambda}(a_1 | x_1, \lambda), \nonumber \\
P_{A_2|X_1, X_2, \Lambda}(a_2 | x_1, x_2, \lambda) &=& P_{A_2 | X_2, \Lambda}(a_2 | x_2, \lambda).
\end{eqnarray}
The justification for this assumption comes from special relativity, which imposes that spacelike separated measurements do not influence each other's underlying outcome probability distributions. 

\item \textit{Measurement Independence}: The measurement inputs are uncorrelated with the underlying variable $\Lambda$
\begin{equation}
\label{eq:two-party-mi}
P_{\Lambda, X_1, X_2}(\lambda, x_1, x_2) = P_{\Lambda}(\lambda) P_{X_1, X_2}(x_1, x_2). 
\end{equation}
This is also sometimes called the \textit{free-choice} assumption, i.e., that the inputs $x_1, x_2$ are chosen freely, independent of $\Lambda$. We elaborate on the notion of freeness in this two-party Bell experiment in Section \ref{subsec:two-party-freewill} \cite{Renner-Colbeck, Bell2}. 
\end{enumerate}
There are other assumptions such as: Fair Sampling, No Backward Causation, Reality being single valued, etc. \cite{Valdenebro, Tumulka} but these will not concern us in this paper. Substituting Eqs.(\ref{eq:two-party-oi}, \ref{eq:two-party-pi}, \ref{eq:two-party-mi}) in Eq.(\ref{eq:two-party-BI}), we obtain the description of a local hidden variable model:
\begin{equation}
\label{eq:LHV-model}
P_{A_1, A_2|X_1, X_2}(a_1, a_2 | x_1, x_2) = \int d\lambda P_{\Lambda}(\lambda) P_{A_1| X_1, \Lambda}(a_1| x_1, \lambda) P_{A_2|X_2, \Lambda}(a_2 | x_2, \lambda).
\end{equation}
Figure \ref{fig:two-party-caus-struc} shows the causal structure represented by the two-party Bell experiment represented in terms of a Directed Acyclic Graph (DAG) \cite{Pearl09}. Remark that the local hidden variable model in Eq.(\ref{eq:LHV-model}) can also be arrived at starting from other postulates such as local causality \cite{Norsen}.

\subsection{Two-party no-signaling}
From the assumptions of parameter-independence and measurement-independence, one can deduce the \textit{no-signaling constraints}:
\begin{eqnarray}
\label{eq:two-party-ns}
P_{A_1|X_1, X_2}(a_1|x_1, x_2) &=& \int d\lambda P_{\Lambda | X_1, X_2}(\lambda | x_1, x_2) P_{A_1|X_1,X_2,\Lambda}(a_1|x_1, x_2, \lambda) \nonumber \\
&\stackrel{Eq.(\ref{eq:two-party-mi})}{=}& \int d\lambda P_{\Lambda}(\lambda) P_{A_1|X_1, X_2,\Lambda}(a_1|x_1, x_2, \lambda) \nonumber \\
&\stackrel{Eq.(\ref{eq:two-party-pi})}{=}& \int d\lambda P_{\Lambda}(\lambda) P_{A_1|X_1, \Lambda}(a_1 | x_1, \lambda) \nonumber \\
&=& P_{A_1 | X_1}(a_1 | x_1). \nonumber \\
P_{A_2|X_1, X_2}(a_2 | x_1, x_2) &=& \int d\lambda P_{\Lambda | X_1, X_2}(\lambda | x_1, x_2) P_{A_2|X_1,X_2,\Lambda}(a_2|x_1, x_2, \lambda) \nonumber \\
&\stackrel{Eq.(\ref{eq:two-party-mi})}{=}& \int d\lambda P_{\Lambda}(\lambda) P_{A_2|X_1, X_2,\Lambda}(a_2|x_1, x_2, \lambda) \nonumber \\
&\stackrel{Eq.(\ref{eq:two-party-pi})}{=}& \int d\lambda P_{\Lambda}(\lambda) P_{A_2|X_2, \Lambda}(a_2 | x_2, \lambda) \nonumber \\
&=& P_{A_2|X_2}(a_2 | x_2).
\end{eqnarray}
The no-signaling conditions in this two-party Bell experiment formally state that probability distribution of the outcomes of any party is independent of the input of the other party. 

\subsection{Two-party freedom-of-choice}
\label{subsec:two-party-freewill}
The assumption of measurement-independence or free-choice can be formally expressed in terms of spacetime random variables. In the two-party Bell experiment, one imposes the free-choice constraints
\begin{eqnarray}  
\label{eq:two-party-freewill}
P_{X_1|\Lambda, X_2, A_2}(x_1|\lambda, x_2, a_2) = P_{X_1}(x_1), \nonumber \\
P_{X_2| \Lambda, X_1, A_2}(x_2|\lambda, x_1, a_1) = P_{X_2}(x_2).
\end{eqnarray}
One formal way to define the freedom-of-choice condition in terms of spacetime random variables was formulated by Colbeck and Renner (CR) in \cite{Renner-Colbeck}. Recall that we denote a causal order relation by $X \rightarrow A$ if $t_X < t_A$ in all inertial reference frames, i.e., $A$ is in the future light cone of $X$ so that $X$ \text{may} cause $A$. CR formulate the notion of free-choice (formalising Bell's notion from \cite{Bell2}) as follows: 
\begin{definition}[\cite{Renner-Colbeck}, \cite{Bell2}]
\label{def:freewill}
A spacetime random variable $X$ is said to be free if $X$ is uncorrelated with every spacetime random variable $A$ such that $X \nrightarrow A$, i.e., for all such $A$, we have $P_{X|A}(x|a) = P_{X}(x)$. 
\end{definition}
In other words, $X$ is free if it is uncorrelated with any $A$ that it could not have caused, where the requirement for $X$ causing $A$ is that $A$ lies in the future light cone of $X$. Clearly, Eq.(\ref{eq:two-party-freewill}) follows if one adopts this notion of freedom-of-choice, although note that the Def. \ref{def:freewill} is strictly stronger than just the conditions (\ref{eq:two-party-freewill}) imposed in the Bell theorem. 

\subsection{A sufficient set of multi-party no-signaling conditions}
In this paper, our focus is on the no-signaling and free-will constraints above which are intimately related to each other. In particular, we consider the generalization of the no-signaling conditions to the multi-party scenario. The generalized multi-party no-signaling constraints are usually stated as follows (see for example \cite{Mas06}): 
\begin{eqnarray}
\label{eq:Multi-party-NS}
\sum_{a_j} P_{A_1, \dots, A_n|X_1, \dots, X_n}(a_1, \dots, a_j, \dots, a_n | x_1, \dots, x_j, \dots, x_n) &=& \sum_{a_j} P_{A_1, \dots, A_n|X_1, \dots, X_n}(a_1, \dots, a_j, \dots, a_n | x_1, \dots, x'_j, \dots, x_n) \nonumber \\
&& \qquad \forall j \in [n], \{a_1, \dots, a_n\} \setminus a_j, \{x_1, \dots, x_j, x'_j, \dots, x_n\} 
\end{eqnarray}
In words, the above constraints state that the outcome distribution of any subset of parties is independent of the inputs of the complementary set of parties (while Eq.(\ref{eq:Multi-party-NS}) imposes this for subsets of $n-1$ parties, one can straightforwardly show that this also implies that the marginal distribution for smaller sized subsets is well-defined \cite{Mas06}). Now, given that as stated earlier the justification of the two-party no-signaling constraint came from the causality constraints of special relativity, the natural question which we investigate in this paper is to what extent the multi-party no-signaling constraints in Eq.(\ref{eq:Multi-party-NS}) are imposed by the causality constraints of special relativity. In other words, while the constraints in Eq.(\ref{eq:Multi-party-NS}) are clearly sufficient to ensure that no superluminal signaling takes place, we derive the set of necessary and sufficient constraints that ensure that no causal loops appear in the theory. We will see that the multi-party no-signaling conditions in fact depend on the space-time coordinates of the measurement events in the Bell experiment and we propose an appropriate modification of these constraints. 

\section{Results.}
The main results of the paper are as follows.  
\begin{enumerate}
\item We derive in Prop. \ref{prop:two-party-NS} the no-signaling constraints in the two-party Bell experiment from relativistic causality and show that these constraints are both necessary and sufficient in this scenario.

\item We show that in general, depending on the exact space-time coordinates of the measurement events of the parties participating in a Bell experiment, relativistic causality strictly only imposes a subset of the usual no-signaling conditions. In particular, we show that in the three-party Bell experiment, two sets of relativistic causal constraints are possible: $(a)$ the usual no-signaling conditions and $(b)$ a subset of the no-signaling conditions in Prop. \ref{prop:rel-cau-constraints} which ensure causality is preserved in certain measurement configurations (such as in Fig. \ref{fig:three-party-meas-config}). 

\item We geometrically analyze in Prop.\ref{prop:spacetime-region} the exact space-time region in the three-party Bell scenario where the constraints in Prop. \ref{prop:rel-cau-constraints} are necessary and sufficient to prevent causal loops. 

\item We propose in Definition \ref{lem:mod-freewill} a rigorous definition of the free-choice constraints in the multi-party Bell experiment and show in Prop. \ref{prop:caus-eq-freewill} how these are compatible with the relativistic causality conditions. 

\item We examine the implications of the causality constraints on device-independent cryptography against an adversary constrained only by the laws of relativity. In the task of randomness amplification, we demonstrate in Prop. \ref{prop:RA-rel-caus} an explicit attack that renders a family of multi-party protocols based on the $n$-party Mermin inequalities insecure, when the measurement events of the honest parties conform to certain spacetime configurations. In the task of key distribution, we show in Prop. \ref{prop:qkd-sec-rel-caus} that even two-party protocols can be rendered insecure unless assumptions are made concerning the spacetime location of the eavesdropper's measurement event or shielding of the honest parties to any possible superluminal influences, even those respecting causality. 

\item We then examine properties that were previously considered to be common to all no-signaling theories. In particular, we show in Prop. \ref{prop:rel-caus-mono} that the paradigmatic phenomenon of monogamy of correlations violating the CHSH inequality no longer holds in relativistically causal theories in certain measurement configurations.

\item We re-examine, in Section \ref{sec:v-causal-exp} the question whether quantum correlations may admit explanations by finite speed superluminal influences propagating between the spacelike separated parties, by modifying the argument from \cite{BPAL+12} against such $v$-causal models to show that they would lead to causal loops in certain measurement configurations. 

\item Finally, we investigate the phenomenon of multiparty non-locality in light of the new considerations and propose in Def. \ref{def:rel-caus-bilocal} a new class of relativistically causal bilocal models in the three-party scenario. We formulate a new Svetlichny-type inequality in Lemma \ref{lem:RCBL-ineq} that is satisfied by this class of models and examine its violation within quantum theory.  

\end{enumerate}

\section{Relativistic Causality.}
In this paper, we work in the regime of special relativity, i.e., flat spacetime with no gravitational fields, this regime is valid within small regions of spacetime where the non-uniformities of any gravitational forces are too small to measure. 
We consider the causal structure of measurement events occurring at fixed spacetime locations $(t,\textbf{r})$. Within this regime, the relativistic causality constraint we consider is simply that  :
\begin{itemize}
\item No \textit{causal loops} occur, where a causal loop is a sequence of events, in which one event is among the causes of another event, which in turn is among the causes of the first event. 
\end{itemize}
Causality implies that for two causally related events taking place at two spatially separated points, the cause always occurs before the effect, and this sequence cannot be changed by any choice of a frame of reference. Relativistic Causality is a consequence of prohibiting faster-than-light transmission of information from one spacetime location $A$ to another space-like separated location $B$ in any inertial frame of reference. A violation of this condition would lead to the well-known ``grandfather-paradoxes".  In other words, if an effect $B$ that occurs at a space-time location $(t_B, \textbf{r}_B)$ in an inertial reference frame $\mathcal{I}$ \textit{precedes} its cause $A$ that occurs at space-time location $(t_A, \textbf{r}_A)$ in $\mathcal{I}$ (i.e., $t_B < t_A$) then an observer at $B$ may in turn affect the cause at $A$ and lead to a paradox. An explicit example of a closed causal loop is shown in the proof of Prop. \ref{prop:two-party-NS}. 

A few remarks are in order. Firstly, note that faster-than-light propagation in one privileged frame of reference alone does not lead to causal loops. Secondly, note that the "principle of causality", i.e., the invariance of the temporal
sequence of causally related events is also directly related to the macroscopic notion of arrow of time from the second law of thermodynamics \cite{Terletskii68}. Finally, we remark that in the General Theory of Relativity, the field equations allow for solutions in the form of closed timelike curves, and there has been much debate over these, with proposals such as the chronology protection conjecture and a self-consistency principle \cite{self-cons, SWH} suggested to prevent time-travel paradoxes.

\section{Deriving No-Signaling constraints from Relativistic Causality.}
\subsection{Derivation of the two-party no-signaling constraints.}

Let us first revisit the derivation of the no-signaling constraint from relativistic causality in the typical two-party Bell experiment. 
In any run of the Bell experiment, Alice \textit{freely} chooses at spacetime location $(t_A, \textbf{r}_{A})$ in $\mathcal{I}$, her measurement input $x_1 \in [m]$ and obtains (instantaneously) the output $a_1 \in [k]$. 
Similarly, Bob who is at a space-like separated location, in the same run, freely chooses at $(t_B, \textbf{r}_{B})$ his input $x_2 \in [m]$ and obtains an output $a_2 \in [k]$. 
The requirement of space-like separation, i.e.,
$\Delta s^2 = \left\vert \textbf{r}_{A}- \textbf{r}_B \right\vert^2 - c^2 \left(t_A - t_B \right)^2 > 0$,
ensures that the measurement events fall outside each other's light cone. 
The conditional probability distributions of outputs given the inputs $P_{A_1, A_2|X_1, X_2}(a_1,a_2|x_1,x_2)$ that they obtain in the Bell experiment jointly constitute a ``box" $\mathcal{P} := \{ P_{A_1, A_2|X_1, X_2}(a_1, a_2|x_1, x_2)\}$.   


The causality constraint is the requirement that faster-than-light information transmission from Alice to Bob or Bob to Alice is forbidden, i.e., Alice and Bob cannot use their local measurements to signal to one another. Note that here, Alice and Bob choose their measurement inputs \textit{freely}, i.e., we assume the free-choice constraints in Eq.(\ref{eq:two-party-freewill}). 

\begin{center}
\begin{figure}[t!]
		\includegraphics[width=0.5\textwidth]{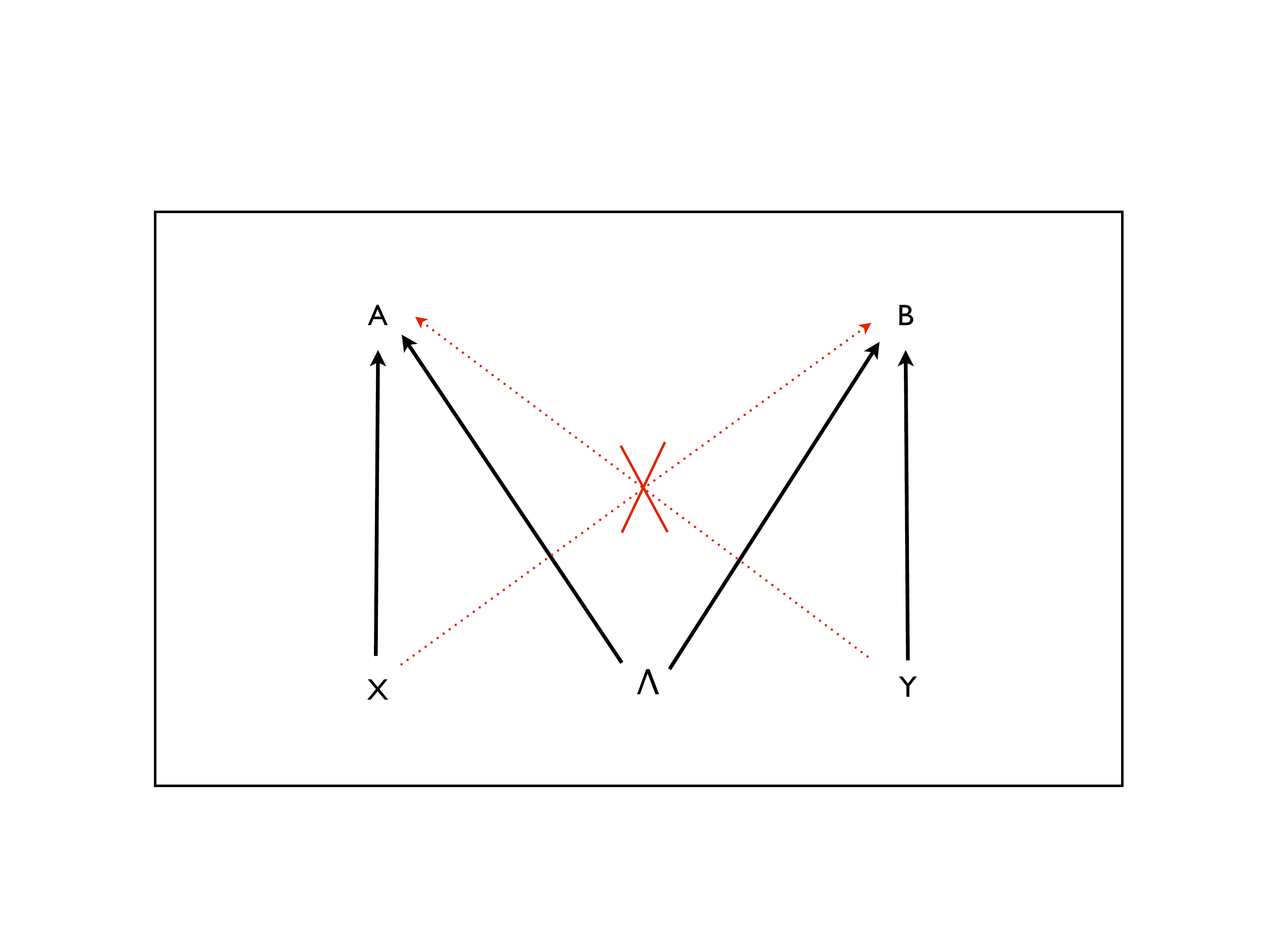}
\caption{The causal structure of the two-party Bell experiment represented by a directed acyclic graph. The inputs and outputs of the two parties are denoted by $X,Y$ and $A,B$ respectively. The inputs are chosen freely so no arrows are pointed towards $X,Y$. The correlations between the outputs $A,B$ may be explained as caused by a shared random variable $\Lambda$ which may be local or non-local depending on the theory under consideration. The input of each party is prevented from influencing the output distribution of the other party due to the spacelike separation between them.}
	\label{fig:two-party-caus-struc}
\end{figure}
\end{center}


\begin{prop}
\label{prop:two-party-NS}
In the two-party Bell experiment, the usual no-signaling constraints 
\begin{eqnarray}
\label{eq:two-party-NS}
\sum_{a_2} P_{A_1, A_2|X_1, X_2}(a_1, a_2|x_1, x_2) &=& \sum_{a_2} P_{A_1, A_2|X_1, X_2}(a_1, a_2|x_1, x'_2) =: P_{A_1|X_1}(a_1|x_1) \;\;\; \forall a_1,x_1,x_2,x'_2 \nonumber \\
\sum_{a_1} P_{A_1, A_2|X_1, X_2}(a_1, a_2|x_1, x_2) &=& \sum_{a_1} P_{A_1, A_2|X_1, X_2}(a_1, a_2|x'_1, x_2) =: P_{A_2|X_2}(a_2|x_2) \;\;\; \forall a_2,x_1,x'_1,x_2.  
\end{eqnarray} 
are necessary and sufficient to ensure that relativistic causality is not violated. 
\end{prop}
\begin{proof}
Suppose by contradiction that one of the constraints in Eq.(\ref{eq:two-party-NS}) was violated, for definiteness, let us suppose that the marginal distribution of Bob's outputs depended upon the input of Alice, i.e., suppose
\begin{equation}
\label{eq:Alice-signal-Bob}
\sum_{a_1} P_{A_1, A_2|X_1, X_2}(a_1, a_2|x_1, x_2) \neq \sum_{a_1} P_{A_1, A_2|X_1, X_2}(a_1, a_2|x'_1, x_2). 
\end{equation}
Note that by assumption, Alice chooses her measurement freely, i.e.,
\begin{equation}
\label{eq:Alice-freewill}
P_{X_1|X_2, A_2}(x_1 | x_2, a_2) = P_{X_1}(x_1), 
\end{equation} 
so suppose that in some specific run Alice chooses the input $X_1 = x_1$ at spacetime location $(t_A, \textbf{r}_A)$. 
Since Alice and Bob's measurement events are spacelike separated, Eq.(\ref{eq:Alice-signal-Bob}) would imply physically that a superluminal influence propagated from Alice's system to Bob's system informing Bob of Alice's input $x_1$. Explicitly, consider that Alice and Bob share many copies of a system obeying (\ref{eq:Alice-signal-Bob}), when Alice chooses input $x_1$ on all her subsystems, Bob guesses $x_1$ with a probability strictly larger than uniform by inspecting his local statistics $P_{A_2|X_1,X_2}$. This is illustrated in Fig.\ref{fig:two-party-ccl} where in a particular frame of reference $\mathcal{I}_{AB}$ labeled by axes $(t, \textbf{r})$ (for concreteness, this can be taken to be the laboratory frame at which Alice and Bob's systems are at rest), the superluminal influence is shown to be instantaneous, i.e., the signal travels parallel to the time axis $t$ and reaches Bob's system at spacetime location $(t_B, \textbf{r}_B)$. Now, consider the inertial reference frame $\mathcal{I}_{CD}$ labeled by $(t', \textbf{r}')$ of Charlie and Dave who are moving uniformly at some speed $v$ relative to $\mathcal{I}_{AB}$. At the space-time location $(t_B, \textbf{r}_B) = (t'_C, \textbf{r}'_C)$ Charlie and Bob's world-lines intersect, suppose Bob informs Charlie of the value $x_1$ at this point. Charlie immediately transmits this information via the same superluminal mechanism to Dave. In the frame $\mathcal{I}_{CD}$, this superluminal influence again travels parallel to the time axis $t'$. The information about $x_1$ thus reaches Dave at spacetime location $(t'_D, \textbf{r}'_D)$ which is in the causal past of $(t_A, \textbf{r}_A)$ as shown in Fig. \ref{fig:two-party-ccl}. Alice has thus managed to transmit the information about $x_1$ to her causal past. Dave may then transmit $x_1$ by a sub-luminal signal that reaches Alice at $(t_A, \textbf{r}_A)$ and Alice could now decide \textit{freely} to not measure $x_1$ and measure $x'_1$ instead. 
Therefore, we see that Eqs.(\ref{eq:Alice-signal-Bob}) and (\ref{eq:Alice-freewill}) have resulted in a causal loop, which would lead to grandfather-style paradoxes. In order to prevent causal loops and preserve the notion of relativistic causality while still keeping the notion of free will in Eq.(\ref{eq:Alice-freewill}), we therefore impose the no-signaling condition 
\begin{equation}
\sum_{a_1} P_{A_1, A_2|X_1, X_2}(a_1, a_2|x_1, x_2) = \sum_{a_1} P_{A_1, A_2|X_1, X_2}(a_1, a_2|x'_1, x_2).
\end{equation}               
Analogous reasoning to prevent a closed causal loop starting from Bob's measurement gives the other no-signaling condition
\begin{equation}
\sum_{a_2} P_{A_1, A_2|X_1, X_2}(a_1, a_2|x_1, x_2) = \sum_{a_2} P_{A_1, A_2|X_1, X_2}(a_1, a_2|x_1, x'_2).
\end{equation}
We have therefore derived the necessity of the two-party no-signaling constraints from the requirement that there be no closed causal loops. That these constraints are also sufficient is clear, since these ensure that the causal structure in Fig. \ref{fig:two-party-caus-struc} is maintained. Note that while the outputs can be correlated with each other, this correlation is attributed to the underlying (local or non-local) hidden variable $\Lambda$ rather than due to any superluminal influence propagating from one party to another.
\end{proof}

\begin{center}
\begin{figure}[t!]
		\includegraphics[width=0.85\textwidth]{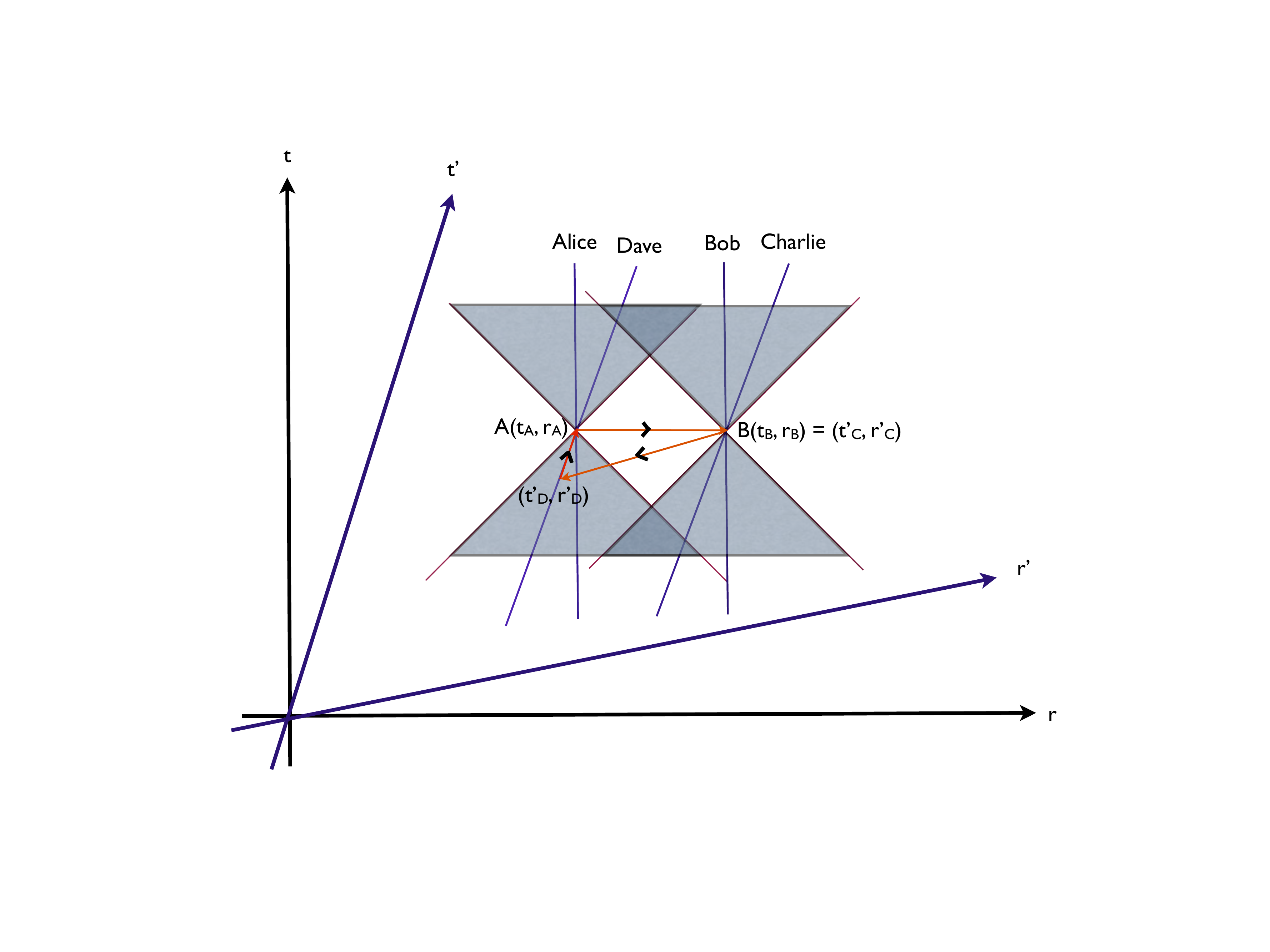}
\caption{An explicit violation of causality occurs when the two-party no-signaling condition in Eq.(\ref{eq:two-party-NS}) is violated as in Eq.(\ref{eq:Alice-signal-Bob}). In the figure, we consider the violation in terms of Bob's output distribution depending on Alice's input. Since Alice's input $x_1$ was chosen freely at $A$ $(t_A, \textbf{r}_A)$, this dependence implies a superluminal signal containing the information $x_1$ was transmitted from $A$ to $B$ at $(t_B, \textbf{r}_B)$, here the signal is shown to travel at infinite speed, i.e. parallel to the $t$-axis. The world-lines of Charlie and Dave who move at uniform velocity relative to Alice and Bob are also shown. At $(t_B, \textbf{r}_B) = (t'_C, \textbf{r}'_C)$, Charlie's world line interesects Bob who informs Charlie of $x_1$. Charlie then uses the same superluminal signal to send $x_1$ to Dave at $(t'_D, \textbf{r}'_D)$, this signal travels parallel to the $t'$-axis as shown. Note that $(t'_D, \textbf{r}'_D)$ is in the causal past of Alice, so Dave can send $x_1$ via a subluminal signal to Alice. Alice can then freely chose to not measure $x_1$ and measure $x'_1$ instead. Thus, we have a closed causal loop resulting in a violation of causality.}
	\label{fig:two-party-ccl}
\end{figure}
\end{center}

\begin{center}
\begin{figure}[t!]
		\includegraphics[width=0.5\textwidth]{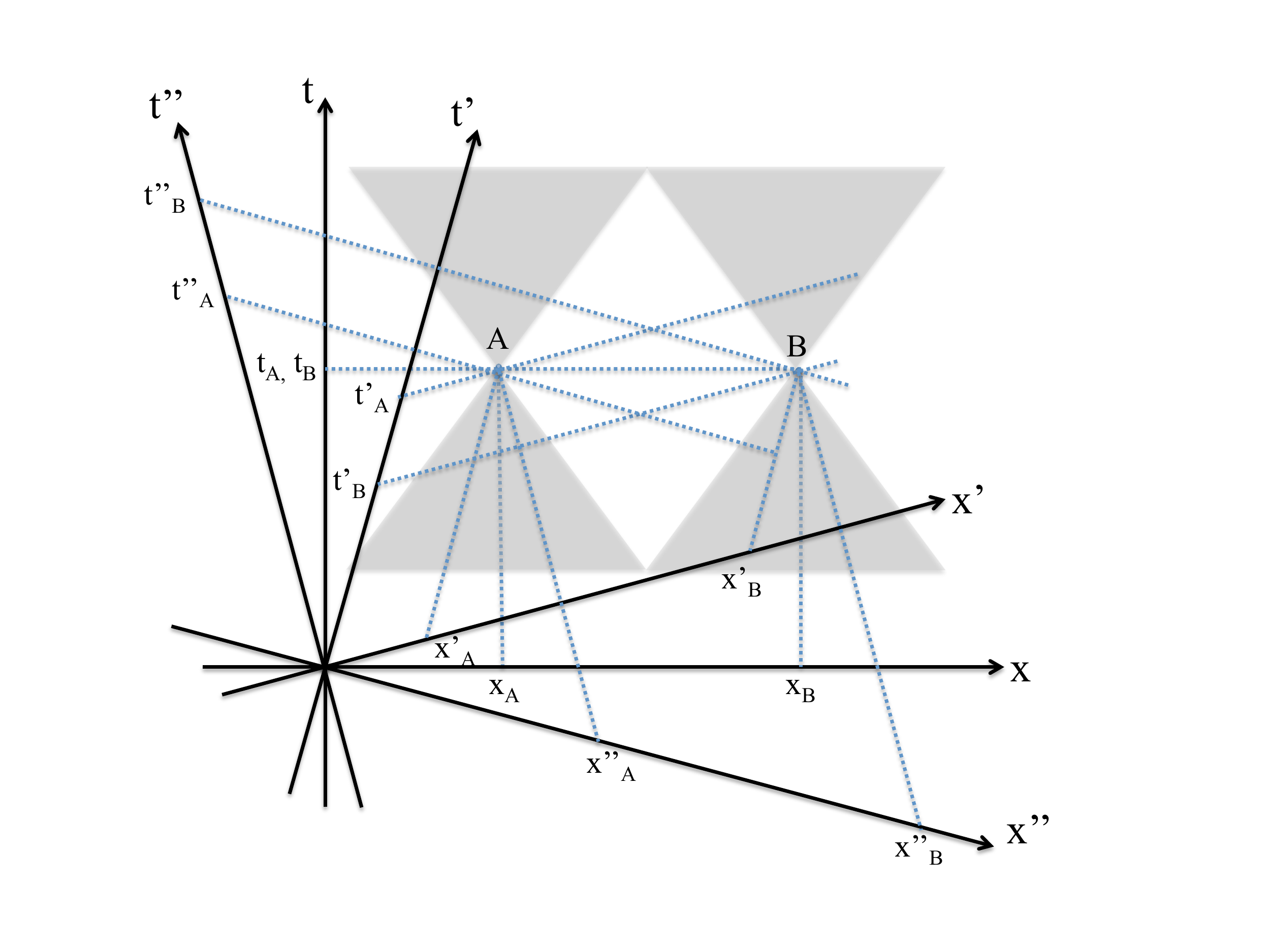}
\caption{In an inertial reference frame $(t,\textbf{r})$ the two (spatially separated) events $A$ and $B$ appear to occur simultaneously, i.e., $t_A = t_B$. In the inertial reference frame $(t', \textbf{r}')$, event $B$ occurs \textit{before} $A$, i.e., $t'_B < t'_A$. In another inertial reference frame $(t'', \textbf{r}'')$ on the other hand, event $A$ occurs before $B$, i.e., $t''_A < t''_B$. To maintain causality, we must have that the output distribution $P(a_1|x_1)$ at $A$ must be independent of the input $x_2$ at $B$, and similarly $P(a_2|x_2)$ must be independent of the input $x_1$ at $A$.}
	\label{fig:two-party-NS}
\end{figure}
\end{center}

This situation is illustrated in the space-time diagram of the measurement process in Fig. \ref{fig:two-party-NS}. As seen in the figure,  while causal relationships are manifestly Lorentz invariant, the specific time sequence of events changes in different inertial frames of reference, in particular spacelike separated events have no absolute time ordering between them. In the $(t,x)$ reference frame, the measurement events of Alice and Bob occur simultaneously. In the $(t',x')$ inertial reference frame, Bob's measurement precedes that of Alice ($t'_B < t'_A$), while in the $(t'', x'')$ reference frame, Alice's measurement precedes that of Bob ($t''_A < t''_B$). We see that the imposition of the constraints in Eq.(\ref{eq:two-party-NS}) prevents any causal loop (and resulting ``grandfather paradoxes"), or in other words relativistic causality is guaranteed by the two-party no-signaling principle. Note that the ``freedom-of-choice" condition that Alice and Bob are allowed to choose their measurements freely (using a private random number generator, for example) is necessary in the argument above to identify the choice of measurement by Alice as cause and observation by Bob of his output distribution as effect. Furthermore, remark that it is essential to consider different inertial reference frames to make the argument, no causality violation would occur if faster-than-light propagation occurred in only one reference frame. 

As shown by Eberhard in \cite{ER89}, the no-signaling requirements are satisfied in both non-relativistic quantum theory and in the relativistic quantum field theory. In quantum field theory, this requirement would be ensured by the vanishing of commutators of field operators $O_A$ and $O_B$ representing Alice and Bob's observables respectively, i.e., $[O_A, O_B] = 0$ \cite{ER89}. In the non-relativistic quantum theory that is usually used to analyze Bell experiments, the description of the Alice-Bob composite system by a density operator $\rho_{AB}$ in the tensor product space $\mathcal{H}_{A} \otimes \mathcal{H}_B$, the description of quantum operations by local Kraus operators acting on the respective Hilbert spaces and the partial trace rule ensure that the local statistics only depend on the reduced density matrices of the respective party, and no superluminal communication even using entangled states is possible. 


\subsection{Modification of the three-party no-signaling constraints to the relativistic causality constraints.}

The causal structure of the three-party Bell experiment is shown in Fig. \ref{fig:three-party-caus-struc}. Alice's spacetime random variables corresponding to her input $X$ and output $A$ are generated at spacetime location $(t_A, \textbf{r}_A)$ in inerial reference frame $\mathcal{I}$, similarly Bob's input-output $Y,B$ are generated at $(t_B, \textbf{r}_B)$ and Charlie's input-output $Z, C$ are generated at $(t_C, \textbf{r}_C)$. 
We now investigate the question:
\begin{itemize}
\item \textit{What are the necessary and sufficient conditions in the three-party Bell scenario that ensure that no causal loops appear}? 
\end{itemize} 
We shall see that the answer depends upon the exact spacetime locations of the measurement events in the Bell experiment. 

\begin{prop}
\label{prop:rel-cau-constraints}
Consider the three-party measurement configuration shown in Fig. \ref{fig:three-party-meas-config}, where in an inertial frame $\mathcal{I}$, the spacetime locations of the measurement events $(t_A, \textbf{r}_A)$, $(t_B, \textbf{r}_B)$ and $(t_C, \textbf{r}_C)$ are such that the intersection of the future light cones of $A$ and $C$ is contained within the future light cone of $B$. 
The necessary and sufficient constraints to ensure that no causality violation occurs in this configuration are given by
\begin{eqnarray}
\label{eq:rel-caus-3-party-1}
P_{B,C|Y,Z}(b,c|y,z) = \sum_{a} P_{A,B,C|X,Y,Z}(a,b,c|x,y, z) &=& \sum_{a} P_{A,B,C|X,Y,Z}(a,b,c|x',y,z) \; \; \; \forall x,x',y,z,b,c \nonumber \\
P_{A,B|X,Y}(a,b|x,y) = \sum_{c} P_{A,B,C|X,Y,Z}(a,b,c|x,y,z) &=& \sum_{c} P_{A,B,C|X,Y,Z}(a,b,c|x,y,z') \; \; \; \forall z,z',x,y,a,b \nonumber \\
P_{A|X}(a|x) = \sum_{b,c} P_{A,B,C|X,Y,Z}(a,b,c|x,y,z) &=& \sum_{b,c} P_{A,B,C|X,Y,Z}(a,b,c|x,y',z') \; \; \; \forall y,y',z,z',x,a \nonumber \\
P_{C|Z}(c|z) = \sum_{a,b} P_{A,B,C|X,Y,Z}(a,b,c|x,y,z) &=& \sum_{a,b} P_{A,B,C|X,Y,Z}(a,b,c|x',y',z) \; \; \; \forall x,x',y,y',z,c. \nonumber \\
\end{eqnarray}
\end{prop}
\begin{proof}
The constraints in Eq.(\ref{eq:rel-caus-3-party-1}) guarantee that the marginal distributions $P_{A|X}(a|x)$, $P_{C|Z}(c|z)$, $P_{A,B|X,Y}(a,b|x,y)$ and $P_{B,C|Y,Z}(b,c|y,z)$ are well-defined. Firstly, notice that the fact that the marginal $P_{B|Y}(b|y)$ is also well-defined for all $b,y$ is guaranteed as a consequence of these constraints by the relation
\begin{eqnarray}
\label{eq:three-party-b-marginal}
\sum_{a,c} P_{A,B,C|X,Y,Z}(a,b,c|x,y,z) &=& \sum_{a,c} P_{A,B,C|X,Y,Z}(a,b,c|x',y,z) \nonumber \\
&=& \sum_{a,c} P_{A,B,C|X,Y,Z}(a,b,c|x',y,z') \; \; \; \; \forall x,x',z,z',y,b, \nonumber \\
\end{eqnarray}
where we used the first and second equalities from Eq.(\ref{eq:rel-caus-3-party-1}) successively. The constraint that each party's marginal distribution is well-defined is necessary in order not to violate causality by the two-party result from Prop. \ref{prop:two-party-NS}, i.e., if either Alice's or Charlie's output statistics exhibited a dependence on Bob's input ($P_{A|X,Y}(a|x,y)$ or $P_{C|Z,Y}(c|z,y)$) then by the Prop. \ref{prop:two-party-NS}, Bob could signal his input $Y$ to his causal past and a causal loop would result.

We now move to the two-party distributions. Here, compared to the usual no-signaling constraints in Eq.(\ref{eq:Multi-party-NS}), the constraints in Eq.(\ref{eq:rel-caus-3-party-1}) only ensure that the A-B joint distribution $P_{A,B|X,Y}$ and the B-C joint distribution $P_{B,C|Y,Z}$ are well-defined independent of the remaining party's input. The fact that this is necessary can be seen as follows. In the measurement configuration in Fig. \ref{fig:three-party-meas-config}, the spacetime random variable $\mathcal{C}_{AB}$ corresponding to the correlations between the outcomes of $A$ and $B$ manifests itself at the intersection of the future light cones of $A$ and $B$, and in particular, $\mathcal{C}_{AB}$ can be verified outside the future light cone of $C$. Any dependence of the joint distribution $P_{A,B|X,Y}$ on the input $Z=z$ at $C$ can then result in a causal loop, following an analogous reasoning to the two-party result in Prop. \ref{prop:two-party-NS}. Similarly, we see that it is necessary for $P_{B,C|Y,Z}$ to be well-defined independent of the input $X=x$ at $A$. 

We now move to see the sufficiency of the four constraints in Eq.(\ref{eq:rel-caus-3-party-1}). To see this, we note that as opposed to the usual sufficient no-signaling constraints (\ref{eq:Multi-party-NS}), these constraints do not ensure that the marginal $P_{A,C|X,Z}(a,c|x,z)$ is well-defined. In other words, the joint output distribution of Alice and Charlie can explicitly depend upon Bob's input $Y = y$ as $P_{A,C|X,Y,Z}(a,c|x,y,z)$ even though the marginal distribution of each party does not depend on $y$. This implies that a superluminal influence propagated from Bob's system to the joint system of Alice and Charlie altering their joint distribution while keeping their marginal distributions unaffected. The spacetime random variable $\mathcal{C}_{AC}$ corresponding to the correlations between the outputs of $A$ and $C$ only manifests itself at the intersection of the future light cones of $A$ and $C$. Briefly, when this intersection is contained within the future light cone of $B$, $\mathcal{C}_{AC}$ is timelike separated from $B$ and hence any influence of $\mathcal{C}_{AC}$ by the choice of input at $B$ does not lead to signaling. More formally, the argument is stated as follows. 

We want that the choice of input $Y=y$ at $B$ does not signal to any space-time location $S$ via the change of correlations $\mathcal{C}_{AC}$. Now, two possibilities (attempts to signal to $S$ via $\mathcal{C}_{AC}$) arise. 

\begin{enumerate}
\item Alice at $A$ and Charlie at $C$ transmit their outputs by a light signal at speed $c$ to $S$. In this case, $S$ must be contained in the intersection of the future light cones of $A$ and $C$. Now, the crucial property of the measurement configuration of Fig. \ref{fig:three-party-meas-config} is that the intersection of the future light cones of $A$ and $C$ is contained within the future light cone of $B$. Therefore, $S$ is timelike separated from $B$. Formally, the sum of the time taken for a superluminal influence (at any arbitrary speed $u>c$) to move from $B$ to $A$ and the time taken for a light signal (at speed $c$) to move from $A$ to $S$ is less than the time taken for a light signal at speed $c$ to travel from $B$ to $S$ directly. 
A similar condition holds for the influence traveling via $C$. Mathematically, these constraints are captured by the following equations
\begin{eqnarray}
\tau_{B \xrightarrow{u} A} + \tau_{A \xrightarrow{c} S} < \tau_{B \xrightarrow{c} S} \nonumber \\
\tau_{B \xrightarrow{u} C} + \tau_{C \xrightarrow{c} S} < \tau_{B \xrightarrow{c} S}
\end{eqnarray} 
where $\tau_{B \xrightarrow{u} A}$ denotes the time taken for a signal at speed $u > c$ to travel from $B$ to $A$, $\tau_{A \xrightarrow{c} S}$ denotes the time taken for a light signal (at speed $c$) to travel from $A$ to $S$, and so on. The fact that $S$ is always in the future light cone of $B$ ensures that no superluminal transmission of information takes place. 

\item Alternatively, Alice (or Charlie or both) transmits her output $A$ via a subsequent measurement $X' = A$ and a similar superluminal influence as Bob did, i.e., Alice superluminally influences the measurement events at two spacelike separated locations $D$ and $E$ (with the property that the intersection of the future light cones of $D$ and $E$ is contained within the future light cone of $A$) in such a way as to change the joint distribution $P_{D,E|U,V,X' = A}$ while retaining the marginal distributions $P_{D|U}$ and $P_{E|V}$. If $D$ and $E$ subsequently send a light signal, then as in the previous argument, the output at $A$ is available only in a location $A'$ in the future light cone of $A$. The intersection of the future light cones of $A'$ and $C$ is directly seen to be contained within the future light cone of $A$ and $C$ which is in turn contained within the future light cone of $B$, so that the information about the outputs at $A$ and $C$ is again only available at a timelike separation from $B$. Similarly, if $D$ and $E$ subsequently send a superluminal influence as Bob did, then we repeat the argument to see that any concentration of information only occurs within the future light cone of $B$.      
\end{enumerate}
We therefore deduce that in measurement configurations such as in Fig. \ref{fig:three-party-meas-config}, the constraints in Eq.(\ref{eq:rel-caus-3-party-1}) are necessary and sufficient as opposed to the usual three-party no-signaling constraints from Eq.(\ref{eq:Multi-party-NS}) which while being sufficient are not strictly necessary in order to prevent a violation of causality.

\end{proof}

\begin{Eg}
As a simple illustrative example of the above, consider a box $P_{A,B,C|X,Y,Z}$ with binary inputs and binary outputs for the three parties of the form
\begin{eqnarray}
P_{A,B,C|X,Y,Z}(0,0,0|x,0,z) &=& P_{A,B,C|X,Y,Z}(1,0,1|x,0,z) = \frac{1}{2} \; \; \; \forall x, z \in \{0,1\} \nonumber \\
P_{A,B,C|X,Y,Z}(0,0,0|x,1,z) &=& P_{A,B,C|X,Y,Z}(1,0,1|x,1,z) = \frac{1}{2} \; \; \; (x,z) \in \{(0,0), (0,1), (1,0) \}\nonumber \\
P_{A,B,C|X,Y,Z}(0,0,1|1,1,1) &=& P_{A,B,C|X,Y,Z}(1,0,0|1,1,1) = \frac{1}{2}.
\end{eqnarray}
This box deterministically outputs $b=0$ for both inputs $y = 0,1$ and has marginals $P_{A,C|X,Y,Z}$ of the form 
\[ P_{A,C|X,Y,Z}(a,c|x,y,z) =  \begin{cases} 
     P_{A,C|X,Z}^{\text{l}}(a,c|x,z) & y=0\\
      P_{A,C|X,Z}^{\text{PR}}(a,c|x,z) & y=1.
   \end{cases}
\]
Here $P_{A,C|X,Z}^{\text{l}}(a,c|x,z)$ is a local box that returns correlated outcomes $a=c$ for any input $x,z$ with uniform marginals $P_{A|X}(a|x) = P_{C|Z}(c|z) = \frac{1}{2}$ for $a,c \in \{0,1\}$  (such a local box is generated by shared randomness between $A$ and $C$), and $P_{A,C|X,Z}^{\text{PR}}(a,c|x,z)$ is the Popescu-Rohrlich box \cite{PR} that returns outcomes satisfying $a \oplus c = x \cdot z$ also with uniform marginals $P_{A|X}$ and $P_{C|Z}$. We see that such a box $P_{A,B,C|X,Y,Z}$ explicitly satisfies the relativistic causality constraints from Prop. \ref{prop:rel-cau-constraints} since the local marginals $P_{A|X}$, $P_{B|Y}$ and $P_{C|Z}$ are well-defined (independent of the inputs of the other party), and the two-party marginals $P_{A,B|X,Y}$ and $P_{B,C|Y,Z}$ are given explicitly as
\begin{eqnarray}
P_{A,B|X,Y}(0,0|x,y) &=& P_{A,B|X,Y}(1,0|x,y) = \frac{1}{2} \; \; \; \forall x,y \in \{0,1\}  \nonumber \\ 
P_{B,C|Y,Z}(0,0|y,z) &=& P_{B,C|Y,Z}(0,1|y,z)  = \frac{1}{2} \; \; \;  \forall y,z \in \{0,1\}
\end{eqnarray}
which are also well-defined, independent of the input of the remaining party. 

This relativistically causal box maximally saturates the three-party Mermin-type expression $a \oplus b \oplus c = x \cdot y \cdot z$ (even beyond the customary promise on the input set $(x,y,z)$), since we have for $y=0$ that $a \oplus b \oplus c = a \oplus c = 0$ and for $y = 1$ that $a\oplus b \oplus c = a \oplus c = x \cdot z$. However, despite this maximal violation, notably in this box, the outputs $B$ are always deterministic.  
\end{Eg} 

\begin{center}
\begin{figure}[t]
		\includegraphics[width=0.75\textwidth]{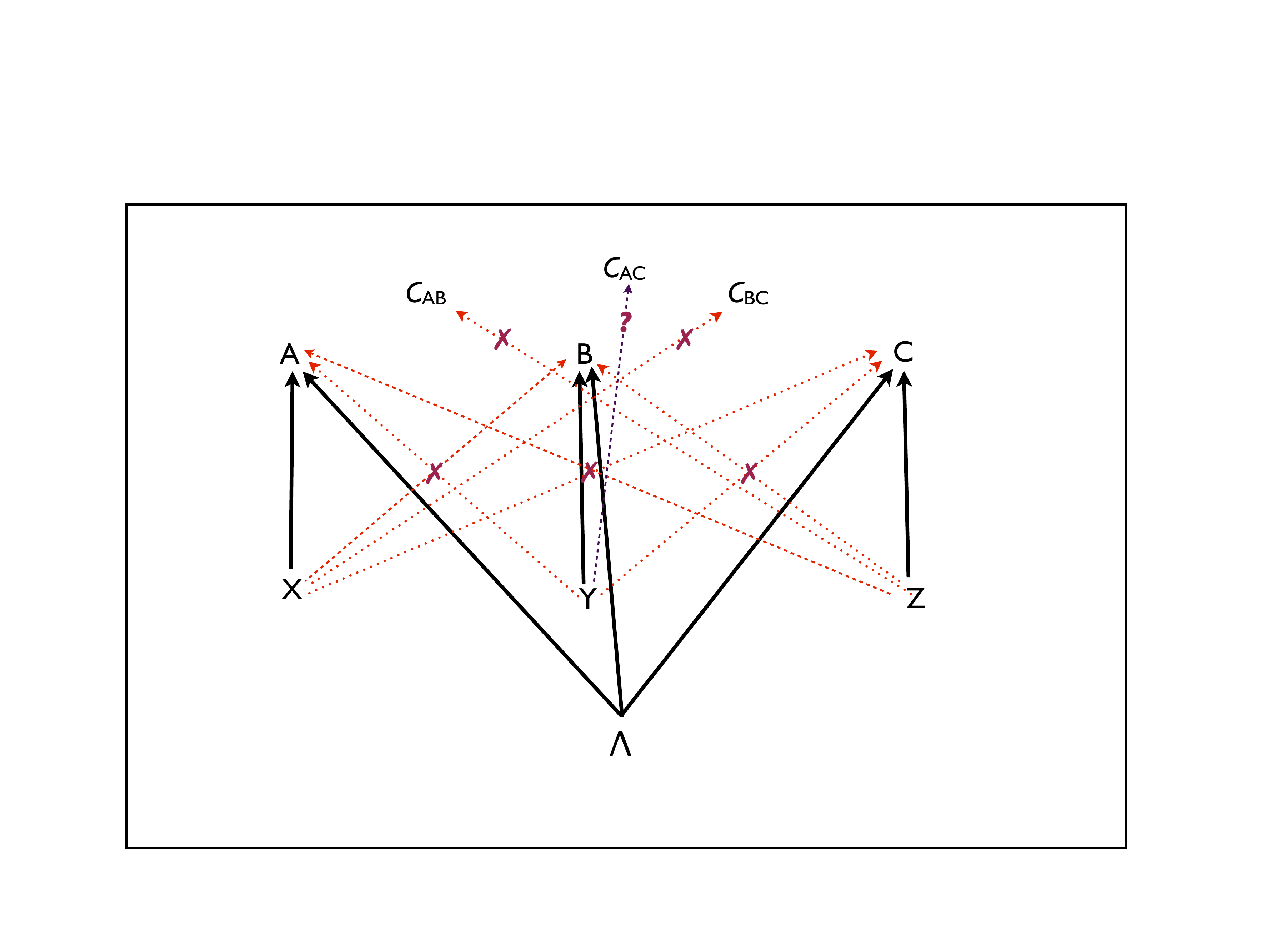}
\caption{The causal structure of the three-party Bell experiment. The outputs $A, B, C$ are correlated via a common $\Lambda$. The inputs $X, Y, Z$ are chosen freely according to a notion of free-choice (see Section \ref{sec:freewill}). The input $X$ cannot signal to change the distribution of a remote party's output $B$ or $C$, analogously for inputs $Y$ and $Z$. On the other hand, depending upon the spacetime configuration of the measurement events, the input may influence the correlations between remote parties' outputs without violating causality. In particular, in the measurement configuration shown in Fig. \ref{fig:three-party-meas-config}, the input $Y$ may influence the correlations $\mathcal{C}_{AC}$.}
	\label{fig:three-party-caus-struc}
\end{figure}
\end{center}

\begin{center}
\begin{figure}[t!]
		\includegraphics[width=0.85\textwidth]{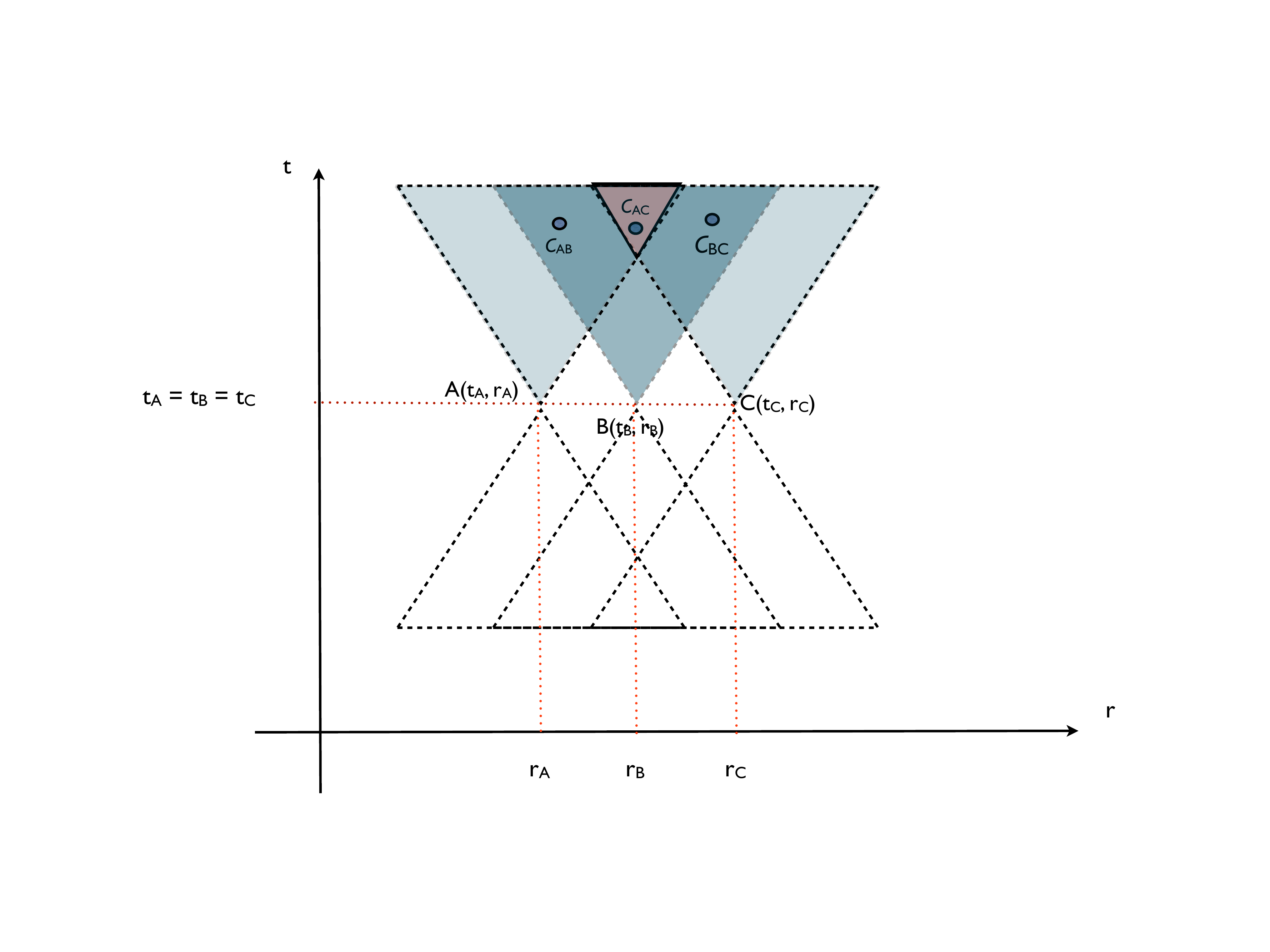}
\caption{A particular spacetime configuration of measurement events in the three-party Bell experiment. The spacetime locations of Alice, Bob and Charlie's measurement events in some inertial reference frame $\mathcal{I}$ labeled by axes $(t,\textbf{r})$ are given by $(t_A, \textbf{r}_A)$, $(t_B, \textbf{r}_B)$ and $(t_C, \textbf{r}_C)$ respectively. The correlations between the outputs of Alice and Bob is denoted by a spacetime random variable $\mathcal{C}_{AB}$ which manifests itself (is checked by Alice and Bob) at a spacetime location within the intersection of the future light cones of Alice and Bob. Similarly, the correlations between the outputs of the other pairs of parties is denoted by $\mathcal{C}_{AC}$ and $\mathcal{C}_{BC}$. The crucial property of this measurement configuration is that the intersection of the future light cones of $A$ and $C$ is contained within the future light cone of $B$.}
	\label{fig:three-party-meas-config}
\end{figure}
\end{center}

Having seen that in certain measurement configurations, the set of constraints in Eq.(\ref{eq:rel-caus-3-party-1}) as opposed to the full set of no-signaling constraints is already necessary and sufficient to ensure no violation of causality, we now identify the spacetime configurations where this occurs. To do this, we consider the situation where the superluminal influence propagates at some fixed speed $u > c$ in an inertial reference frame $\mathcal{I}$. We slightly alter the notation and consider the situation from a cryptographic perspective with Bob $B$ replaced by Eve $E$ with input-output $W, E$ and Alice-Charlie $A-C$ replaced by Alice-Bob $A-B$. Thus, we fix the spacetime coordinates $(t_A, \textbf{r}_A)$ and $(t_B, \textbf{r}_B)$ of Alice and Bob, and investigate from which spacetime location $(t_E, \textbf{r}_E)$ Eve is able to influence the correlations $P_{A,B|X,Y,W}(a,b|x,y,w)$ without affecting the marginal distributions $P_{A|X}$ and $P_{B|Y}$, i.e., in which measurement configurations (\ref{eq:rel-caus-3-party-1}) is necessary and sufficient to prevent the violation of causality. Accordingly, our result is then the following Proposition \ref{prop:spacetime-region} whose proof is in the Supplemental Material.   

\begin{prop}
\label{prop:spacetime-region}
Consider measurement events $A, B$ with corresponding spacetime coordinates $(t_A,\textbf{r}_A)$, $(t_B, \textbf{r}_B)$ in a chosen inertial reference frame $\mathcal{I}$. Then a measurement event $E$ can superluminally influence the correlations between $A$ and $B$ at speed $u > c$ without violating causality in $\textsl{I}$ if and only if its space coordinate $\textbf{r}_E$ satisfies
\begin{equation}
\textbf{r}_E \in Seg(\bigcirc(AB; \varphi_{\alpha}))
\end{equation}  
for any circle $\bigcirc$ with $AB$ as a chord and having angle $\varphi_{\alpha}$ as the angle in the corresponding minor segment, where $\varphi_{\alpha} = \pi - 2 \arcsin(\alpha)$ and $\alpha = c/u$, and if its time coordinate $t_E$ satisfies 
\begin{equation}
t_E \leq \min \left[ t_A - \frac{|\textbf{r}_A - \textbf{r}_E|}{u}, t_B - \frac{|\textbf{r}_B - \textbf{r}_E|}{u} \right]. 
\end{equation} 
\end{prop}

Outside of the spacetime region in Proposition \ref{prop:spacetime-region}, the usual no-signaling constraints from (\ref{eq:Multi-party-NS}) turn out to be necessary and sufficient to ensure causality is preserved. Figure \ref{fig:three-party-caus-struc} shows the causal structure represented by the three-party Bell experiment represented in terms of a Directed Acyclic Graph (DAG) \cite{Pearl09}. We can now consider the implications of the spacetime structure on the causal relationship, namely we can consider the study of Bayesian networks of spacetime random variables (SRVs) in light of the considerations of this section, a study which we pursue in future work. Recall that a Bayesian network, or probabilistic directed acyclic graphical model is a probabilistic graphical model (a type of statistical model) that represents a set of random variables and their conditional dependencies via a directed acyclic graph (DAG). Now, in the  DAG that represents a causal structure in the three-party scenario, an additional ingredient of a new type of edge representing a causal link between the spacetime variable $Y$ and the effects $A, C$ would be added depending on the spacetime location of the measurement events. This edge represents the causal link that $Y$ does not influence the marginal distributions of $A$ and $C$ themselves but changes the joint distribution of $A,C$. The additional ingredient when considering spacetime variables is that the random variable $\mathcal{C}_{AC}$ representing the correlations between $A$ and $C$ is only created in the future light cone of $Y$ provided that the intersection of the future light cones of $A$ and $C$ is contained within the future light cone of $Y$. The causal network corresponding to a particular spacetime arrangement of measurement events would incorporate the additional edges as possible influences that do not lead to superluminal signaling and respect causality.

\subsection{Multi-party Relativistic Causality.}
\label{subsec:multi-party-rel-caus}
We now extend to the general $n$-party scenario the considerations of the previous subsections. In the mutli-party case, different subsets of the usual no-signaling constraints are sufficient to preserve causality, depending on the spacetime locations of the measurement events of the $n$ parties.
Here, we focus on an extension in the simplest $1+1$-dimensional scenario, when the parties are arranged in a line in some reference frame.

Consider $n-1$ parties with respective fixed laboratory space coordinates $\textbf{r}_1, \dots, \textbf{r}_{n-1}$ in frame $I$. Let us identify the possible space-time region $(\textsl{t}_n, \textbf{r}_n)$ from which a party $n$ can influence the correlations between the systems of the $n-1$ parties. Clearly, the time coordinates $\textsl{t}_i$ of the $n-1$ parties at which the influence traveling at $u$ in frame $I$ can be felt are given by 
\begin{equation}
t_i > t_n + \frac{|\textbf{r}_i - \textbf{r}_n|}{u} \; \; \; \; \; \forall i \in [n-1].
\end{equation}
By the results of the previous section, we know that to influence two-party correlations for parties $i$ and $j$ where $i, j \in [n-1]$, $\textbf{r}_n$ must belong to 
\begin{equation}
\textbf{r}_n \in R_{i,j} := Seg(\bigcirc(\textbf{r}_i - \textbf{r}_j; \varphi_{\alpha})). 
\end{equation} 
All points in the intersection of $R_{i,j}$ over all pairs $i,j \in [n-1]$, i.e., $\textbf{r}_n \in  \bigcap_{i,j} R_{i,j}$ satisfy the two conditions: (i) $(t_i, \textbf{r}_i)$ is contained in the future $u$-cone of $(t_n, \textbf{r}_n)$ for all $i \in [n-1]$, and (ii) the intersection of the future light cones of the $n-1$ parties is entirely contained within the future light cone of the $n$-th party. 

In particular, when the $n$ parties are in a line, the relativistic causality constraints (the subset of the usual no-signaling constraints needed to ensure that causality is not violated) are given as follows. Let $S_{m,k}^{n} \subset [n]$ denote a contiguous subset of $[n] = \{1, \dots, n\}$ of size $k$ with initial element $m$, i.e., $S_{m,k}^{n} = \{m, m+1, \dots, m+k-1\}$ for some $1 \leq m \leq n-k+1$. Clearly, the number of such contiguous subsets for fixed $k$ is $n-k+1$, and the total number of contiguous subsets of all possible sizes $k$ is $\sum_{k=1}^{n-1} (n-k+1) = \frac{n^2 + n -2}{2}$. For a string of outputs $\textbf{a}$ of the $n$ parties, let $\textbf{a}_{S_{m,k}^{n}}$ denote the substring of outputs of the parties belonging to the set $S_{m,k}^{n}$, i.e., $\textbf{a}_{S_{m,k}^{n}} = \{\textbf{a}_i\}$ with $i \in S_{m,k}^{n}$ and let $\textbf{a}_{(S_{m,k}^{n})^{c}}$ denote the outputs of the complementary set of parties. Similarly, let $\textbf{x}_{S_{m,k}^{n}}$ denote the inputs of the parties in the set $S_{m,k}^{n}$ and $\textbf{x}_{(S_{m,k}^{n})^{c}}$ denote the inputs of the complementary set. Then the relativistic causality constraints for the $n$ parties stationed in $1$-D are given by
\begin{widetext}
\begin{eqnarray}
\label{eq:rel-caus}
P(\textbf{a}_{S_{m,k}^{n}} | \textbf{x}_{S_{m,k}^{n}} ) = \sum_{\textbf{a'}_{(S_{m,k}^{n})^{c}}} P(\textbf{a'} | \textbf{x'}) = \sum_{\textbf{a''}_{(S_{m,k}^{n})^{c}}} P(\textbf{a''} | \textbf{x''})  \qquad \forall 1 \leq k \leq n-1, 1 \leq m \leq n-k+1
\end{eqnarray}
\end{widetext}
for all $\textbf{a'}, \textbf{a''}$ with $\textbf{a'}_{S_{m,k}^{n}} = \textbf{a''}_{S_{m,k}^{n}} = \textbf{a}_{S_{m,k}^{n}}$ and for all $\textbf{x'}, \textbf{x''}$ with $\textbf{x'}_{S_{m,k}^{n}} = \textbf{x''}_{S_{m,k}^{n}} = \textbf{x}_{S_{m,k}^{n}}$.
In other words, the marginal distribution of the outputs for a contiguous subset of parties' inputs is independent of the complementary set of parties' inputs. On the other hand, for a non-contiguous subsets of parties, the joint probability distribution of their outputs can depend on the inputs of the complementary set, i.e., the parties in between can change the marginal distributions by their choice of inputs while still respecting relativistic causality. 

\begin{center}
\begin{figure}[t!]
		\includegraphics[width=0.85\textwidth]{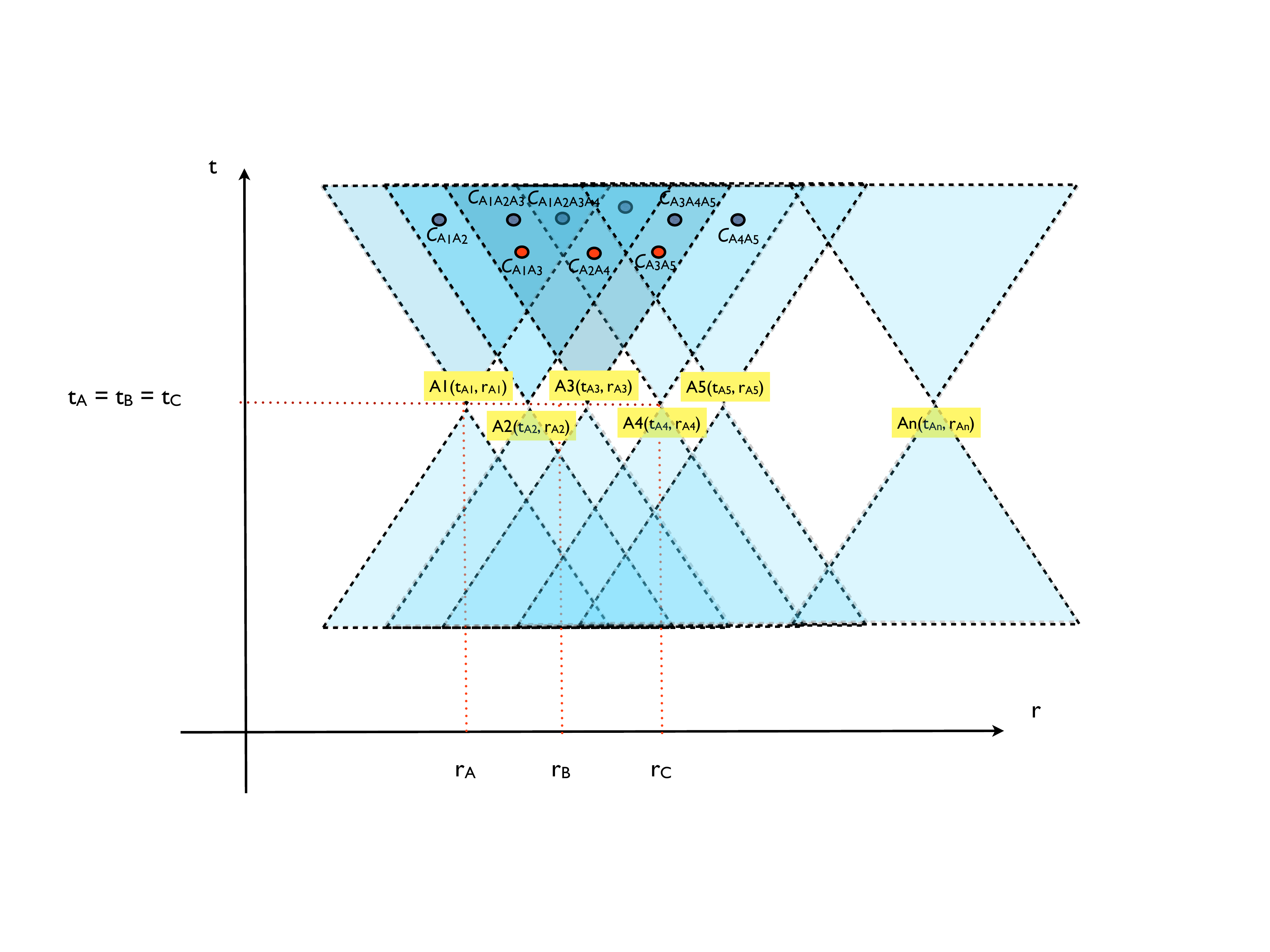}
\caption{A particular spacetime configuration of measurement events in the n-party Bell experiment. The correlations between the observables $\mathcal{C}_{A_{j_1}, \dots, A_{j_k}}$ manifest themselves at the intersection of the future light cones of $A_{j_1}, \dots, A_{j_k}$. For instance, as shown in the Figure, the correlations $\mathcal{C}_{A_1,A_3}$ are seen in the intersection of $A_1$ and $A_3$ which in this measurement configuration lies entirely within the future light cone of $A_2$. As such, any dependence of the joint distribution of outcomes $A_1$ and $A_3$ on the measurement input at $A_2$ does not lead to any causal loops as shown in the Proposition \ref{prop:rel-cau-constraints}.}
	\label{fig:n-party-meas-config}
\end{figure}
\end{center}

The causal structure of the $n$-party Bell experiment is depicted in Fig. \ref{fig:n-party-caus-struc}, where the state of the system is denoted by $\Lambda$, the input and output of the $i$-th party by $X_i, A_i$ for $i \in [n]$.  

\begin{center}
\begin{figure}[t!]
		\includegraphics[width=0.5\textwidth]{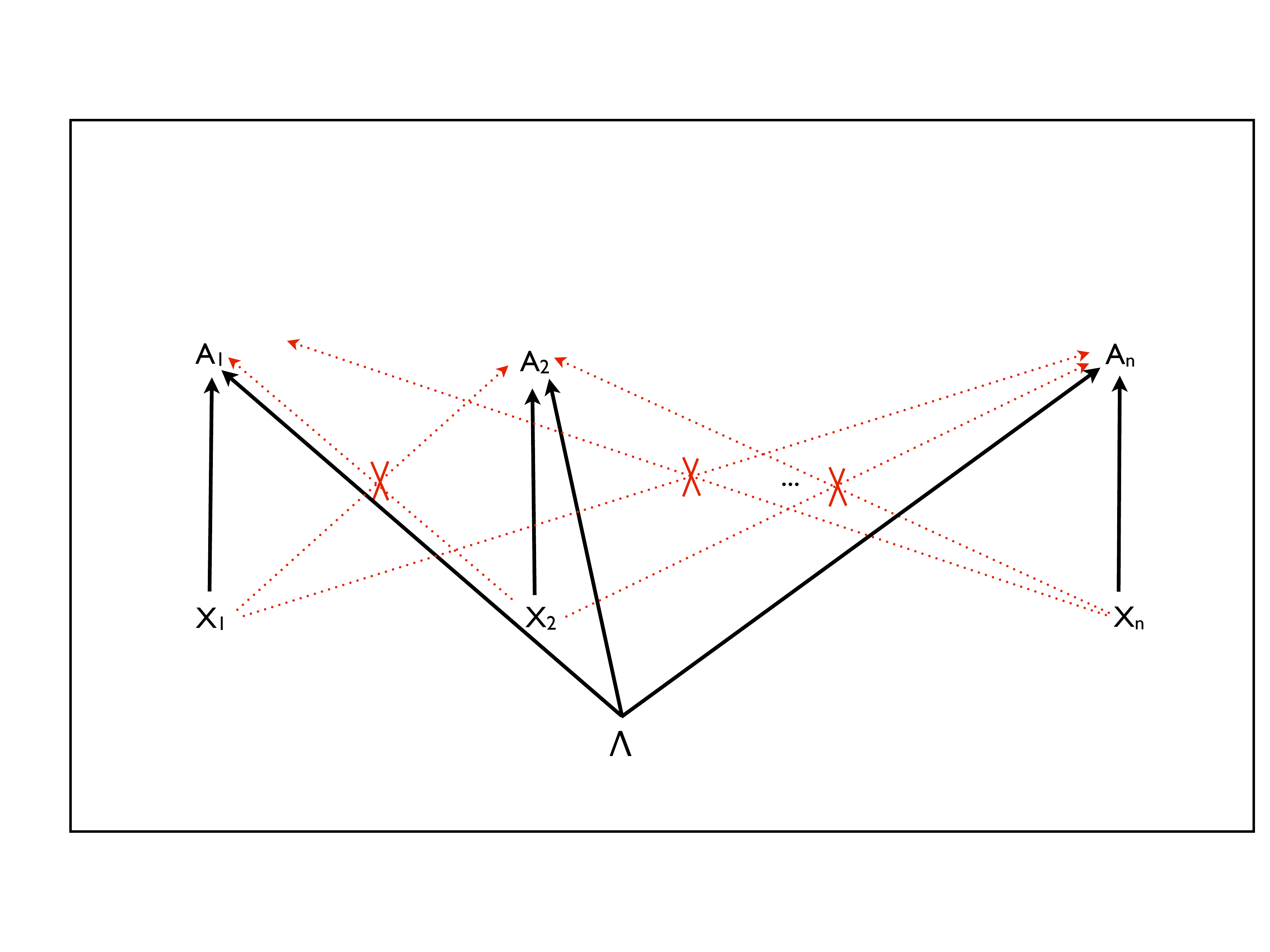}
\caption{The causal structure of the general $n$-party Bell experiment. The outputs $A_i$ for $i \in [n]$ are correlated via a common $\Lambda$, while the inputs $X_i$ are chosen freely according to a notion of free-choice.}
	\label{fig:n-party-caus-struc}
\end{figure}
\end{center}

\section{Relationship to previous works.}
Before we proceed, let us examine the relationship of the present work with previous studies on relativistic causality in the foundations of quantum theory. The paradigmatic example of superluminal explanations for correlations is the de Broglie-Bohm pilot-wave theory \cite{BH93, Hol93}. Here, the wave function of the system is supplemented by the specification of hidden particle positions which evolve according to a \textit{guidance equation}. The velocity of any particle can depend upon the positions of other distant particles, and the theory is manifestly nonlocal. Besides positing a privileged reference frame, the theory also contradicts our usual notion of free choice \cite{CR2} (as for example formalised in \cite{Renner-Colbeck}). In \cite{BPAL+12}, the authors studied the possibility of explaining quantum correlations using influences propagating at a finite speed $v > c$ (denoted as $v$-causal models) defined with respect to a privileged reference frame. It was shown that in the multi-party Bell scenario, if one assumes the usual no-signaling conditions Eq.(\ref{eq:Multi-party-NS}) and the corresponding free-choice conditions, any such $v$-causal model allows correlations which lead to superluminal signaling. 
Recently, \cite{OCB13} looked at a new framework for multi-party correlations without a pre-defined global causal structure (no explicit space-time structure as considered here) but with the validity of quantum mechanics assumed locally. Using \textit{process matrices} (a generalization of the usual density matrices) as fundamental entities, it was shown that there exist process matrices that are compatible with local quantum mechanics which are however causally nonseparable, i.e., incompatible with a global causal order.  This paper, in contrast, follows a suggestion of Shimony \cite{Shim83}, followed by Grunhaus, Popescu and Rohrlich \cite{GPR} of ``jamming" non-local correlations, i.e., the possibility of non-local correlations that go beyond the usual quantum formalism, and yet do not lead to any contradictory causal loops, that therefore might appear in post-quantum theories. 

\section{Lorentz invariance of the theory.}
Recall that the Lorentz symmetry is a fundamental principle of nature stating that the laws of physics stay the same in all inertial reference frames, or more formally, that the laws must be expressed in terms of Lorentz covariant quantities. Now, the hypothetical influences that we considered in the derivation of the relativistic causal constraints propagate at speeds $u > c$ in some chosen inertial reference frame. 
One might wonder if the resulting theory can still be Lorentz invariant. 
Firstly, as we have argued, the superluminal influence does not lead to \textit{superluminal transmissions of information}, which are what relativity theory actually prohibits, since they can lead to causal loops and grandfather paradoxes. Secondly, observe that the conditions imposed by relativistic causality in Proposition \ref{prop:rel-cau-constraints} are manifestly Lorentz covariant, i.e., \textit{if the intersection of the future light cones of $A$ and $C$ is contained within the future light cone of $B$ in one inertial reference frame, then this intersection is contained in the future light cone of $B$ in all inertial reference frames}. Indeed, this is the reason why these constraints do not allow for a violation of causality. 

The speed $u > c$ of the superluminal influence is specified in the chosen inertial frame $\mathcal{I}$, however the Lorentz transformations do not transform into superluminal frames. The fact that special relativity does not prohibit such faster-than-light influences has been known since the seminal papers by Bilaniuk, Deshpande and Sudarshan \cite{BDS62} and Feinberg \cite{Feinberg67}. The standard approach is to not consider transformations into the rest frame of a particle moving at speed $u > c$, so long as all the predictions for the usual subluminal inertial reference frames are Lorentz covariant as considered here. In this respect, the theory resembles the paradigmatic non-local theory, namely the de Broglie-Bohm theory which also only makes Lorentz covariant predictions. However, the de Broglie-Bohm theory allows for instantaneous action-at-a-distance throughout all spacetime and is a non-local deterministic theory. In contrast, a theory incorporating the relativistic causality constraints manifestly allows for a consistent notion of free-choice (and consequently intrinsic randomness) which is detailed in the following Section \ref{sec:freewill}. Finally, we remark that if one wishes to also consider an extended Lorentz transformation into superluminal frames of reference, one might follow the approach of \cite{SS86} who show that by
re-deriving Lorentz-like transformations for the superluminal case (rather than merely substituting
into the usual subluminal transformations) one arrives at a Lorentz factor $\frac{1}{\sqrt{v^2/c^2 - 1}}$ which makes all
proper physical quantities real-valued for superluminal frames. Also note that very recently, Hill and Cox \cite{HC12} have developed an `extended’ version of special relativity that uses such a real-valued Lorentz factor for $v > c$. 

\section{Free-choice.}
\label{sec:freewill}

Recall the freedom-of-choice assumption as stated in Def. \ref{def:freewill}. This notion of free-choice was captured by the requirement that each party's input $X_i$ be uncorrelated with any random variable $Z$ that is not in its causal future (denoted by $X_i \nrightarrow Z$), where $A$ is said to be in the causal future of $B$ if $t_{B} < t_A$ in all inertial reference frames. In other words, when $X_i \nrightarrow Z$, we impose $P_{X_i | Z} = P_{X_i}$. 

Consider the causal structure of the $n$-party Bell experiment in Fig. \ref{fig:n-party-caus-struc}, where the state of the system is denoted by $\Lambda$, the input and output of the $i$-th party by $X_i, A_i$ for $i \in [n]$.
In \cite{CR, CR2}, it is claimed that if for all $i \in [n]$, $X_i$ are free according to Def. \ref{def:freewill}, then $P_{A_1, \dots, A_n|X_1, \dots, X_n}(a_1, \dots, a_n|x_1, \dots, x_n)$ obeys the usual no-signaling constraints in Eq.(\ref{eq:Multi-party-NS}). With $\textbf{X} = \{X_1, \dots, X_n\}$ and $\textbf{A} = \{A_1, \dots, A_n\}$, the formulation of the condition that $X_i$ are free in \cite{CR, CR2} is as follows:
\begin{definition}[Free-choice notion in \cite{CR, CR2}]
\label{lem:CR-freewill}
Consider the $n$-party Bell experiment, where the state of the system is denoted by $\Lambda$, the input and output of the $i$-th party by $X_i, A_i$ for $i \in [n]$. Then $X_i$ is free according to Def. \ref{def:freewill} if the following condition is satisfied:
\begin{eqnarray}
\label{eq:CR-freewill}
P_{X_i | \textbf{X} \setminus X_i, \textbf{A} \setminus A_i, \Lambda}(x_i | \textbf{x} \setminus x_i, \textbf{a} \setminus a_i, \lambda) = P_{X_i}(x_i).
\end{eqnarray}
\end{definition}
But this notion of freedom-of-choice in \cite{CR,CR2} misses the slight subtlety raised in this paper which Def. \ref{def:freewill} allows for. Namely, the correlation between two (or more) output spacetime random variables is a spacetime random variable $\mathcal{C}_{A_j, A_k}$ in its own right, with the associated spacetime location being in the intersection of the future light cones of the two (or more) output spacetime random variables $A_j, A_k$ themselves. Then, a party's input $X_i$ does not need to be independent of the joint distribution (of correlations) of output SRV's if the correlation only manifests itself within the future light cone of $X_i$, even though $X_i$ is required to be independent of the individual marginal distributions of each output SRV themselves. 
Accordingly, one can have a non-local effect precede a cause in the sense that one (or even both) of the outputs $A_j$ or $A_k$ may occur before the $X_i$ was chosen. This has important consequences as we shall see below.

We propose the following equivalent reformulation (different from the formulation in \cite{CR, CR2}) of the free-choice definition \ref{def:freewill} in the multi-party Bell scenario. 
\begin{definition}[Modified notion of free-choice]
\label{lem:mod-freewill}
Consider the $n$-party Bell experiment, where the state of the system is denoted by $\Lambda$, the input and output of the $i$-th party by $X_i, A_i$ for $i \in [n]$. Let $\textbf{A}_{X_i \nrightarrow} = \{A_j \}$ denote a set of outputs $A_j$ such that the correlation SRV $C_{\{A_j\}}$ between \textit{all} the outputs $A_j$ is generated outside the future light cone of $X_i$. Then $X_i$ is free according to Def. \ref{def:freewill} if the following condition is satisfied:
\begin{eqnarray}
\label{eq:mod-freewill}
P_{X_i | \textbf{X} \setminus X_i, \textbf{A}_{X_i \nrightarrow}, \Lambda}(x_i | \textbf{x} \setminus x_i, \textbf{a}_{x_i \nrightarrow}, \lambda) = P_{X_i}(x_i).
\end{eqnarray}
\end{definition}
Let us illustrate the Definition \ref{lem:mod-freewill} by explicitly stating the free-choice conditions in the three-party Bell scenario with causal structure shown in Fig. \ref{fig:three-party-caus-struc} and measurement configuration shown in Fig.\ref{fig:three-party-meas-config}. 

\begin{Eg}
In the three-party Bell scenario with causal structure shown in Fig. \ref{fig:three-party-caus-struc} and measurement configuration shown in Fig.\ref{fig:three-party-meas-config} with the intersection of $A$ and $C$'s future light cones falling within the future light cone of $B$, we have the following free-choice constraints
\begin{eqnarray}
\label{eq:three-party-freewill}
P_{X|Y,Z,B,C,\Lambda}(x|y,z,b,c,\lambda) &=& P_{X}(x), \nonumber \\
P_{Y|X,A, \Lambda}(y|x,a,\lambda) &=& P_{Y}(y), \nonumber \\
P_{Y|Z,C,\Lambda}(y|z,c, \lambda) &=& P_{Y}(y), \nonumber \\
P_{Z|X,Y,A,B,\Lambda}(z|x,y,a,b,\lambda) &=& P_{Z}(z).
\end{eqnarray}
\end{Eg}

We can now deduce a formal condition following \cite{CR, CR2} stating that incorporating the three-party free-choice conditions in Eq.(\ref{eq:three-party-freewill}) in a theory requires that the theory obey the three-party causality constraints in Eq.(\ref{eq:rel-caus-3-party-1}). 

\begin{prop}
\label{prop:caus-eq-freewill}
Consider the causal structure of the $3$-party Bell experiment in Fig. \ref{fig:three-party-caus-struc}, where the state of the system is denoted by $\Lambda$, the input and output of Alice, Bob and Charlie are denoted by $(X, A)$, $(Y, B)$ and $(Z, C)$ respectively. If the input variables $X, Y$ and $Z$ are free according to Definition \ref{def:freewill}, i.e., obey Eq.(\ref{eq:three-party-freewill}), then $P_{A,B,C|X,Y,Z}(a,b,c|x,y,z)$ obeys the relativistic causality constraints in Eq.(\ref{eq:rel-caus-3-party-1}). 
\end{prop}
\begin{proof}
The proof follows by rewriting 
\begin{eqnarray}
P_{A,B|X,Y,Z,\Lambda}(a,b|x,y,z,\lambda) &=& P_{A,B|X,Y,\Lambda}(a,b|x,y,\lambda) \frac{P_{Z|X,Y,A,B,\Lambda}(z|x,y,a,b,\lambda)}{P_{Z|X,Y,\Lambda}(z|x,y,\lambda)} \nonumber \\
&=& P_{A,B|X,Y,\Lambda}(a,b|x,y,\lambda).
\end{eqnarray}
Here, the first equality follows by Bayes' rewriting and the second equality uses the free-choice constraints from Eq.(\ref{eq:three-party-freewill}) that $P_{Z|X,Y,A,B,\Lambda}(z|x,y,a,b,\lambda) = P_{Z|X,Y,\Lambda}(z|x,y,\lambda) = P_{Z}(z)$. Analogously, we can use the free-choice constraints $P_{X|Y,Z,B,C,\Lambda}(x|y,z,b,c,\lambda) = P_{X|Y,Z,\Lambda}(x|y,z,\lambda) = P_{X}(x)$ to show that $P_{B,C|X,Y,Z,\Lambda} = P_{B,C|Y,Z,\Lambda}$. This recovers the two-party marginals in Eq.(\ref{eq:rel-caus-3-party-1}). 

Similarly, the single-party marginals can also be deduced as follows.
\begin{eqnarray}
P_{A|X,Y,Z,\Lambda}(a|x,y,z,\lambda) &=& P_{A|X,\Lambda}(a|x,\lambda) \frac{P_{Y|X,A,\Lambda}(y|a,x,\lambda)}{P_{Y|X,\Lambda}(y|x,\lambda)} \frac{P_{Z|X,Y,A,\Lambda}(z|x,y,a,\lambda)}{P_{Z|X,Y,\Lambda}(z|x,y,\lambda)}  \nonumber \\
&=& P_{A|X,\Lambda}(a|x,\lambda).
\end{eqnarray}
Here again the first equality follows by Bayes' rewriting and the second equality uses the free-choice constraints that $P_{Y|X,A,\Lambda} = P_{Y|X,\Lambda}$ and $P_{Z|X,Y,A,\Lambda} = P_{Z|X,Y,\Lambda}$. Analogously, one can deduce that $P_{C|X,Y,Z,\Lambda} = P_{C|Z,\Lambda}$. The fact that the single-party marginal $P_{B|Y,\Lambda}$ is well-defined follows as in Eq.(\ref{eq:three-party-b-marginal}). 

Note that the distribution $P_{A,C|X,Y,Z,\Lambda}$ does not in general equal $P_{A,C|X,Z,\Lambda}$, with
\begin{eqnarray}
P_{A,C|X,Y,Z,\Lambda}(a,c|x,y,z,\lambda) &=& P_{A,C|X,Z,\Lambda}(a,c|x,z,\lambda) \frac{P_{Y|X,Z,A,C,\Lambda}(y|x,z,a,c,\lambda)}{P_{Y|X,Z,\Lambda}(y|x,z,\lambda)} \nonumber \\
&\neq & P_{A,C|X,Z,\Lambda}(a,c|x,z,\lambda).
\end{eqnarray}
This is because in general $P_{Y|X,Z,A,C,\Lambda} \neq P_{Y}$ since the correlation $\mathcal{C}_{A,C}$ only manifests itself within the future light cone of the measurement event at $B$. As such, the amount by which the free-choice condition is violated $\frac{P_{Y|X,Z,A,C,\Lambda}}{P_{Y|X,Z,\Lambda}}$ is exactly equal to the amount by which the no-signaling constraint is violated $\frac{P_{A,C|X,Y,Z,\Lambda}}{P_{A,C|X,Z,\Lambda}}$.   
\end{proof}

The considerations of the three-party scenario in Prop. \ref{prop:caus-eq-freewill} can also be extended straightforwardly to the $n$-party scenario.
\begin{prop}
Consider the causal structure of the $n$-party Bell experiment in Fig. \ref{fig:n-party-caus-struc} for the measurement configuration in Fig.\ref{fig:n-party-meas-config}, where the state of the system is denoted by $\Lambda$, the input and output of the $i$-th party by $X_i, A_i$ for $i \in [n]$. If, for all $i \in [n]$, $X_i$ are free according to Def. \ref{def:freewill}, then $P_{A_1, \dots, A_n|X_1, \dots, X_n}(a_1, \dots, a_n|x_1, \dots, x_n)$ obeys the relativistic causality constraints in Eq.(\ref{eq:rel-caus}).
\end{prop}

\subsection{Reichenbach's Common Cause Principle.}
In \cite{Reich}, Reichenbach proposed a principle of common cause in an attempt to characterize the asymmetry of time by exploiting a statistical distinction between cause and effect. Reichenbach's principle \cite{Reich} (in the classical setting) states that if two events $A$ and $C$ are statistically correlated, then either $A$ causes $C$ or $C$ causes $A$ or they have a cause operative in their common past, where a classical common cause is a shared random variable $Y$ from which the correlations derive. In other words, we have that $P_{A,C|Y}(a,c|y) = P_{A|Y}(a|y) \times P_{C|Y}(c|y)$, i.e., $A$ and $C$ are independent conditional upon the common cause $Y$ so that this feature could be used to define the distinction between past and future. The notion of common cause can also be extended to the non-classical setting with the shared random variable $Y$ replaced by a quantum state $\rho$. This is akin to the notion of outcome independence from Section \ref{subsec:two-party-Bell-theorem}, where the common cause of the correlations in the two-party Bell experiment is denoted by $\Lambda$. More formally, we have the following definition
\begin{definition}
\label{def:RCCP}
For a classical probability measure space $(\Sigma, P)$, let $A=a$ and $C=c$ be two positively correlated events in $\Sigma$, i.e., $P_{A,C}(a,c) > P_{A}(a) \times P_{C}(c)$. $Y \in \Sigma$ is said to be a common cause of the correlations between $A$ and $C$ if the following conditions hold:
\begin{eqnarray}
P_{A,C|Y}(a,c|y) &=& P_{A|Y}(a|y) P_{C|Y}(c|y) \nonumber \\
P_{A,C|Y}(a,c|y^{\perp}) &=& P_{A|Y}(a|y^{\perp}) P_{C|Y}(c|y^{\perp}) \nonumber \\
P_{A|Y}(a|y) &>& P_{A|Y}(a|y^{\perp}) \nonumber \\
P_{C|Y}(c|y) &>& P_{C|Y}(c|y^{\perp}).
\end{eqnarray}
\end{definition}
The condition can also be extended to include a \textit{system} of cooperating common causes. 
A partition $\{Y_k\}_{k \in K}$ in $\Sigma$ is said to be the (Reichenbachian) common cause system of the correlation between $A$ and $C$ if the 
screening condition $P_{A,C|Y_k} = P_{A|Y_k} \times P_{C|Y_k}$ holds for all $k \in K$. In this case, the common causes $\{Y_{k}\}$ are in the common causal past of $A$ and $C$, i.e., in the intersection of the past light cones of $A$ and $C$. 

Now, consider that $A$ and $C$ are the output SRV's of measurements by Alice and Charlie, and let $Y$ be the measurement input of Bob. As we have seen, relativistic causality allows, in certain spacetime configurations, for $Y$ to influence the correlations between $A$ and $C$ without affecting their marginal distributions. Still in the classical setting, denoting this influence by $\Lambda_{Y}$, we obtain the screening conditions 
\begin{equation}
P_{A,C|X,Z,Y,\Lambda_{Y}}(a,c|x,z,y, \lambda_{y}) = P_{A|X,Z,Y,\Lambda_{Y}}(a|x, y, \lambda_{y}) P_{C|Z,Y, \Lambda_{Y}}(c|z,y, \lambda_{y})  \; \; \forall y, \lambda_{y}.
\end{equation}
with the restriction on the marginals 
\begin{eqnarray}
\sum_{\lambda_{y}} P_{\Lambda_{Y}}(\lambda_{y}) P_{A|X, Y, \Lambda_{Y}}(a|x,y,\lambda_{y}) &=& \sum_{\lambda_{y'}} P_{\Lambda_{Y}}(\lambda_{y'}) P_{A|X,Y, \Lambda_{Y}}(a|x,y',\lambda_{y'}) = P_{A|X}(a|x)\nonumber \\
\sum_{\lambda_{y}} P_{\Lambda_{Y}}(\lambda_{y}) P_{C|Z,Y, \Lambda_{Y}}(c|z,y,\lambda_{y}) &=& \sum_{\lambda_{y'}} P_{\Lambda_{Y}}(\lambda_{y'}) P_{C|Z,Y, \Lambda_{Y}}(c|z,y',\lambda_{y'}) = P_{C|Z}(c|z).
\end{eqnarray}
This extension of Reichenbach's principle from Def. \ref{def:RCCP} allows the SRV's $Y, \Lambda_{Y}$ to lie in a region that is even outside the common causal past of $A$ and $C$ so long as the causality constraints in (\ref{eq:rel-caus-3-party-1}) are obeyed. 


\section{Device-independent cryptography against relativistic adversaries.}
The considerations of the previous sections have important implications for post-quantum cryptographic tasks in the device-independent scenario, where the honest parties are not assumed to know the exact internal workings of their device and the eavesdropper is only constrained by the laws of relativity, see for example \cite{BHK, BCK, BKP}. In particular, two important considerations appear. 
\begin{enumerate}
\item Firstly, note that the set of boxes obeying the relativistic causality considerations forms a larger dimensional polytope than the usual no-signaling polytope. This confers a larger set of attack strategies for an eavesdropper who may prepare boxes for the honest parties from this larger set. As shown in Proposition \ref{prop:RA-rel-caus}, in certain known device-independent protocols for the cryptographic task of randomness amplification \cite{Acin, our}, this larger set of attack strategies can severely compromise the security of the protocol, \text{if} the honest parties were to perform the required Bell test in a spacetime measurement configuration where the superluminal influence can take effect.  

\item Secondly, and crucially, the security of cryptographic protocols (even relying on a two-party Bell test) can be compromised when the measurement event of Eve's system happens in a suitable spacetime location. In particular, as shown in Proposition \ref{prop:qkd-sec-rel-caus}, the property of monogamy of non-local correlations can break down under the relativistic causality constraints so that such an eavesdropper can obtain full information about the output of the honest parties in the protocol.   
\end{enumerate}

Remark that the first type of attack strategy above may be circumvented if the honest parties perform their measurements in a carefully chosen measurement configuration where the usual no-signaling constraints (\ref{eq:Multi-party-NS}) are both necessary and sufficient. In contrast, the second type of attack can only be avoided if certain assumptions are made about the space-time location of the eavesdropper's measurement event, or alternatively if the honest parties' systems are assumed to be sufficiently shielded from all influences, even those respecting causality.

\subsection{Device-independent Randomness Amplification.}
The implications manifest as an adversarial attack strategy in device-independent protocols for randomness generation, amplification and key distribution when the adversary is only subject to the laws of relativity such as considered in \cite{BHK, our, Acin} as opposed to the scenario when the adversary is in addition assumed to obey the laws of quantum theory \cite{VV, Pironio}. The key step in such a device-independent protocol is a test for the violation of a Bell inequality. Multiple observers, who are spatially separated, choose random inputs on their parts of a device and check that the outcomes conform to the violation of some specific Bell inequality. The notion of device-independence then stems from the black-box scenario, the parties do not have access to the inner workings of the device and the security of the protocol is assessed solely based on the input-output statistics of the device, so that the device may have been prepared by the adversary Eve herself.  

Mathematically, the device in any single run of the protocol is described by the set of conditional probability distributions (box) $\mathcal{P} = \{P_{\textbf{A} | \textbf{X}}(\textbf{a} | \textbf{x}) \}$. Here $\textbf{A} = (\textbf{A}^1, \dots, \textbf{A}^n)$ and $\textbf{X} = (\textbf{X}^1, \dots, \textbf{X}^n)$ denote the outputs and inputs of the device for the $n$ honest parties in the protocol. 
The parties check the input-output statistics to assess the value of a Bell expression and make an inference about $\mathcal{I} \cdot \{P_{\textbf{A} | \textbf{X}}(\textbf{a} | \textbf{x}) \}$, where $\mathcal{I}$ is an indicator vector for the Bell expression that picks out the suitable input-output pairs. The proof of security of such protocols relies on the inference drawn from the observed Bell violation regarding the amount of randomness present in the output. This is achieved by checking that for all non-signaling boxes $\{P_{\textbf{A} | \textbf{X}}(\textbf{a}| \textbf{x})\}$ that achieve this Bell value, for some function $f(\textbf{a}): |\mathcal{A}|^n \rightarrow \{0,1\}$ of the outputs for some set of inputs $\textbf{x}$, it happens that $\eta \leq P_{\textbf{A} | \textbf{X}}(f(\textbf{a}) = 0| \textbf{x}) \leq 1 - \eta$, where $0 < \eta < 1$ \cite{Acin, our}. As an example, in the protocol of Barrett, Hardy and Kent \cite{BHK} for quantum key distribution, the observation of correlations between Alice and Bob's outputs as in the Braunstein-Caves-Pearle chained Bell inequality \cite{BC, Pearle} leads to the conclusion that in the limit of a large number of runs, each party's output was necessarily uniformly random. This is because all no-signaling boxes that maximally violate the chained Bell inequality possess the property of uniformly random marginals for each party. Similar considerations also hold for the protocols considered in \cite{Acin, our} where the maximal violation of a multi-party inequality belonging to the GHZ-Mermin family \cite{GHZ}, leads to the conclusion that a particular function of the outputs (the majority of a subset of the parties' output bits in this case) gives rise to randomness.

Now, from the considerations of the previous sections, we see that there is a larger set of boxes that respect relativistic causality than those belonging to the usually considered no-signaling polytope, depending on the spatial locations of the parties. This implies that an adversary who is aware of the spatial locations of the parties involved in the protocol may provide them with a box that satisfies only the relativistic causality conditions stated in this paper.

In particular, when the $n$ parties are in a line, the relativistic causality constraints (the subset of the usual no-signaling constraints needed to ensure that causality is not violated) are given from Section \ref{subsec:multi-party-rel-caus}  as 
\begin{widetext}
\begin{eqnarray}
\label{eq:rel-caus-2}
P(\textbf{a}_{S_{m,k}^{n}} | \textbf{x}_{S_{m,k}^{n}} ) = \sum_{\textbf{a'}_{(S_{m,k}^{n})^{c}}} P(\textbf{a'} | \textbf{x'}) = \sum_{\textbf{a''}_{(S_{m,k}^{n})^{c}}} P(\textbf{a''} | \textbf{x''})  \qquad \forall 1 \leq k \leq n-1, 1 \leq m \leq n-k+1
\end{eqnarray}
\end{widetext}
for all $\textbf{a'}, \textbf{a''}$ with $\textbf{a'}_{S_{m,k}^{n}} = \textbf{a''}_{S_{m,k}^{n}} = \textbf{a}_{S_{m,k}^{n}}$ and for all $\textbf{x'}, \textbf{x''}$ with $\textbf{x'}_{S_{m,k}^{n}} = \textbf{x''}_{S_{m,k}^{n}} = \textbf{x}_{S_{m,k}^{n}}$. Recall that here $S_{m,k}^{n} \subset [n]$ denotes a contiguous subset of $[n] = \{1, \dots, n\}$ of size $k$ with initial element $m$, i.e., $S_{m,k}^{n} = \{m, m+1, \dots, m+k-1\}$ for some $1 \leq m \leq n-k+1$ and for a string of outputs $\textbf{a}$ of the $n$ parties, $\textbf{a}_{S_{m,k}^{n}}$ denotes the substring of outputs of the parties belonging to the set $S_{m,k}^{n}$, i.e., $\textbf{a}_{S_{m,k}^{n}} = \{\textbf{a}_i\}$ with $i \in S_{m,k}^{n}$ and let $\textbf{a}_{(S_{m,k}^{n})^{c}}$ denotes the outputs of the complementary set of parties. 
In other words, the marginal distribution of the outputs for a contiguous subset of parties' inputs is independent of the complementary set of parties' inputs. On the other hand, for a non-contiguous subsets of parties, the joint probability distribution of their outputs can depend on the inputs of the complementary set, i.e., the parties in between can change the marginal distributions by their choice of inputs while still respecting relativistic causality.

The boxes $\mathcal{P}$ are thus only required to obey the restricted set of constraints imposed in Eq.(\ref{eq:rel-caus-2}) as opposed to the usual no-signaling constraints (which posit that the marginal distribution of \textit{every} subset of parties' outputs is independent of the input of the complementary set of parties). The boxes $\mathcal{P}$ under the reduced set of constraints constitute a set of enhanced attack strategies for an eavesdropper assumed to only obey the causality constraints imposed by relativity in a device-independent cryptographic protocol. Note that the set of boxes still forms a polytope in a larger dimensional space (the number of relativistic causality constraints is smaller than the number of usual no-signaling constraints). Furthermore, note that not all the constraints in Eq.(\ref{eq:rel-caus-2}) are independent.

We illustrate this attack strategy with a fairly generic example, namely we will show that this enhanced set of attack strategies allows the adversary to ensure that the parties cannot extract any randomness from the outputs (for inputs appearing in the inequality), from the violation of the quintessential multi-party Bell inequalities, namely the GHZ-Mermin inequalities \cite{GHZ}. The Mermin inequality is set in the Bell scenario when each of $n$ parties (for $n$ odd, $n \geq 3$) measures one of two inputs $x_i \in \{0,1\}$ and obtains one of two outputs $a_i \in \{0,1\}$. The inequality is given as the following set of constraints on the $n$-party correlators $\langle x_1 \dots x_n \rangle$ for $\sum_i x_i = n - 2k$, with $0 \leq k \leq \lfloor n/2 \rfloor$, $k \in \mathbb{Z}$
\begin{eqnarray}
\label{eq:Mermin-corr}
\langle x_1 \dots x_n \rangle = 
(-1)^k,  \; \text{for} \sum_i x_i = n - 2k.
\end{eqnarray}
Here, the $n$-party correlation function $\langle x_1 \dots x_n \rangle$ is defined as
\begin{eqnarray}
\label{eq:n-party-correlation}
\langle x_1 \dots x_n \rangle =  P_{\textbf{A} | \textbf{X}}(\oplus_{i=1}^{n} a_i = 0 | x_1, \dots, x_n) - P_{\textbf{A} | \textbf{X}}(\oplus_{i=1}^{n} a_i = 1 | x_1, \dots, x_n).  
\end{eqnarray}
When $\sum_i x_i$ is even, no constraints are imposed on the corresponding correlator. 

Measurements by each of the $n$ parties of the Pauli $\sigma_y$ and $\sigma_x$ operators (for $x_i = 0, 1$ respectively) on the $n$-qubit GHZ state 
\begin{equation}
|GHZ_n \rangle = \frac{1}{\sqrt{2}} \left( |0 \rangle_1 \otimes \dots \otimes |0 \rangle_n + |1 \rangle_1 \otimes \dots \otimes |1\rangle_n \right)
\end{equation}
result in a violation of the inequality up to its maximum value, i.e., satisfies all the constraints in Eq.(\ref{eq:Mermin-corr}). 

The inputs $\textbf{x}$ for which constraints are imposed on the correlator are said to appear in the inequality, this set of inputs is denoted by $\mathcal{X}_{\text{Merm}} := \{ \textbf{x} | \sum_i x_i = n - 2k, k \in \mathbb{Z} \}$.
In \cite{DTA14}, it was shown that satisfying Eq.(\ref{eq:Mermin-corr}) implies, in the asymptotic setting of an infinite number of parties $n \rightarrow \infty$, that a particular function of the output bits is fully random. In particular, the following function of the outputs $g_{\textbf{x}}(\textbf{a})$ was considered for any input $\textbf{x}$ that appears in the Mermin inequality. 
\begin{equation}
g(\textbf{a}) = \left\{
\begin{array}{ll}
      1 & \sum_{i} a_i = (4k+2), \; \; k \in \mathbf{Z}_{\geq 0} \wedge  \textbf{x} \in \mathcal{X}_{\text{Merm}} \\
      0 & \text{else} \\
   \end{array}
\right.
\end{equation}
As the number of parties $n \rightarrow \infty$ it was shown that $P(g(\textbf{a}) = 1 | \textbf{x}) \rightarrow \frac{1}{2}$, implying that for all boxes satisfying the usual multi-party no-signaling conditions and the Mermin constraints Eq.(\ref{eq:Mermin-corr}), the bit defined by the function $g( \cdot )$ possesses full intrinsic randomness and defines a process where full randomness amplification takes place. We show in the following proposition that this conclusion no longer holds when in place of the usual no-signaling conditions, only the relativistic causality conditions are taken into account. Furthermore, we show that when considering the inputs appearing in the Mermin inequality, \textit{no function} of the output bits possesses \textit{any} randomness for all odd $n \geq 3$. We leave as an open question whether there exists any multi-party Bell inequality with the property of maximum violation such that all the boxes obeying the new relativistic causality conditions admit a hashing function that defines a (partially) random bit.  

\begin{prop}
\label{prop:RA-rel-caus}
Consider the $n$-party GHZ-Mermin Bell inequality, for odd $n \geq 3$.  Suppose that in some inertial reference frame, the $n$ space-like separated parties are arranged in $1$-D, with $r_1 < \dots < r_n$ and perform their measurements simultaneously, i.e., $t_1 = \dots = t_n$. Then for any input $\textbf{x}^*$ appearing in the inequality, i.e., $\textbf{x}^* \in \mathcal{X}^{n}_{\text{Merm}}$, there exists a box $\mathcal{P}_{\textbf{A} | \textbf{X}}$ violating the Mermin inequality maximally and obeying the relativistic causality constraints in Eq.(\ref{eq:rel-caus}), such that no randomness can be extracted from the outputs $\textbf{a}$ of the box under input $\textbf{x}^*$. In other words, we have 
\begin{equation}
\label{eq:det-hash}
\mathcal{P}_{\textbf{A} | \textbf{X}}(\textbf{a}^* | \textbf{x}^*) = 1,
\end{equation}   
for some fixed output bit string $\textbf{a}^*$.
\end{prop}
\begin{proof}
See the Supplemental Material.  
\end{proof}

\subsection{Device-independent key distribution against relativistic eavesdroppers.}
The considerations of the previous subsection also extend to multi-party protocols for device-independent key distribution against relativistic eavesdoppers, namely that the boxes that the honest parties use in such protocols are only subject to the subset of the no-signaling constraints imposed by relativistic causality. Here, we focus on the second important consideration highlighted at the beginning of this section, namely the compromised security of two-party protocols due to the failure of monogamy of non-local correlations under the relativistic causality constraints when the eavesdropper's measurement event is at an appropriate spacetime location. 

In a QKD protocol, we model the initial correlations between Alice, Bob and an eavesdropper Eve by a distribution $P_{A,B,C|X,Y,Z}$.
In the device-independent setting, this distribution is a priori unknown and may have been chosen by the adversary. We consider the situation when the security claims are supposed to hold for any possible initial distribution compatible with the constraints imposed by relativity. Alice and Bob have access to a public authenticated communication channel so that any information sent through this channel is available to Eve.  
In each step of the protocol, Alice and Bob either access correlated data, perform local operations or exchange messages $M$ over the public authenticated channel and in the final step, they generate the keys $K_A, K_B \in \{0,1\}^{N}$. The protocol is said to be secure if the resulting distribution $P^{\text{real}}_{K_A, K_B, M, C|Z}$ is indistinguishable from an ideal one $2^{-N} \delta_{k_A, k_B} P^{\text{real}}_{M,C|Z}$, where $P^{\text{real}}_{M,C|Z}$ is the marginal distribution obtained from $P^{\text{real}}_{K_A, K_B, M, C|Z}$ and $2^{-N} \delta_{k_A, k_B}$ indicates that Alice and Bob's versions of the secret key are identical and uniformly distributed independent of $M, C, Z$. In the notion of universally composable security, the distinguishability of these two distributions, quantified by $D\left( P^{\text{real}}_{K_A, K_B, M, C|Z}, 2^{-N} \delta_{k_A, k_B} P^{\text{real}}_{M,E|Z} \right)$, can be made arbitrarily small for large $N$, so that when the secret key is used in another subsequent cryptographic task, the composed protocol is as secure as if an ideal secret key was used. Here $D(P_{X}, Q_{X})$ denotes the total variation distance $D(P_{X}, Q_{X}) := \frac{1}{2} \sum_{x} \left| P_{X}(x) - Q_{X}(x) \right|$. 

Here, we focus on the key distribution protocols from \cite{BHK, BCK} that were proven secure against eavesdroppers obeying the usual no-signaling constraints and examine what happens when we modify these to the relativistically causal constraints. These protocols are based on the quantum violation of the Braunstein-Caves-Pearle chain inequalities \cite{BC,Pearle}.  
The chained Bell inequalities are a family of two-party correlation Bell inequalities (\textsc{XOR} games) with $m \geq 2$ inputs and two outputs per party, that generalize the well-known CHSH inequality. As usual, labeling the inputs of Alice and Bob by $x, y \in [m]$ respectively, and the outputs by $a, b \in \{0,1\}$, the chain Bell expression $\mathcal{I}^{m, \text{ch}}_{AB}$ is explicitly written as
\begin{equation}
\label{eq:m-chain}
\mathcal{I}^{m, \text{ch}}_{AB}\left(P_{A,B|X,Y} \right) := \left( \sum_{\substack{x,y \in [m] \\ x = y \; \vee \; x = y + 1}} P_{A,B|X,Y}(a \oplus b = 1 | x,y) \right) + P_{A,B|X,Y}(a \oplus b = 0 | 1, m) \geq 1.
\end{equation}
The minimum value of the chain expression in Eq.(\ref{eq:m-chain}) within classical theories is $1$. In general no-signaling theories obeying Eq.(\ref{eq:two-party-NS}), the algebraic minimum of $0$ can be achieved. Measurements on a quantum state can give rise to correlations that achieve the minimum quantum value of $2 m \sin^2\left(\frac{\pi}{4m}\right)$, which approaches the algebraic minimum of $0$ for large $m$. This value is achieved when Alice and Bob share the maximally entangled state
\begin{equation}
| \phi_{+} \rangle = \frac{1}{\sqrt{2}} \left( |00 \rangle + |11 \rangle \right)
\end{equation}
and perform the measurements corresponding to the basis
\begin{eqnarray}
\left\lbrace \cos \left(\frac{\theta_i}{2} \right) | 0 \rangle + \sin \left(\frac{\theta_i}{2} \right) | 1 \rangle, \sin \left(\frac{\theta_i}{2} \right) |0\rangle - \cos \left(\frac{\theta_i}{2} \right) |1\rangle \right\rbrace,
\end{eqnarray}
where
\begin{eqnarray}
\theta_x = \frac{\pi(x-1)}{m},  \; \; \; \forall x \in [m] \nonumber  \\ 
\theta_y = \frac{\pi(2y-1)}{2m} \; \; \; \forall y \in [m]
\end{eqnarray}
denote Alice and Bob's measurement angles. 


The significance of the quantum violation of the chained Bell inequalities for generating secure keys is that, in the limit of large $m$, the correlations that achieve the minimum quantum value become monogamous and uniform. In other words, consider any distribution $P_{A,B,C|X,Y,Z}$ obeying the usual tripartite no-signaling constraints (where $(Z,C)$ denote the input and output of Eve on her system) for which $\mathcal{I}^{m, \text{ch}}_{AB}(P)$ is close to $0$. Then, for any choice of input $Z = z$ by Eve, her output $C$ is virtually uncorrelated with $A$ and $P_{A|X}(a|x)$ is virtually indistinguishable from uniform, for all $a \in \{0,1\}$. This fact that the violation of the chain inequality leads to private randomness and secure key in the outputs of Alice and Bob for large number of inputs $m$ is captured by the following Proposition \ref{prop:chain-sec-ns}. 

\begin{prop}[\cite{BHK, BCK}] 
\label{prop:chain-sec-ns}
For any distribution $\{P_{A,B,C|X,Y,Z} \}$ obeying the usual multi-party no-signaling constraints (\ref{eq:Multi-party-NS}) in which $A$ and $B$ are binary, there holds
\begin{equation}
D\left( P_{A,C|X,Y,Z}(a,c| x,y,z) , P^{u}_{A|X}(a|x) \times P_{C|X,Y,Z}(c|x,y,z) \right) \leq \frac{1}{2} \mathcal{I}^{m, \text{ch}}_{AB}\left(P_{A,B,C|X,Y,Z} \right),
\end{equation}
for all $x,y \in [m]$, and $z$, and
\begin{equation}
D(P_{C|x,y,z,a}, P_{C|z}) \leq \mathcal{I}^{m, \text{ch}}_{AB}\left(P_{A,B,C|X,Y,Z} \right), 
\end{equation}
for all $a \in \{0,1\}$, $x,y \in [m]$ and $z$. Here $\mathcal{I}^{m, \text{ch}}_{AB}\left(P_{A,B,C|X,Y,Z} \right) = \mathcal{I}^{m, \text{ch}}_{AB}\left(P_{A,B|X,Y} \right)$ denotes the violation of the chain inequality by the marginal box $P_{A,B|X,Y}$ and $P^{u}_{A|X}$ denotes the uniform distribution, i.e., $P^{u}_{A|X}(a|x) = \frac{1}{2}$ for all $a \in \{0,1\}$, $x \in [m]$. 
\end{prop}
This is also sometimes equivalently stated as the classical mutual information between the outputs $A$ and $C$ of Alice and Eve being bounded by the violation of the chain inequality
$I(A : C) \leq \frac{1}{2}\mathcal{I}^{m, \text{ch}}_{AB}\left(P_{A,B|X,Y} \right)$. 

We now show that the above Proposition \ref{prop:chain-sec-ns} does not hold for relativistic causal distributions $\{P_{A,B,C|X,Y,Z}\}$ in the measurement configuration of Fig. \ref{fig:three-party-meas-config}, thus implying that the security of QKD protocols such as \cite{BHK, BCK} that are based on this proposition may be compromised when Eve happens to perform her measurement at an appropriate spacetime location. 
\begin{prop}
\label{prop:qkd-sec-rel-caus}
Consider a two-party QKD protocol where Alice and Bob perform a test of the Braunstein-Caves-Pearle chained Bell inequality $\mathcal{I}^{m}_{\text{ch}}$ (\ref{eq:m-chain}) with $m \geq 2$ of inputs per party, and Eve measures a single input $Z = 1$. Suppose that in some inertial reference frame, the three space-like separated parties are arranged in the measurement configuration of Fig. \ref{fig:three-party-meas-config}, i.e., in $1$-D with $t_A = t_B = t_C$ and $\textbf{r}_{A} < \textbf{r}_B < \textbf{r}_C$ with the intersection of $A$ and $C$'s future light cones contained within the future light cone of $B$. 
For fixed $m$, and any chosen input $y$ of Bob, there exists a relativistic causal box $\mathcal{P}_{A,B,C|X,Y,Z}$ such that 
\begin{eqnarray}
\label{eq:BC-rel}
&&\mathcal{I}^{m, \text{ch}}_{AB}\left(\mathcal{P}_{A,B,C|X,Y,Z} \right) = 0, \; \;\text{and} \nonumber \\
&&D\left( P_{B,C|X,Y,Z}(b,c| x,y,z=1) , P^{u}_{B|Y}(b|y) \times P_{C|X,Y,Z}(c|x,y,z=1) \right) = 1 \; \; \forall  x \in [m].
\end{eqnarray}
\end{prop}
\begin{proof}
The proof is by explicit construction of the box $\mathcal{P}_{A,B,C|X,Y,Z}$ that satisfies the conditions in (\ref{eq:BC-rel}). Simply, this box is defined by the entries
\begin{eqnarray}
P_{A,B,C|X,Y,Z}(0,1,1|1,m,1) = P_{A,B,C|X,Y,Z}(1,0,0|1,m,1) &=& \frac{1}{2} \nonumber \\
P_{A,B,C|X,Y,Z}(0,0,0|x,y,1) = P_{A,B,C|X,Y,Z}(1,1,1|x,y,1) &=& \frac{1}{2} \; \; \; \forall (x,y) \neq (1,m).
\end{eqnarray}
By construction, the marginal box $\mathcal{P}_{A,B|X,Y}$ outputs correlated answers for all input pairs $(x,y)$ except for $(1,m)$ for which it outputs anti-correlated answers so that $\mathcal{I}^{m, \text{ch}}_{AB}\left(\mathcal{P}_{A,B,C|X,Y,Z} \right) = 0$. Also, by construction each marginal $\mathcal{P}_{A|X}, \mathcal{P}_{B|Y}, \mathcal{P}_{C|Z}$ is uniform and the marginals $\mathcal{P}_{A,B|X,Y}$ and $\mathcal{P}_{B,C|Y,Z}$ are well-defined, so that the box obeys the relativistic causality constraints for this scenario. Finally, Bob's output is perfectly correlated with the output of Eve, so that her guessing probability of his seemingly random output is $1$ inspite of the fact that the chain inequality is maximally violated. Also note that during the key run, as usual in a QKD protocol \cite{BHK, BCK}, Alice and Bob announce their inputs whence Eve is able to guess Alice's output for any of her inputs also with unit probability. 
\end{proof}

\begin{remark} 
We remark that while the attack in Proposition \ref{prop:qkd-sec-rel-caus} is allowed within the laws of relativity in that it does not lead to any causal loops, in a cryptographic scenario one may additionally assume that the honest parties perform their Bell test within shielded laboratories, where the assumption of shielding (which ensures that no influence, superluminal or otherwise, is transmitted from their systems to Eve's and vice versa) is formally equivalent to imposing the no-signaling constraints on the distributions $P_{A,B,C|X,Y,Z}$. 
\end{remark}

\begin{remark}
Additionally, the honest parties may also vary (randomly) the timing of their measurement events to attempt to ensure that Eve's measurement happens outside the space-time region in Proposition \ref{prop:spacetime-region}.   
\end{remark}

\begin{remark}
We also remark that in such applications, the boxes should be labeled by the set of space-time coordinates of the measurement events $P^{(t_A, \textbf{r}_A), (t_B, \textbf{r}_B), (t_C, \textbf{r}_C)}_{A,B,C|X,Y,Z}$ to deduce the relativistic causality constraints obeyed by the box (in each run of the Bell test).  
\end{remark}

\section{Properties of Causal vs No-Signaling theories.}
The general properties of theories that obey the no-signaling constraints (\ref{eq:Multi-party-NS}) were studied in \cite{Mas06}. A number of properties that were considered to be quintessentially quantum such as monogamy of correlations, no-cloning, randomness, secrecy etc. were found to be present in such theories. As we have seen, in general the set of relativistically causal correlations is a larger set than the set of no-signaling correlations. Paradoxically, just as the relaxation to the Minkowski causality conditions lead to a larger set of attack strategies for an eavesdropper (and a consequent loss in extractable randomness), we find that the paradigmatic property of monogamy of non-local correlations which was hitherto thought to be generic in all theories compatible with relativity may be circumvented in specific spacetime measurement configurations.

\subsection{Monogamy of non-local correlations.}
We now show that the property of monogamy of non-local correlations is significantly weakened under the relativistic causality constraints and the feature of monogamy of correlations violating the CHSH inequality even disappears in certain spacetime configurations. 
Specifically, we consider the three-party Bell scenario where the parties Alice, Bob and Charlie perform two binary outcome measurements $x, y, z \in \{0,1\}$ and obtain outcomes $a, b, c \in \{0,1\}$ respectively. We label the corresponding binary observables of each party by $A_x, B_y$ and $C_z$ respectively. We consider the well-known CHSH inequality \cite{CHSH} between Alice-Bob and Bob-Charlie. The CHSH expression $\langle CHSH \rangle_{AB}$ reads as
\begin{equation}
\langle CHSH \rangle^{(AB)} := \langle A_0 B_0 \rangle + \langle A_0 B_1 \rangle + \langle A_1 B_0 \rangle - \langle A_1 B_1 \rangle,
\end{equation}
where as usual $\langle A_x B_y \rangle = P_{A,B|X,Y}(a \oplus b = 0|x, y) - P_{A,B|X,Y}(a \oplus b = 1|x,y)$. 
In a local hidden variable theory, the value $\langle CHSH \rangle_{AB}$ is bounded by $2$. The well-known Popescu-Rohrlich (PR) box is a two-party no-signaling box that achieves a value of $4$ for the expression. In the three-party scenario, under the usual no-signaling constraints, the non-local correlations exhibit a phenomenon of monogamy that is captured by the relation \cite{Toner}
\begin{equation}
\label{eq:ns-mono}
\langle CHSH \rangle^{(AB)} + \langle CHSH \rangle^{(BC)} \leq 4.
\end{equation}
In other words, when Alice-Bob observe maximum violation of the CHSH inequality, no correlations can occur between the observables of Bob and Charlie. This phenomenon of monogamy of correlations has found application as the underlying feature that is responsible for the security of many device-independent cryptographic protocols. By observing a sufficiently high violation of the Bell inequality, Alice and Bob are able to ensure that their systems are not highly correlated with any system held by a third party such as an eavesdropper. 

We find however that under the relativistically causal constraints that occur when the three parties are arranged in $1$-dimension, the phenomenon of monogamy is considerably weakened in general and in the above mentioned Bell scenario, it completely disappears. This is captured by the following proposition.      

\begin{prop}
\label{prop:rel-caus-mono}
Consider a three-party Bell scenario, with Alice, Bob and Charlie each performing two measurements $x, y, z \in \{0,1\}$ of two outcomes $a, b, c \in \{0,1\}$ respectively. Suppose that in some inertial reference frame, the three space-like separated parties are arranged in $1$-D, with $r_A < r_B < r_C$ and perform their measurements simultaneously, i.e., $t_A = t_B = t_C$. Then, there exists a three-party relativistically causal box $\mathcal{P}_{A,B,C|X,Y,X}(a,b,c|x,y,z)$ such that
\begin{equation}
\label{eq:rel-caus-mono}
\langle CHSH \rangle^{(AB)}_{\mathcal{P}} + \langle CHSH \rangle^{(BC)}_{\mathcal{P}} = 8.
\end{equation}   
\end{prop}
\begin{proof}
See the Supplemental Material. 
\end{proof}

\section{Explaining quantum correlations by finite (superluminal) speed v-causal models.}
\label{sec:v-causal-exp}
In \cite{BPAL+12, BBLG13}, the question was considered whether quantum correlations somehow arise from outside spacetime or can the correlations be explained by causal (even if superluminal) influences propagating continuously in space. The possibility of ``$v$-causal theories" was considered, i.e., theories incorporating a finite superluminal influence propagating at speed $v > c$ in a privileged reference frame. For two-party Bell experiments, the best one can hope to do is to experimentally lower-bound such a $v$ by testing Bell inequality violations between systems in laboratories farther apart and better synchronized. However, in the multi-party scenario, the authors of \cite{BPAL+12, BBLG13} showed that surprisingly the strength of multi-party quantum nonlocality is large enough that any $v$-causal model that attempts to explain the quantum correlations also gives rise to predictions that can lead to superluminal communication. The argument used to reach this conclusion relied on a special kind of Bell inequality, termed ``hidden influence" inequalities. Here, we re-examine this argument in light of the formulation of relativistic causality and free-choice constraints in this paper.


In particular, two measurement configurations were considered: a four-party Bell experiment in \cite{BPAL+12} and a three-party Bell experiment in \cite{BBLG13}. In the four-party Bell experiment \cite{BPAL+12} with spacelike separated parties $A$, $B$, $C$ and $D$, the measurement configuration was as shown in Figure \ref{fig:four-party-v-caus}. Here, in addition to the usual light cones, the $v$-causal model also gives rise to a past and future $v$-cone in a privileged reference frame. In these models, two events that are causally related by a $v$ signal could be correlated while correlations between events that are outside each other's $v$ cones are assumed to be local.   
A particular dichotomy theorem was proven: namely that either the statistics of $BC|AD$ is local or the no-signaling conditions are violated. This was based on the following ``hidden influence" inequality \cite{BPAL+12}
\begin{eqnarray}
\label{eq:hidden-inf-ineq}
\mathcal{I} &=& -3 \langle A_0 \rangle - \langle B_0 \rangle - \langle B_1 \rangle - \langle C_0 \rangle - 3 \langle D_0 \rangle 
- \langle A_1 B_0 \rangle - \langle A_1 B_1 \rangle  + \langle A_0 C_0 \rangle + 2 \langle A_1 C_0 \rangle + \langle A_0 D_0 \rangle \nonumber \\
&&+ \langle B_0 D_1 \rangle 
- \langle B_1 D_1 \rangle - \langle C_0 D_0 \rangle - 2 \langle C_1 D_1 \rangle + \langle A_0 B_0 D_0 \rangle + \langle A_0 B_0 D_1 \rangle + \langle A_0 B_1 D_0 \rangle
- \langle A_0 B_1 D_1 \rangle  \nonumber \\ 
&&  - \langle A_1 B_0 D_0 \rangle - \langle A_1 B_1 D_0 \rangle
 +\langle A_0 C_0 D_0 \rangle + 2 \langle A_1 C_0 D_0 \rangle - 2 \langle A_0 C_1 D_1 \rangle \leq 7,
\end{eqnarray}
which is satisfied by all no-signaling boxes $P_{A,B,C,D|X,Y,Z,W}$ for which the conditional boxes $P_{B,C|Y,Z,A,X,D,W}$ are local. The hidden influence inequality (\ref{eq:hidden-inf-ineq}) is violated by quantum boxes which achieve a value $\mathcal{I}^{\text{qm}} \simeq 7.2$ giving rise to the conclusion that quantum correlations cannot be explained by finite-speed $v$-causal influences. 

Now, by Eq. (\ref{eq:rel-caus}), the set of constraints that are necessary and sufficient to preserve causality in the measurement configuration of Fig. \ref{fig:four-party-v-caus} is given by the condition that the marginal distribution of the outputs for any contiguous subset of parties is well-defined, i.e., that $P_{A_1, A_2, A_3|X_1, X_2, X_3}(a_1, a_2, a_3 | x_1, x_2, x_3)$, $P_{A_2, A_3, A_4| X_2, X_3, X_4}(a_2, a_3, a_4 | x_2, x_3, x_4)$, $P_{A_1, A_2|X_1, X_2}(a_1, a_2 | x_1, x_2)$, $P_{A_2, A_3|X_2, X_3}(a_2, a_3 | x_2, x_3)$, $P_{A_3, A_4|X_3, X_4}(a_3, a_4|x_3, x_4)$, $P_{A_1|X_1}(a_1|x_1)$, $P_{A_2|X_2}(a_2|x_2)$, $P_{A_3|X_3}(a_3|x_3)$ and $P_{A_4|X_4}(a_4|x_4)$ are well-defined. 
\begin{eqnarray}
\label{eq:four-party-rel-caus}
\sum_{a_4} P(a_1, a_2, a_3, a_4 | x_1, x_2, x_3, x_4) &=& \sum_{a_4} P(a_1, a_2, a_3, a_4 | x_1, x_2, x_3, x'_4) \; \; \; \forall  a_1, a_2, a_3, x_1, x_2, x_3, x_4, x'_4 \nonumber \\
\sum_{a_1} P(a_1, a_2, a_3, a_4 | x_1, x_2, x_3, x_4) &=& \sum_{a_1} P(a_1, a_2, a_3, a_4 | x'_1, x_2, x_3, x_4) \; \; \; \forall a_2, a_3, a_4, x_2, x_3, x_4, x_1, x'_1 \nonumber \\
\sum_{a_1, a_2} P(a_1, a_2, a_3, a_4|x_1, x_2, x_3, x_4) &=& \sum_{a_1, a_2} P(a_1, a_2, a_3, a_4|x'_1, x'_2, x_3, x_4) \; \; \; \forall a_3, a_4, x_3, x_4, x_1, x_2, x'_1, x'_2 \nonumber \\
\sum_{a_2, a_3} P(a_1, a_2, a_3, a_4|x_1, x_2, x_3, x_4) &=& \sum_{a_2, a_3} P(a_1, a_2, a_3, a_4 | x_1, x'_2, x'_3, x_4) \; \; \; \forall a_1, a_4, x_1, x_4, x_2, x_3, x'_2, x'_3 \nonumber \\
\sum_{a_3, a_4} P(a_1, a_2, a_3, a_4|x_1, x_2, x_3, x_4) &=& \sum_{a_3, a_4} P(a_1, a_2, a_3, a_4 | x_1, x_2, x'_3, x'_4) \; \;\; \forall a_1, a_2, x_1, x_2, x_3, x_4, x'_3, x'_4 \nonumber \\
\sum_{a_2, a_3, a_4} P(a_1, a_2, a_3, a_4 | x_1, x_2,x_3, x_4) &=& \sum_{a_2, a_3, a_4} P(a_1, a_2, a_3, a_4 | x_1, x'_2, x'_3, x'_4) \; \; \; \forall a_1, x_1, x_2, x_3, x_4, x'_2, x'_3, x'_4 \nonumber \\
\sum_{a_1, a_2, a_3} P(a_1, a_2, a_3, a_4 | x_1, x_2, x_3, x_4) &=& \sum_{a_1, a_2, a_3} P(a_1, a_2, a_3, a_4 | x'_1, x'_2, x'_3, x_4) \; \; \; \forall a_4, x_4, x_1, x_2, x_3, x'_1, x'_2, x'_3 \nonumber \\
\sum_{a_1, a_3, a_4} P(a_1, a_2, a_3, a_4 | x_1, x_2, x_3, x_4) &=& \sum_{a_1, a_3, a_4} P(a_1, a_2, a_3, a_4 | x'_1, x_2, x'_3, x'_4) \; \; \; \forall a_2, x_2, x_1, x_3, x_4, x'_1, x'_3, x'_4 \nonumber \\
\sum_{a_1, a_2, a_4} P(a_1, a_2, a_3, a_4 | x_1, x_2, x_3, x_4) &=& \sum_{a_1, a_2, a_4} P(a_1, a_2, a_3, a_4 | x'_1, x'_2, x_3, x'_4) \; \; \; \forall a_3, x_3, x_1, x_2, x_4, x'_1, x'_2, x'_4. 
\end{eqnarray}
For the particular sequence of measurement events in the Fig. \ref{fig:four-party-v-caus}, with $t_A < t_D < t_B \sim t_C$, the $v$-causal model predicts that the correlations $BC|AD$ of systems $B$ and $C$ conditioned upon the input-output of systems $A$ and $D$ are local. This implies by the result of \cite{BPAL+12} that the usual no-signaling conditions which impose that $P_{A_1, A_2, A_3|X_1, X_2, X_3}(a_1, a_2, a_3|x_1, x_2, x_3)$, $P_{A_1, A_2, A_4|X_1, X_2, X_4}(a_1, a_2, a_4|x_1, x_2, x_4)$, $P_{A_1, A_3, A_4|X_1, X_3, X_4}(a_1, a_3, a_4|x_1, x_3, x_4)$ and $P_{A_2, A_3, A_4| X_2, X_3, X_4}(a_2, a_3, a_4 | x_2, x_3, x_4)$ are well-defined is violated. In \cite{BPAL+12}, the authors argue that it is the $P_{A_1, A_2, A_3|X_1, X_2, X_3}(a_1, a_2, a_3|x_1, x_2, x_3)$ or the $P_{A_2, A_3, A_4| X_2, X_3, X_4}(a_2, a_3, a_4 | x_2, x_3, x_4)$ that is not well-defined. The argument is that since the $ACD$ measurement events happen \textit{outside} the future $v$-cone of $B$ in Fig. \ref{fig:four-party-v-caus}, and $B$ chooses his measurement \textit{freely} at that location (with the notion of free will according to Definition \ref{lem:CR-freewill}), $B$'s input $X_2$ cannot cause (be correlated to) the output distribution of $ACD$. An analogous argument gives that the joint output distribution of $ABD$ cannot depend on $C$'s input $X_3$. 

While the above argument was in a single privileged reference frame, to make the argument for causality violation by the $v$-causal model (modulo the assumptions of \cite{BPAL+12}), one must take into account the following.
\begin{enumerate}
\item A consideration of measurement events in different inertial reference frames is necessary for superluminal influences to lead to causal loops. 
This is for example seen in Figure \ref{fig:two-party-ccl} where the introduction of Charlie and Dave moving with uniform velocity relative to Alice and Bob was necessary to create a closed causal loop and the consequent violation of causality. Superluminal signals within a single reference frame do not lead to any causal paradoxes. Note also that, when  investigating $v$-causal models in different reference frames, one can always find a reference frame in which two space-like separated events are simultaneous, in which case the non-locality seen from that reference frame requires an explanation in terms of influences propagating at infinite speed (in that frame). 

\item An influence propagating faster-than-light in one frame of reference, can travel back in time in another frame. Explicitly, consider a superluminal $v$-speed influence propagating in an inertial reference frame $\mathcal{I}$ from $(t_B, x_B)$ to the event at $(t_D, x_D)$ with $v = \frac{x_D - x_B}{t_D - t_B}$ where $x_D > x_B$ and $t_D > t_B$ so that $D$ happens after $B$ in frame $\mathcal{I}$. Now, consider another reference frame $\mathcal{I}'$ moving at speed $u < c$ relative to $\mathcal{I}$. According to the Lorentz transformations, we have
\begin{eqnarray}
t'_D - t'_B &=& \frac{(t_D - t_B) - \frac{u}{c^2} (x_D - x_B)}{\sqrt{1 - \frac{u^2}{c^2}}} \nonumber \\
&=& (t_D - t_B) \frac{1 - \frac{u v}{c^2}}{\sqrt{1 - \frac{u^2}{c^2}}}.
\end{eqnarray}
Now with $v > c$, we may find a reference frame $\mathcal{I}'$ moving at speed $u < c$ and obeying $u v > c^2$ so that $1 - \frac{uv}{c^2} < 0$. In this case, we get $t'_D < t'_B$ and the event $D$ precedes the event $B$ in reference frame $\mathcal{I}'$. 
In particular, the above argument shows that the fact that the $ACD$ measurement events are outside the future $v$-cone of $B$ in a privileged reference frame does not imply that they will remain outside the future $v$-cone of $B$ in every inertial reference frame, since it is $c$-causality that is preserved by the Lorentz transformation. Therefore, one may find inertial frames in which a superluminal influence propagates backwards in time and influences the $ACD$ distribution. That such an influence would not lead to causal loops in the measurement configuration of Fig. \ref{fig:four-party-v-caus} is guaranteed by the argument of Proposition \ref{prop:rel-cau-constraints} since the intersection of the future $c$-cones of $A$, $C$ and $D$ lies entirely within the future $c$-cone of $B$. Analogously, the $ABD$ correlations may also depend on $C$'s input. Remark however, that in order for the argument of \cite{BPAL+12} to work, it is necessary to assume that the $BC|AD$ correlations are local, despite this fact. 
 
\item In order to incorporate superluminal influences within a theory with the aim to preserve causality, it is necessary to consider the free-choice definition in Def. \ref{lem:mod-freewill}. In particular, since such superluminal influences allow correlations between measurement inputs and joint distributions of sets of measurement outcomes as long as they do not lead to causal paradoxes, the imposition of the free-choice condition in Def. \ref{lem:CR-freewill} is too strong a requirement in a theory with superluminal influences. A modification of the measurement configuration is then required to ensure that the argument stands even with the notion of free-will in Def. \ref{lem:mod-freewill}. As an example, let us impose a requirement that the parties' measurement events be performed at the coordinates given by $(t_A, \textbf{r}_A) = \left(0, (0,0,0) \right)$, $(t_D, \textbf{r}_D) = \left(\frac{1}{v}, (1,1,0) \right)$, $(t_B, \textbf{r}_B) = \left(\frac{1+\sqrt{2}}{v}, (1,0,0) \right)$ and $(t_C, \textbf{r}_C) = \left(\frac{1+\sqrt{2}}{v}, (0,1,0) \right)$. For $v > (1+\sqrt{2}) c$, these coordinates ensure that the usual no-signaling restrictions are both necessary and sufficient to preserve causality, while still ensuring that the $v$-causal order $A < D < (B \sim C)$ is maintained. In other words, one can check that for any subset of the parties, there exists a point in the intersection of their future light cones that is outside the future light cone of the remaining party. An explanation of the quantum violation of the hidden influence inequality in this configuration would then lead to the conclusion (modulo the assumptions above and in \cite{BPAL+12}) that the $v$-causal model results in a causal loop. 
\end{enumerate}

For the three-party measurement configuration of \cite{BBLG13}, an analogous argument as above shows that since the intersection of $A$ and $C$'s future light cones is contained within $B$'s future light cone, by the argument in Proposition \ref{prop:rel-cau-constraints} no causal loops would result if the correlations $\mathcal{C}_{AC}$ depend upon $B$'s input. A suitable modification of the measurement configuration would ensure that the argument is maintained even under the notion of free-choice in (\ref{lem:mod-freewill}). 

\begin{center}
\begin{figure}[t!]
		\includegraphics[width=0.75\textwidth]{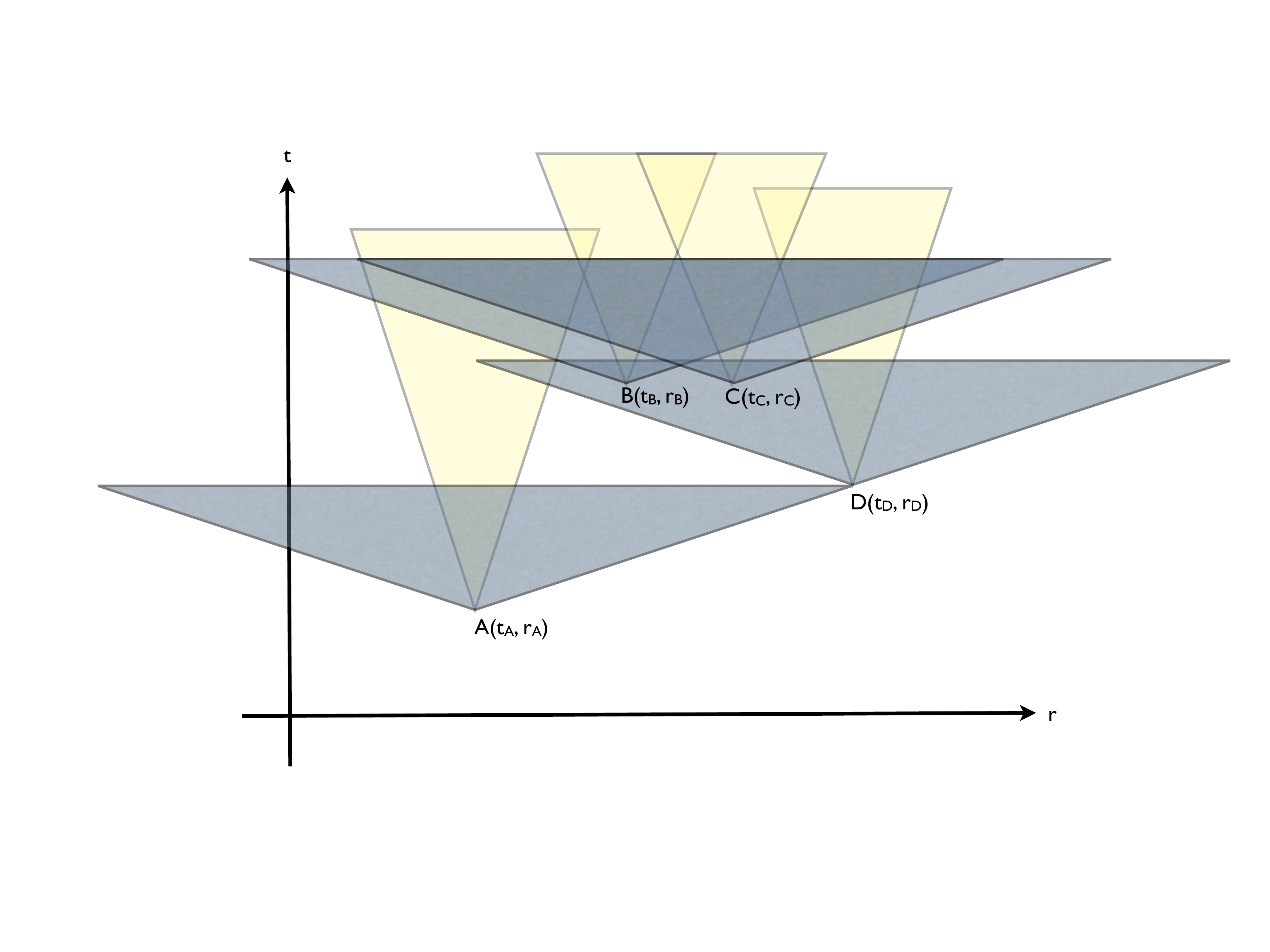}
\caption{The spacetime measurement configuration of the four-party Bell experiment considered in \cite{BPAL+12}. The spacetime coordinates of the measurement events of the four parties are labeled as $(t_A, \textbf{r}_A)$, $(t_B, \textbf{r}_B)$, $(t_C, \textbf{r}_C)$ and $(t_D, \textbf{r}_D)$ respectively. The yellow cones denote the future light cones of the four measurement events in the chosen reference frame while the grey cones denote the $v$-like future cones, for $v > c$. The argument in \cite{BPAL+12} is that due to a randomness in the choice of measurement time and freedom in the choice of measurement setting, either $B$ or $C$ is unable to influence the correlations $ACD$, $ABD$ respectively.}
	\label{fig:four-party-v-caus}
\end{figure}
\end{center}


\section{Multiparty Nonlocality.}
Compared with the scenario of two-party nonlocality where $P_{A,B|X,Y}$ is either local or non-local, in the multi-party scenario, different kinds of nonlocality can be distinguished. In the tripartite scenario, the local correlations are defined as usual as those $P_{A,B,C|X,Y,Z}$ that can be expressed as
\begin{equation}
P_{A,B,C|X,Y,Z}(a,b,c|x,y,z) = \sum_{\lambda} q_{\Lambda}(\lambda) P_{A|X,\Lambda}(a| x, \lambda) P_{B|Y, \Lambda}(b|y, \lambda) P_{C|Z, \Lambda}(c|z, \lambda),
\end{equation}
with $\sum_{\lambda} q_{\Lambda}(\lambda) = 1$. 
While correlations not of the form above are non-local, different kinds of tripartite nonlocality can be distinguished \cite{BBGP13, GWAN12}. 

In pioneering work \cite{Sve87}, Svetlichny introduced the notion of genuine $3$-way nonlocal or Svetlichny nonlocal correlations which are those $P_{A,B,C|X,Y,Z}$ that \textit{cannot} be expressed in the form
\begin{eqnarray}
\label{eq:bilocal}
P_{A,B,C|X,Y,Z}(a,b,c|x,y,z) &=& r_{AB|C} \sum_{\lambda} q_{\Lambda}(\lambda) P_{A,B| X, Y, \Lambda}(a,b|x,y,\lambda) P_{C|Z, \Lambda}(c|z, \lambda) \nonumber \\ &+&  r_{AC|B} \sum_{\gamma} q_{\Gamma}(\gamma) P_{A,C|X,Z, \Gamma}(a,c|x,z, \gamma) P_{B|Y, \Gamma}(b|y, \gamma)  \nonumber \\ &+& r_{BC|A} \sum_{\upsilon} q_{\Upsilon}(\upsilon) P_{B,C|Y, Z, \Upsilon}(b,c|y,z, \upsilon) P_{A|X, \Upsilon}(a|x,\upsilon),  
\end{eqnarray}
with $r_{AB|C}, r_{AC|B}, r_{BC|A} \geq 0$, $r_{AB|C} + r_{AC|B} + r_{BC|A} = 1$ and $\sum_{\lambda} q_{\Lambda}(\lambda) = \sum_{\gamma} q_{\Gamma}(\gamma) = \sum_{\upsilon} q_{\Upsilon}(\upsilon) = 1$. 
Here, the bipartite distributions are \textit{not} required to obey the bipartite no-signaling conditions, so these can be arbitrary signaling correlations. Correlations of the form in Eq.(\ref{eq:bilocal}) are referred to as \textit{bilocal} (BL) or ``$S_2$-local" . Svetlichny introduced an inequality, the violation of which guarantees that the correlations are Svetlichny nonlocal, i.e., not of the form in (\ref{eq:bilocal}). The inequality is expressed in terms of the correlation functions  $\langle A_x B_y C_z \rangle$ with 
\begin{eqnarray}
\label{eq:prob-corrfn-old}
P_{A,B,C|X,Y,Z}(a,b,c|x,y,z) = \frac{1}{8} [1 + (-1)^a \langle A_x \rangle + (-1)^b \langle B_y \rangle + (-1)^c \langle C_z \rangle + (-1)^{a + b} \langle A_x B_y \rangle + \nonumber \\
(-1)^{b+ c} \langle B_y C_z \rangle + (-1)^{a+ c} \langle A_x C_z \rangle + (-1)^{a+b+c} \langle A_x B_y C_z \rangle ].
\end{eqnarray}
In terms of the correlators, Svetlichny's inequality is explicitly written as
\begin{eqnarray}
\label{eq:Svetlichny-ineq}
\mathcal{I}_{\text{Sve}} := \vert \langle A_1 B_1 C_1 \rangle + \langle A_1 B_1 C_2 \rangle + \langle A_2 B_1 C_1 \rangle - \langle A_2 B_1 C_2 \rangle + \langle A_1 B_2 C_1 \rangle - \langle A_1 B_2 C_2 \rangle - \langle A_2 B_2 C_1 \rangle - \langle A_2 B_2 C_2 \rangle \vert \leq 4. 
\end{eqnarray} 
Suitable measurements on $GHZ$ and $W$-class states lead to a violation of the inequality showing that quantum correlations are genuinely tripartite non-local according to this notion. 

Alternatively, when the bipartite distributions in (\ref{eq:bilocal}) obey the two-party no-signaling constraints, the resulting correlations $P_{A,B,C|X,Y,Z}$ are said to be \textit{no-signaling bilocal} (NSBL). In the intermediate situation, when the bipartite distributions allow for signaling in one direction, the resulting correlations are termed \textit{time-ordered bilocal} (TOBL) i.e., when $P_{A,B|X,Y, \Lambda}$ and the other bipartite distributions are of the form $P^{A \rightarrow B}_{A,B|X,Y, \Lambda}$ or $P^{B \rightarrow A}_{A,B|X,Y, \Lambda}$ where these obey only the unidirectional no-signaling restrictions
\begin{eqnarray}
P^{A \rightarrow B}_{A|X, \Lambda}(a|x,\lambda) &=& \sum_{b} P^{A \rightarrow B}_{A,B|X,Y, \Lambda}(a,b|x,y,\lambda) \nonumber \\
P^{B \rightarrow A}_{B|Y, \Lambda}(b|y,\lambda) &=& \sum_{a} P^{B \rightarrow A}_{A,B|X,Y, \Lambda}(a,b|x,y,\lambda). 
\end{eqnarray}   
Clearly, $NSBL \subset TOBL \subset BL$. In formulating an operational framework for nonlocality \cite{GWAN12}, i.e., terming genuine tripartite nonlocality as that which cannot be created by local operations and allowed classical communication when two parties collaborate, the authors of \cite{GWAN12} showed that surprisingly being outside the set $BL$ is sufficient but not necessary for genuine tripartite nonlocality, i.e., there exist correlations $P_{A,B,C|X,Y,Z} \in BL$ which would also be genuinely tripartite nonlocal under an operational consideration. In other words, there are correaltions in $BL$ which under wirings and classical communication prior to the inputs between $B$ and $C$ can result in nonlocality in the cut $A - BC$, so that violation of Svetlichny's inequality is sufficient but not necessary for genuine multiparty nonlocality. A fundamental reason why the seemingly natural definition of Svetlichny (\ref{eq:bilocal}) does not capture the operational notion of genuine multi-way nonlocality is the lack of the no-signaling restrictions on the bipartite distributions in (\ref{eq:bilocal}). Violation of an inequality separating $NSBL$ correlations from the rest by $P_{A,B,C|X,Y,Z}$ indicates that these correlations are genuinely $3$-way NS nonlocal \cite{BBGP13}. Here we review the notion of genuine $3$-way nonlocality in light of the considerations of previous sections. In particular, we introduce a notion of relativistically causal bilocal (RCBL) correlations and propose an inequality the violation of which suggests that the resulting correlations $P_{A,B,C|X,Y,Z}$ are genuinely relativistically causal $3$-way nonlocal. 

\begin{definition}
\label{def:rel-caus-bilocal}
Suppose that $P_{A,B,C|X,Y,Z}(a,b,c|x,y,z)$ can be written in the form
\begin{eqnarray}
P_{A,B,C|X,Y,Z}(a,b,c|x,y,z) &=& r_{AB|C} \sum_{\lambda} q_{\Lambda}(\lambda) P_{A,B|X,Y,\lambda}(a,b|x,y) P_{C|Z,\lambda}(c|z) \nonumber \\ &+& r_{AC|B} \sum_{\gamma} q_{\Gamma}(\gamma) P_{A,C|X,Y,Z, \Gamma}(a,c|x,y,z, \gamma) P_{B|Y, \Gamma}(b|y, \gamma)  \nonumber \\
&+& r_{BC|A} \sum_{\upsilon} q_{\Upsilon}(\upsilon) P_{B,C|Y,Z,\Upsilon}(b,c|y,z,\upsilon) P_{A|X,\Upsilon}(a|x,\upsilon) 
\end{eqnarray}
with $r_{AB|C}, r_{AC|B}, r_{BC|A} \geq 0$, $r_{AB|C} + r_{AC|B} + r_{BC|A} = 1$ and $\sum_{\lambda} q_{\Lambda}(\lambda) = \sum_{\gamma} q_{\Gamma}(\gamma) = \sum_{\upsilon} q_{\Upsilon}(\upsilon) = 1$, 
where the terms obey the relativistic causality constraints Eq.(\ref{eq:rel-caus-3-party-1}), i.e., each of the marginals $P_{A|X}$, $P_{B|Y}$ and $P_{C|Z}$ is well-defined independently of the other parties' inputs while the two-party term $P_{A,C|X,Y,Z, \Gamma}$ exhibits explicit dependence on $Y$. Then the correlations $P_{A,B,C|X,Y,Z}(a,b,c|x,y,z)$ are said to be relativistically causal bi-local (RCBL). Otherwise, we say that they are genuinely tripartite relativistically causal non-local.  
\end{definition}

In the final statements of this section, we will show $RCBL$ correlations that violate the Svetlichny inequality and are hence outside $BL$  ($RCBL \nsubseteq BL$) and also we exhibit bilocal $BL$ correlations that are outside the $RCBL$ set ($BL \nsubseteq RCBL$) showing that these two sets are incomparable and give rise to two distinct notions of genuine tripartite non-locality.  
\begin{lemma} 
We have 
\begin{equation}
NSBL \subseteq RCBL \nsubseteq BL. 
\end{equation}
\end{lemma}
\begin{proof} 
Evidently, $NSBL \subseteq RCBL$ since the relativistic causality constraints Eq.(\ref{eq:rel-caus-3-party-1}) are a subset of the usual no-signaling constraints (\ref{eq:Multi-party-NS}). To see that $RCBL \nsubseteq BL$, we exhibit a box $P_{A,B,C|X,Y,Z} \in RCBL$ that violates the Svetlichny inequality (\ref{eq:Svetlichny-ineq}). This box is simply (and non-uniquely) given by the correlators 
\begin{eqnarray}
\langle A_1 C_1 \rangle_{y=1} &=& \langle A_1 C_2 \rangle_{y=1} = \langle A_2 C_1 \rangle_{y=1} = - \langle A_2 C_2 \rangle_{y=1}  = 1 \nonumber \\
\langle A_1 C_1 \rangle_{y=2} &=& - \langle A_1 C_2 \rangle_{y=2} = - \langle A_2 C_1 \rangle_{y=2} = - \langle A_2 C_2 \rangle_{y=1}  = 1 \nonumber \\
\langle B_1 \rangle &=& \langle B_2 \rangle = 1. 
\end{eqnarray} 
As an explicit example, we consider the box
\begin{eqnarray}
P_{A,B,C|X,Y,Z}(1,1,1|x,y,z) &=& P_{A,B,C|X,Y,Z}(2,1,2|x,y,z) = \frac{1}{2} \; \;\; (x,y,z) \in \{(1,1,1), (1,1,2), (1,2,1), (2,1,1)\} \nonumber \\
P_{A,B,C|X,Y,Z}(1,1,2|x,y,z) &=& P_{A,B,C|X,Y,Z}(2,1,1|x,y,z) = \frac{1}{2} \; \; \; (x,y,z) \in \{(1,2,2), (2,1,2), (2,2,1), (2,2,2)\}
\end{eqnarray}
This achieves the algebraic value of $8$ for the Svetlichny Bell expression (\ref{eq:Svetlichny-ineq}), maximally violating the inequality. Evidently, any inequality based only on the three-body correlators $\langle A_x B_y C_z \rangle$ cannot be used to separate $RCBL$ from genuine tripartite causal non-local boxes.   
\end{proof}

As such, an interesting question is to find a separating hyperplane to separate the set of $RCBL$ from the set $Q$ of quantum correlations. While a detailed investigation of this is carried out in future work \cite{our3}, here we introduce an inequality that accomplishes this in the tripartite scenario. Let us first write the probablilities $P_{A,B,C|X,Y,Z}$ in terms of the correlation functions. The difference here is that the correlations $\langle A_{x} C_{z} \rangle_{y}$ exhibit a dependence on Bob's measurement input $Y=y$. We thus have  
\begin{eqnarray}
\label{eq:prob-corrfn-new}
P_{A,B,C|X,Y,Z}(a,b,c|x,y,z) = \frac{1}{8} [1 + (-1)^a \langle A_x \rangle + (-1)^b \langle B_y \rangle + (-1)^c \langle C_z \rangle + (-1)^{a + b} \langle A_x B_y \rangle + \nonumber \\
(-1)^{b+ c} \langle B_y C_z \rangle + (-1)^{a+ c} \langle A_x C_z \rangle_{y} + (-1)^{a+b+c} \langle A_x B_y C_z \rangle ],
\end{eqnarray}
with the correlator $\langle A_x C_z \rangle_{y}$ exhibiting explicit dependence on $y$ in contrast to Eq.(\ref{eq:prob-corrfn-old}). 
\begin{lemma}
\label{lem:RCBL-ineq}
Consider the three-party Bell scenario, with each party performing one of two dichotomic measurements. In measurement configurations such as Fig.\ref{fig:three-party-meas-config}, the following inequality holds for all relativistically causal bi-local boxes $P_{A,B,C|X,Y,Z} \in RCBL$.  
\begin{eqnarray}
\label{eq:RCBL-ineq}
\mathcal{I}_{\text{RCBL}} := 2\langle A_1 B_1 \rangle + \langle A_1 C_1 \rangle_{y=1} + \langle A_1 C_1 \rangle_{y=2} +2 \langle B_1 C_2 \rangle -2 \langle A_2 B_2 C_1 \rangle +2 \langle A_2 B_2 C_2 \rangle \leq 6.
\end{eqnarray}
Suitably chosen measurements on the $GHZ$ state $|GHZ \rangle = \frac{1}{\sqrt{2}} \left( |000 \rangle + |111 \rangle \right)$ attain a value $2(1+2 \sqrt{2}) \approx 7.657$ violating the inequality, showing that quantum correlations are genuinely tripartite relativistically causal non-local. Furthermore, these quantum correlations belong to the set $BL$, hence $BL \nsubseteq RCBL$. 
\end{lemma}

\begin{proof}
One can directly check that the inequality (\ref{eq:RCBL-ineq}) is valid for all relativistically causal bilocal (RCBL) correlations by linear programming (recall that the set $RCBL$ is a convex polytope). The measurements on $|GHZ \rangle$ are given by
\begin{eqnarray}
A_1 &=& \sigma_z, \quad A_2 = \sigma_x, \nonumber \\
B_1 &=& \sigma_z, \quad B_2 = \sigma_x, \nonumber \\
C_1 &=& \frac{1}{\sqrt{2}} \left(\sigma_z - \sigma_x \right), \; \; C_2 = \frac{1}{\sqrt{2}} \left(\sigma_z + \sigma_x \right).
\end{eqnarray}
These measurements on the $GHZ$ state give the value $\mathcal{I}^{q}_{\text{RCBL}} \geq 2(1+2\sqrt{2})$ violating the inequality showing that quantum correlations are genuinely tripartite relativistically causal nonlocal. Interestingly, these measurements on the $GHZ$ state do not lead to a violation of Svetlichny's original inequality (\ref{eq:Svetlichny-ineq}), and one can in fact find an explicit decomposition of the resulting box in the form (\ref{eq:bilocal}), showing that $BL \nsubseteq RCBL$. 

\end{proof}

\section{Discussion.}
A central point of departure in deriving the relativistic causality constraints as opposed to the usual no-signaling constraints is that in the former, we also account for influences that propagate faster than light as long as they do not lead to causal loops. We discuss in this section why this is not any cause for alarm. Firstly, remark that resonance-induced superluminal \textit{group} velocities (with values up to $300$ times $c$) have been known for a while \cite{Brillouin}, these group velocities do not create any problem of principle, since the \textit{signal} velocity still has $c$ as an upper bound (as in this paper). The group velocity of a wave may exceed the speed of light and can even become negative, but in such cases, no energy or information actually travels faster than $c$. Experiments showing group velocities greater than $c$ include that of Wang et al. \cite{WKD00}, who produced a laser pulse in atomic cesium gas with a group velocity index of $n_g = -310 (\pm 5)$. The observed superluminal propagation of the wave is not at odds with causality, and is instead a consequence of classical interference between its constituent frequency components in a region of anomalous dispersion \cite{WKD00}.

Secondly, as mentioned before, the justification that special relativity does not necessarily exclude faster-than-light propagation has also been known for a while \cite{BDS62, Feinberg67}.
A main objection to the existence of these particles was that they would lead to violations of causality and as such must be unphysical \cite{Terletskii68}, which as we have seen can be overcome by the considerations of this paper. 

Thirdly, recall that the notion of point-to-region signaling considered in this paper is formulated as the nonlocal influence of a field rather than of a superluminal particle. 
Evidence in favor of a superluminal field emerges from the fact that spin-zero FTL particles (satisfying the Klein-Gordon equation) \textit{cannot} be localized as pointlike particles, contrary to all other known fundamental constituents of matter \cite{BFKL71}. A brief overview of the argument is as follows, for details see \cite{BFKL71, FKL69}. Consider the wavefunction $\psi(x,t)$ of a free tachyon of spin zero obeying the Klein-Gordon equation in $1+1$ dimensions
\begin{equation}
\left(\frac{\partial^2}{\partial t^2} - \frac{\partial^2}{\partial x^2} + m_0^2 \right) \psi(x,t) = 0,
\end{equation} 
with solutions that are superpositions of 
\begin{equation}
\psi(x,t) =  e^{-i(E t - p x)}, 
\end{equation}
with the energy-momentum relation $E^2 - p^2 = m_0^2$. 
For the particle moving with velocity $v > c$, we have that $m_0$ is  imaginary and $m_0^2$ is negative, so that we can consider the solutions with $|p| > |E|$ as well as those with $|p| < |E|$. If the initial conditions for the wave equation, $\psi(x,t)$ and $\frac{\partial \psi(x,t)}{\partial t}$ at time $t=0$ are zero outside some interval $[-X, X]$, then they remain outside the interval $[-X - t, X + t]$ for any $t > 0$ by the standard property of the wave equation. In other words, the support of the evolved field configuration lies inside the light-cone of the initial data. So that localized (particle-like) point-to-point influences do not move faster than the speed of light in vacuum, as considered in this paper. On the other hand, if we ignore the solution with $|p| < |E|$ (which has energy $E$ imaginary), then it becomes impossible to solve the equation for initial conditions that are zero outside some interval $[-X,X]$ \cite{BFKL71} so that we do not have a localized influence but only a delocalized one which is then allowed to be superluminal, consistent with the considerations of this paper.   

Finally, note that in Algebraic Quantum Field Theory, as mentioned earlier, to every bounded open subset $\mathcal{O}$ of spacetime, one can associate observables $\mathcal{U}(\mathcal{O})$ and observables in spacelike separated regions $\mathcal{O}, \mathcal{O'}$ are required to commute, i.e., for all $A \in \mathcal{U}(\mathcal{O})$ and $B \in \mathcal{U}(\mathcal{O'})$ we have $[A, B] = 0$ thus ensuring that the causal constraints are satisfied. In linear field theories, there is a complete duality between particles and fields, via the well-known Fock space construction. However as Haag pointed out \cite{Haag}, the one-to-one correspondence between particles and fields holds only in the asymptotic sense, i.e, as time goes to $\pm \infty$ (in the Heisenberg picture), the limit fields obey the free field equation and canonical commutation relations so as to describe free particles; so that it may well be that a quantum field theory does not describe any localized particles at all \cite{Haag}. 

\section{Conclusions.}

In this paper, we have examined the no-signaling constraints from the point of view of causality and proposed a modification to these constraints to a subset that ensures preservation of causality depending on the spacetime configuration of measurement events. As a consequence, we find that boxes in the post-quantum scenario, where one only restricts to the constraints imposed by relativistic causality, should be labeled by the space-time locations of the parties performing the measurement. The correlations that the parties observe can then exhibit various dependencies while still being consistent with relativistic causality. 

Several questions are open. 
One important question is to formulate protocols (under suitable assumptions) for cryptographic tasks such as randomness amplification, expansion and key distribution against an adversary only limited by the laws of relativity (from creating closed causal loops). To do this, an intermediate task is to investigate the existence of multi-party Bell inequalities which allow for certification of randomness against such a relativistically causal adversary. A geometric question is to characterize the subset of multi-party relativistically causal constraints and the space-time regions where they are necessary and sufficient. The notion of genuine multi-party nonlocality as evidenced by the violation of a relativistically causal bilocal ($RCBL$) inequality calls for further investigation \cite{our3}, in particular in light of studies such as \cite{Fritz}. A pressing open question is to investigate whether years after special relativity was formulated, such superluminal influences can be tested experimentally, and what new effects they could lead to. Both non-relativistic quantum theory and relativistic quantum field theory \cite{ER89} are well-known to obey a no superluminal signaling condition, proposals to modify quantum theory by introducing non-linearities have been shown to lead to signaling \cite{Wei89, Cza91, Gis90}. An important open question is to investigate feasible mechanisms for the 
point-to-region superluminal influences within the paradigm of non-linear modifications to quantum theory.

\textit{Acknowledgments.-}
We acknowledge stimulating discussions with Andrzej Grudka and Jakub Rembieli\'{n}ski. R.R. acknowledges useful discussions with Piotr Mironowicz, Karol Horodecki and Roberto Salazar and thanks Stefano Pironio for valuable comments on an earlier version of the manuscript. This work was made possible through the support of a grant from the John Templeton Foundation. The opinions expressed in this publication are those of the authors and do not necessarily reflect the views of the John Templeton Foundation. It was also supported by the ERC AdG grant QOLAPS. R.R. acknowledges support from the grant "Causality in quantum theory: foundations and applications" funded by La Fondation Philippe Wiener - Maurice Anspach.


\textbf{Supplemental Material.}
Here, we present the formal proofs of the propositions stated in the main text. 

\textbf{Proposition 5}. \textit{Consider measurement events $A, B$ with corresponding space coordinates $\textbf{r}_A, \textbf{r}_B$ in a chosen inertial reference frame $\mathcal{I}$. Then a measurement event $E$ can superluminally influence the correlations between $A$ and $B$ at speed $u > c$ without violating causality in $\textsl{I}$ if and only if its space coordinate $\textbf{r}_E$ satisfies
\begin{equation}
\textbf{r}_E \in Seg(\bigcirc(AB; \varphi_{\alpha}))
\end{equation}  
for any circle $\bigcirc$ with $AB$ as a chord and having angle $\varphi_{\alpha}$ as the angle in the corresponding minor segment, where $\varphi_{\alpha} = \pi - 2 \arcsin(\alpha)$ and $\alpha = c/u$, and if its time coordinate $t_E$ satisfies 
\begin{equation}
t_E \leq \min \left[ t_A - \frac{|\textbf{r}_A - \textbf{r}_E|}{u}, t_B - \frac{|\textbf{r}_B - \textbf{r}_E|}{u} \right]. 
\end{equation}} 
\begin{proof}
We observe that in order for Eve to signal to a space-time point $S$ via her influence of the correlations at $A$ and $B$ the following constraints must hold. The sum of the time taken for an influence (at speed $u>c$) to move from $E$ to $A$ and the time taken for a signal (at speed $c$) to move from $A$ to $S$ must be less than the time taken for a signal at speed $c$ to travel from $E$ to $S$ directly. To elaborate, $E$ may superluminally influence the measurement at $A$ which subsequently signals to the event $S$ or alternatively, $E$ may signal to $S$ (signaling being by definition at speed $c$).
A similar condition must also hold for the influence traveling via point $B$. 

Let us examine the first constraint captured by the following equation
\begin{equation}
\tau_{E \xrightarrow{u} A} + \tau_{A \xrightarrow{c} S} < \tau_{E \xrightarrow{c} S}
\end{equation} 
and derive a geometric criterion for when it can hold.

\begin{center}
\begin{figure}[t!]
		\includegraphics[width=0.6\textwidth]{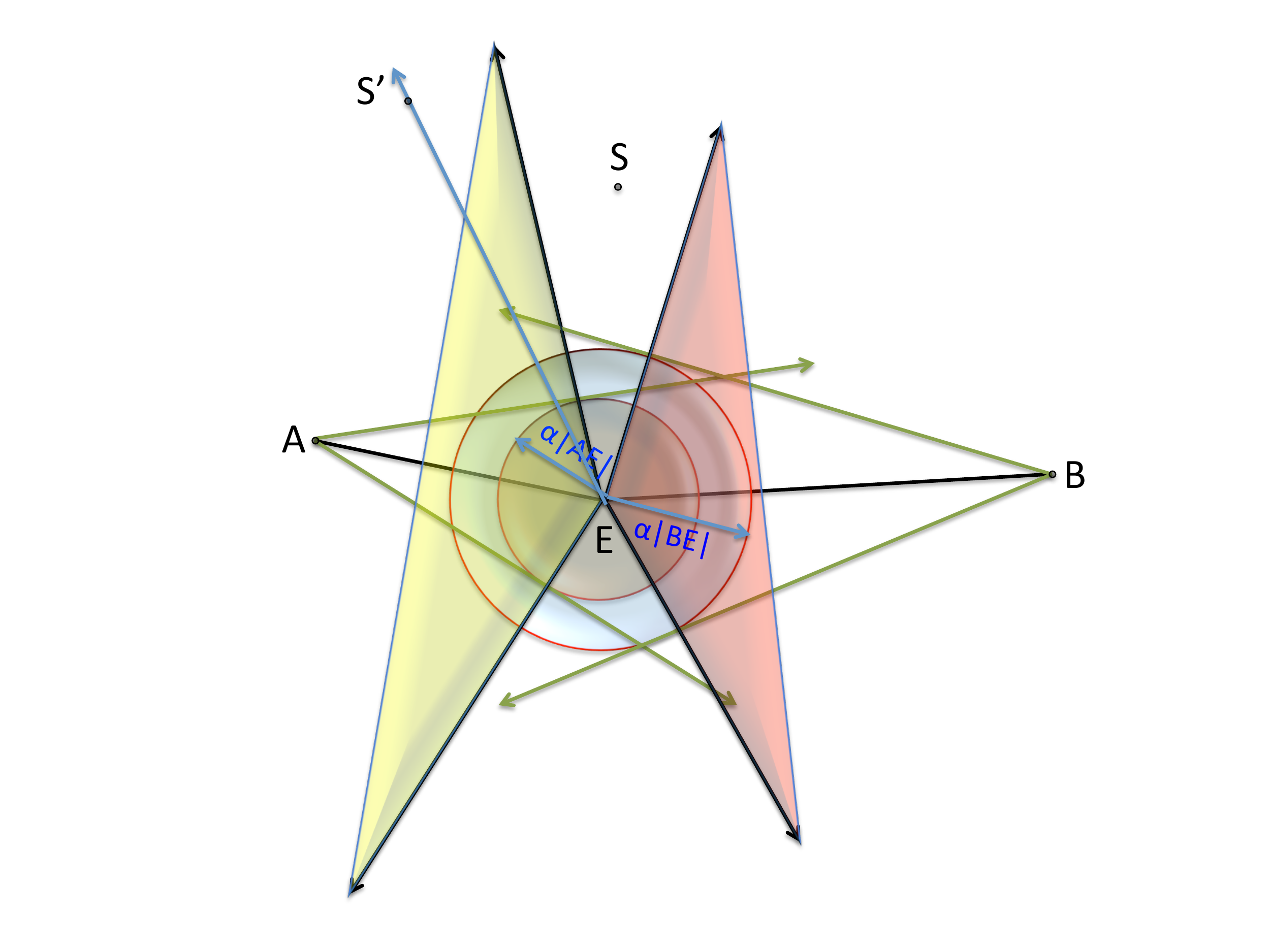}
\caption{Illustration of the constraint on the spacetime locations of the measurement events from Lemma \ref{lem:dual-signal}. $E$ superluminally influences the correlations between $A$ and $B$ while ensuring that the marginal distributions are unaffected. In order to ensure that $E$ does not send a superluminal signal to a spacetime location $S$ via the influence of the correlations between $AB$, we impose the constraints $\tau_{E \xrightarrow{u} A} + \tau_{A \xrightarrow{c} S} < \tau_{E \xrightarrow{c} S}$ and $\tau_{E \xrightarrow{u} B} + \tau_{B \xrightarrow{c} S} < \tau_{E \xrightarrow{c} S}$. Lemma \ref{lem:dual-signal} shows that for fixed locations $\textbf{r}_A$ and $\textbf{r}_B$, no superluminal signaling takes place iff the cones shown in the figure $\textsl{C}(E; T(A, \textbf{B}(E; \;  \alpha |EA|)))$ and $\textsl{C}(E; T(B, \textbf{B}(E; \;  \alpha |EB|)))$ either share a common tangent plane or intersect only at $E$.}
	\label{fig:signaling-region-cond}
\end{figure}
\end{center}

\begin{lemma}
\label{lem:single-signal}
Given measurement events $E, S$ and $A$ with corresponding space-time co-ordinates $\textsc{x}_E$, $\textsc{x}_S$ and $\textsc{x}_A$ in a chosen inertial reference frame $\textsl{I}$ and an influence propagating from $E$ to $A$ at speed $u > c$, $E$ superluminally signals to $S$ via $A$ if and only if the space coordinates satisfy
\begin{equation}
\label{eq:single-signal}
\texttt{B}(E; \; \alpha |EA|) \cap \texttt{B}(S; \; |AS|) = \varnothing.
\end{equation}
where $\textbf{B}(K; \; r)$ denotes a sphere with center at $K$ and radius $r$, $|EA| = |\textbf{r}_A - \textbf{r}_E|$, $|AS| = |\textbf{r}_S - \textbf{r}_A|$ and $\alpha = \frac{c}{u} < 1$.  
\end{lemma}

\begin{proof}
We have that $E$ influences the measurement event $A$ at speed $u > c$. There are thus two ways for $E$ to signal to $S$, 
\begin{enumerate}
\item $E \xrightarrow{u} A \xrightarrow{c} S$
\item $E \xrightarrow{c} S$.
\end{enumerate}
Our task is to show that for any $S$ with space coordinate $\textbf{r}_S$, the first signal does not reach $S$ before the second. 
Consider a ball $\texttt{B}(E; \; \alpha |EA|)$ of radius $\alpha |EA|$ with centre $E$, where $|EA| := | \textbf{r}_E - \textbf{r}_A|$ corresponds to the distance between $E$ and $A$ in the reference frame $\textsl{I}$. We also consider the ball $\texttt{B}(S; \; |AS|)$ with centre $S$ and radius $|AS| := |\textbf{r}_S - \textbf{r}_A|$.

The constraint we impose is  
\begin{equation}
\tau_{E \xrightarrow{u} A} + \tau_{A \xrightarrow{c} S} < \tau_{E \xrightarrow{c} S}
\end{equation} 
This is equivalent to 
\begin{equation}
(\textsl{t}_A - \textsl{t}_E) + \frac{|AS|}{c} < \frac{|ES|}{c}.
\end{equation}
Now recall that $(\textsl{t}_A - \textsl{t}_E) = \frac{|EA|}{u}$ so that the above condition may be rewritten as
\begin{equation}
\alpha |EA| + |AS| < |ES|.
\end{equation}
This in turn is equivalent to the criterion that the spheres $\texttt{B}(E; \; \alpha |EA|)$ and  $\texttt{B}(S; \; |SA|)$ do not intersect, 
\begin{equation}
\texttt{B}(E; \; \alpha |EA|) \cap \texttt{B}(S; \; |SA|) = \varnothing.
\end{equation}
In other words, $E$ is able to superluminally signal to a point $S$ via the superluminal influence on $A$ if and only if the two spheres so constructed do not intersect.  
\end{proof}

We now proceed to analyze the question whether $E$ is able to superluminally signal to any point $S$ via her superluminal influence of the correlations $A-B$ and a subsequent signal (at speed $c$) propagating from both $A$ and $B$ to $S$. Let $\textbf{r}_A$ and $\textbf{r}_B$ be the fixed laboratory locations. We propose a geometric condition identifying for given  
$\textbf{r}_A, \textbf{r}_B$ the positions $\textbf{r}_E$ for which \textit{no} $S$ exists such that $E$ superluminally signals to $S$ via the above mentioned route. 

Let $T(M, \textbf{B}(K; \; r))$ denote the tangent cone to the sphere $\textbf{B}(K; \; r)$ emanating from point $M$. For a given tangent cone to a sphere $\textbf{B}(K; \; r)$, let $\textsl{C}(K; T(M, \textbf{B}(K; \; r)))$ denote the cone with apex point $K$ and passing through the points of intersection of the sphere and its tangent cone, i.e., through $\textbf{B}(K; \; r) \cap T(M, \textbf{B}(K; \; r))$.
\begin{lemma}
\label{lem:dual-signal}
For fixed $\textbf{r}_A$ and $\textbf{r}_B$, there does not exist any $S$ such that $E$ superluminally signals to $S$ via $A$ and $B$ if and only if the cones $\textsl{C}(E; T(A, \textbf{B}(E; \;  \alpha |EA|)))$ and $\textsl{C}(E; T(B, \textbf{B}(E; \;  \alpha |EB|)))$ either share a common tangent plane or intersect only at $E$. 
\end{lemma}
\begin{proof}
From Lemma \ref{lem:single-signal}, for a fixed inertial reference frame, we know that for any $\textbf{r}_E, \textbf{r}_A$, $E$ is able to superluminally signal to any point $S$ with corresponding space coordinate $\textbf{r}_S$ via $E$'s superluminal influence of the event $A$ if and only if $\texttt{B}(E; \; \alpha |EA|) \cap \texttt{B}(S; \; |SA|) = \varnothing$. 

In order for $E$ to superluminally signal to $S$ via her influence of the correlations between $A$ and $B$, a signal (at speed $c$) must reach $S$ from \textit{both} $A$ and $B$ before the signal from $E$ at speed $c$ reaches $S$. Now, in order for the condition in Eq. (\ref{eq:single-signal}) to hold for the signal from $A$, we show that $S$ must obey 
\begin{equation}
\label{eq:signal-A}
 S \in \text{int} \left[ \textsl{C}(E; T(A, \textbf{B}(E; \;  \alpha |EA|)))\right].
\end{equation}
Points in space that are outside this cone do not perceive any superluminal signal from $E$ in that for these points $\tau_{E \xrightarrow{c} S} \leq \tau_{E \xrightarrow{u} A} + \tau_{A \xrightarrow{c} S}$. On the other hand, we have the following.
\begin{claim}
\label{cl:sign-point}
On every ray from the apex $E$ through the interior of the cone, one can find a point $P$ such that 
\begin{equation}
\tau_{E \xrightarrow{u} A} + \tau_{A \xrightarrow{c} P} \leq \tau_{E \xrightarrow{c} P}.
\end{equation} 
\end{claim}
\begin{proof}
Consider a far enough point $P$ on the ray emanating from $E$ and passing through the interior of the cone such that the sphere $\textbf{B}(P; |PA|)$ does not intersect $\textbf{B}(E; \alpha |EA|)$. Let $F$ denote the point at which the ray intersects the sphere $\textbf{B}(P; |PA|)$. By construction, we have  $\measuredangle PFA < \frac{\pi}{2}$ so that such a sphere can always be drawn and we see that $F$ does not lie on $\textbf{B}(E; \alpha |EA|)$. 
\end{proof}

Similarly, for the signal from $B$ to reach $S$ before the one from $E$, we must have 
\begin{equation}
\label{eq:signal-B}
 S \in \text{int} \left[\textsl{C}(E; T(B, \textbf{B}(E; \;  \alpha |EB|)))\right].
\end{equation}
Therefore, in order to exclude every point in space from a superluminal influence from $E$, one must ensure that either $\textsl{C}(E; T(A, \textbf{B}(E; \;  \alpha |EA|))) \cap \textsl{C}(E; T(B, \textbf{B}(E; \;  \alpha |EB|))) = E$ or that the two cones only intersect at a common tangent plane. In this case, no point $S$ different from $E$ exists such that $E$ superluminally signals to $S$ via her influence of $A$ and $B$.

Alternatively, if the intersection $\textsl{C}(E; T(A, \textbf{B}(E; \;  \alpha |EA|))) \cap \textsl{C}(E; T(B, \textbf{B}(E; \;  \alpha |EB|)))$ is a three-dimensional region, then by Claim \ref{cl:sign-point} one can explicitly find an $S$ in this region obeying both Eqs. (\ref{eq:signal-A}) and (\ref{eq:signal-B}), so that superluminal signaling to such an $S$ occurs.
\end{proof}

The constraint in Lemma.(\ref{lem:dual-signal}) is illustrated in Fig \ref{fig:signaling-region-cond}. 
Having Lemma \ref{lem:dual-signal} we can proceed to identify for given $\textbf{r}_A, \textbf{r}_B$ the region of allowable $\textbf{r}_E$ such that no causality condition is violated (in the chosen reference frame). Given the points $A, B$ let $\bigcirc(AB; \theta)$ denote a circle containing $AB$ as a chord and $\theta$ as the angle in the minor segment of the circle. Note that there are an infinite number of such circles, one for each plane containing the line segment $AB$. Let us also denote by $Seg(\bigcirc(AB; \theta))$ the corresponding minor segment of the circle. We now show that the possible space coordinates $\textbf{r}_E$ lie within the shaded region (in green) in Fig. \ref{fig:aposteriori-ns-region}.  

We have seen in Lemma \ref{lem:dual-signal} that in order for $E$ to influence the correlations between $A$ and $B$ at speed $u$ without violating causality, we must have that the cones $\textsl{C}(E; T(A, \textbf{B}(E; \;  \alpha |EA|)))$ and $\textsl{C}(E; T(B, \textbf{B}(E; \;  \alpha |EB|)))$ at most share a common tangent plane, in particular that the intersection of their interiors is empty. 

Let us first consider the case that $\textsl{C}(E; T(A, \textbf{B}(E; \;  \alpha |EA|)))$ and $\textsl{C}(E; T(B, \textbf{B}(E; \;  \alpha |EB|)))$ share a common tangent plane passing through the points $A_1, B_1$ where $A_1$ is a point of intersection of the sphere $\textbf{B}(E; \;  \alpha |EA|)$ and its tangent plane from $A$, $T(A, \textbf{B}(E; \;  \alpha |EA|))$ and similarly $B_1$ is a point of intersection of the sphere $\textbf{B}(E; \;  \alpha |EB|)$ and $T(B, \textbf{B}(E; \;  \alpha |EB|))$, such that $E, A_1, B_1$ lie on one line. By the property of the tangent we have that $\measuredangle E A_1 A = \measuredangle E B_1 B = \frac{\pi}{2}$. Also by definition we have that $\frac{|EA_1|}{|AE|} = \frac{|E B_1|}{|BE|} = \alpha = c/u$ so that $\measuredangle E A A_1 = \measuredangle E B B_1 = \arcsin{(\alpha)}$. Consequently, we obtain that any point $E$ for which the cones $\textsl{C}(E; T(A, \textbf{B}(E; \;  \alpha |EA|)))$ and $\textsl{C}(E; T(B, \textbf{B}(E; \;  \alpha |EB|)))$ share a common tangent plane has space coordinate in the chosen inertial reference frame satisfying
\begin{equation}
\measuredangle A E B = \pi - 2 \arcsin(\alpha) = \varphi_{\alpha}.
\end{equation}
By the fact that angles on the same segment of a circle are equal, we have that all such points lie on a minor segment of a circle (the segment is minor since $\pi - 2 \arcsin(\alpha) \geq \frac{\pi}{2}$ for $\alpha < 1$) having $AB$ as a chord and with $\varphi_{\alpha}$ as the corresponding angle in the segment. This shows the statement for points on the boundary of the region $Seg(\bigcirc(AB; \varphi_{\alpha}))$. For all such points the cones $\textsl{C}(E; T(A, \textbf{B}(E; \;  \alpha |EA|)))$ and $\textsl{C}(E; T(B, \textbf{B}(E; \;  \alpha |EB|)))$ share a common tangent plane. 

Now consider the points that lie within this region. For any such point, we have that $\measuredangle AEB > \pi - 2 \arcsin(\alpha)$. By construction we know that $\measuredangle E A A_1 = \measuredangle E B B_1 = \arcsin{(\alpha)}$ so that by the perpendicularity of the tangent with the radius of the sphere, we have $\measuredangle A E A_1 = \measuredangle B E B_1 = \pi/2 - \arcsin(\alpha)$. 
Therefore, for points that lie within the region, the cones $\textsl{C}(E; T(A, \textbf{B}(E; \;  \alpha |EA|)))$ and $\textsl{C}(E; T(B, \textbf{B}(E; \;  \alpha |EB|)))$ intersect only at $E$, this is the case illustrated in Fig. \ref{fig:signaling-region-cond}. 

For points outside this region we readily see that $\textsl{C}(E; T(A, \textbf{B}(E; \;  \alpha |EA|)))$ and $\textsl{C}(E; T(B, \textbf{B}(E; \;  \alpha |EB|)))$ neither intersect only at $E$ nor share only a common tangent plane. We see by Lemma \ref{lem:dual-signal} that such points violate causality and therefore cannot be the space coordinates of $E$ in the chosen reference frame.
\end{proof}

\begin{center}
\begin{figure}[t!]
		\includegraphics[width=0.45\textwidth]{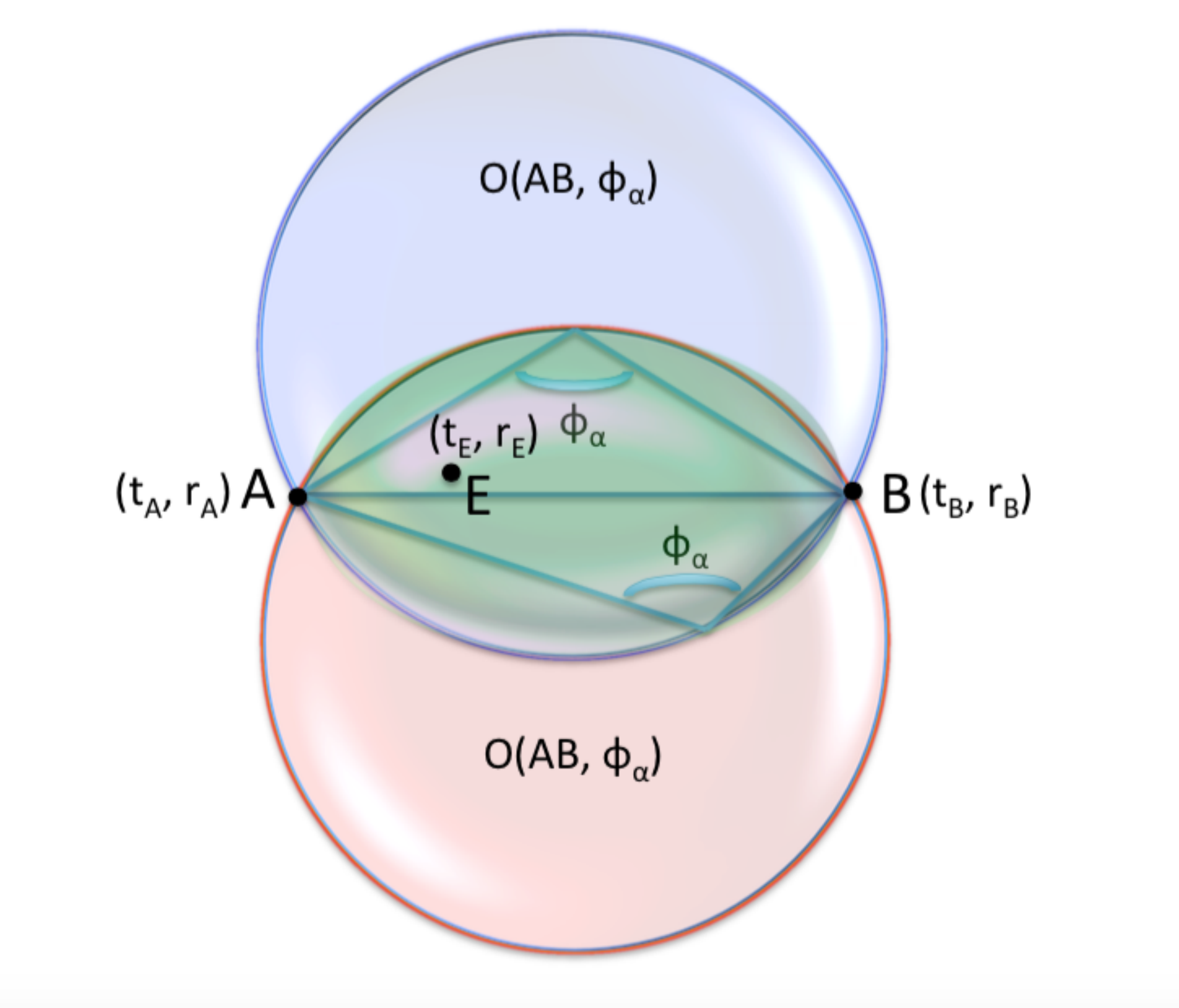}
\caption{The spacetime region, for fixed $\textbf{r}_A$, $\textbf{r}_B$ from which $E$ is able to superluminally influence the correlations $AB$ without changing the marginals while preserving causality. The figure illustrates the algebraic condition found in Prop. \ref{prop:spacetime-region}.}
	\label{fig:aposteriori-ns-region}
\end{figure}
\end{center}

\textbf{Proposition 12.} \textit{Consider the $n$-party GHZ-Mermin Bell inequality, for odd $n \geq 3$.  For any input $\textbf{x}^*$ appearing in the inequality, i.e., $\textbf{x}^* \in \mathcal{X}^{n}_{\text{Merm}}$, there exists a box $\mathcal{P}_{\textbf{A} | \textbf{X}}$ violating the Mermin inequality maximally and obeying the relativistic causality constraints in Eq.(\ref{eq:rel-caus}), such that no randomness can be extracted from the outputs $\textbf{a}$ of the box under input $\textbf{x}^*$. In other words, we have 
\begin{equation}
\label{eq:det-hash}
\mathcal{P}_{\textbf{A} | \textbf{X}}(\textbf{a}^* | \textbf{x}^*) = 1,
\end{equation}   
for some fixed output bit string $\textbf{a}^*$.}
\begin{proof}
The proof is by explicit construction of the $n$-party box $\mathcal{P}$ satisfying the relativistic causality constraints Eq.(\ref{eq:rel-caus}) and violating the Mermin inequality maximally that returns a deterministic output for the fixed input $\textbf{x}^* \in \mathcal{X}^{n}_{\text{Merm}}$. 

Let us first list the set of constraints that $\mathcal{P}$ must obey. Split the set of inputs $\mathcal{X}^{n}_{\text{Merm}}$ appearing in the Mermin inequality into two subsets 
\begin{eqnarray}
\mathcal{X}^{n,-1}_{\text{Merm}} & := & \{\textbf{x} | \sum_i x_i = n - 2k, k \in \mathbb{Z}, \text{k odd}\} \nonumber \\
\mathcal{X}^{n,1}_{\text{Merm}} & := &\{\textbf{x} | \sum_i x_i = n - 2k, k \in \mathbb{Z}, \text{k even}\}. 
\end{eqnarray}
Then for the inputs $\textbf{x}$ in the set $\mathcal{X}^{n,-1}_{\text{Merm}}$, the box $\mathcal{P}$ is required to give outputs such that the value of the $n$-party correlation function is $-1$, i.e., 
\begin{eqnarray}
\label{eq:constraint-1}
\langle x_1 \dots x_n \rangle = P_{\textbf{A} | \textbf{X}}(\oplus_{i=1}^{n} a_i = 0 | x_1, \dots, x_n) - P_{\textbf{A} | \textbf{X}}(\oplus_{i=1}^{n} a_i = 1 | x_1, \dots, x_n) = -1, \quad \forall \textbf{x} \in \mathcal{X}^{n,-1}_{\text{Merm}}.
\end{eqnarray}
Similarly for the inputs $\textbf{x} \in \mathcal{X}^{n,1}_{\text{Merm}}$, the outputs from the box are required to yield a value of $1$ for the $n$-party correlation function, i.e.,
\begin{eqnarray}
\label{eq:constraint-2}
\langle x_1 \dots x_n \rangle = P_{\textbf{A} | \textbf{X}}(\oplus_{i=1}^{n} a_i = 0 | x_1, \dots, x_n) - P_{\textbf{A} | \textbf{X}}(\oplus_{i=1}^{n} a_i = 1 | x_1, \dots, x_n) = 1, \quad \forall \textbf{x} \in \mathcal{X}^{n,1}_{\text{Merm}}.
\end{eqnarray}
For the given fixed input $\textbf{x}^* \in \mathcal{X}^{n}_{\text{Merm}}$, the box $\mathcal{P}$ is required to return a deterministic input $\textbf{a}^*$, this is the constraint from Eq.(\ref{eq:det-hash}). The box $\mathcal{P}$ is also required to obey all of the relativistic causality constraints from Eq.(\ref{eq:rel-caus}). Finally, as usual we have the non-negativity and normalization constraints for the probability distributions, i.e.,
\begin{eqnarray}
\label{eq:norm-nonneg} 
\mathcal{P}_{\textbf{A} | \textbf{X}}(\textbf{a} | \textbf{x}) \geq 0 \; \; \forall \textbf{a}, \textbf{x} \nonumber \\
\sum_{\textbf{a}} \mathcal{P}_{\textbf{A} | \textbf{X}}(\textbf{a} | \textbf{x}) = 1 \; \; \forall \textbf{x}. 
\end{eqnarray}
We will show a binary-tree algorithmic construction of $\mathcal{P}_{\textbf{A} | \textbf{X}}$ obeying all the constraints Eqs.(\ref{eq:rel-caus}, \ref{eq:det-hash}, \ref{eq:constraint-1}, \ref{eq:constraint-2}, \ref{eq:norm-nonneg}) for given $\textbf{x}^* \in \mathcal{X}^{n,1}_{\text{Merm}}$. 
The construction for $\textbf{x} \in \mathcal{X}^{n,-1}_{\text{Merm}}$ will follow along analogous lines.


\begin{algorithm}[h]
\caption{Construction of box $\mathcal{P}_{\textbf{A} | \textbf{X}}$}
\label{CHalgorithm}
\begin{algorithmic}[1]
\Procedure{construction of $\mathcal{P}_{\textbf{A} | \textbf{X}}$} {}
\State Let $\textbf{x}^* \in \mathcal{X}^{n,1}_{\text{Merm}}$ be given. Initiate as step $0$, $\textbf{a}^l(\textbf{x}^*) = \textbf{a}^r(\textbf{x}^*) = \textbf{a}^*$ (the all-$0$ bit string). 
\State At the $(2j+1)$-th step, $0 \leq j \leq \frac{n-1}{2}$, $\forall$ $1 \leq i_1 \leq \dots \leq i_{2j+1} \leq n$, 
if $\textbf{x}^*_{i_{2j+1}} = 0$ define
\begin{eqnarray}
\label{eq:odd-step-0}
\textbf{a}^l(\textbf{x}^* \oplus \textbf{1}^{i_1} \oplus \dots \oplus \textbf{1}^{i_{2j+1}}) &:=& \textbf{a}^l(\textbf{x}^* \oplus \textbf{1}^{i_1} \oplus \dots \oplus \textbf{1}^{i_{2j}}) \nonumber \\
\textbf{a}^r(\textbf{x}^* \oplus \textbf{1}^{i_1} \oplus \dots \oplus \textbf{1}^{i_{2j+1}}) &:=& \textbf{a}^r(\textbf{x}^* \oplus \textbf{1}^{i_1} \oplus \dots \oplus \textbf{1}^{i_{2j}}) \oplus \textbf{1}^{i_j}. \nonumber \\
\end{eqnarray} 
If on the other hand, $\textbf{x}^*_{i_{2j+1}} = 1$ define
\begin{eqnarray}
\label{eq:odd-step-1}
\textbf{a}^l(\textbf{x}^* \oplus \textbf{1}^{i_1} \oplus \dots \oplus \textbf{1}^{i_{2j+1}}) &:=& \textbf{a}^l(\textbf{x}^* \oplus \textbf{1}^{i_1} \oplus \dots \oplus \textbf{1}^{i_{2j}}) \oplus \textbf{1}^{i_{2j+1}} \nonumber \\
\textbf{a}^r(\textbf{x}^* \oplus \textbf{1}^{i_1} \oplus \dots \oplus \textbf{1}^{i_{2j+1}}) &:=& \textbf{a}^r(\textbf{x}^* \oplus \textbf{1}^{i_1} \oplus \dots \oplus \textbf{1}^{i_{2j}}). \nonumber \\
\end{eqnarray}
\State At the $2j$-th step $1 \leq j \leq  \frac{n-1}{2}$,  $\forall$ $1 \leq i_1 \leq \dots \leq i_{2j} \leq n$, if $\textbf{x}^*_{i_{2j}} = 0$ define
\begin{eqnarray}
\label{eq:even-step-0}
\textbf{a}^l(\textbf{x}^* \oplus \textbf{1}^{i_1} \oplus \dots \oplus \textbf{1}^{i_{2j}}) &:=& \textbf{a}^l(\textbf{x}^* \oplus \textbf{1}^{i_1} \oplus \dots \oplus \textbf{1}^{i_{2j-1}}) \oplus \textbf{1}^{i_{2j}} \nonumber \\
\textbf{a}^r(\textbf{x}^* \oplus \textbf{1}^{i_1} \oplus \dots \oplus \textbf{1}^{i_{2j}}) &:=& \textbf{a}^r(\textbf{x}^* \oplus \textbf{1}^{i_1} \oplus \dots \oplus \textbf{1}^{i_{2j-1}}). \nonumber \\
\end{eqnarray} 
If on the other hand, $\textbf{x}^*_{i_{2j}} = 1$ define
\begin{eqnarray}
\label{eq:even-step-1}
\textbf{a}^l(\textbf{x}^* \oplus \textbf{1}^{i_1} \oplus \dots \oplus \textbf{1}^{i_{2j}}) &:=& \textbf{a}^l(\textbf{x}^* \oplus \textbf{1}^{i_1} \oplus \dots \oplus \textbf{1}^{i_{2j-1}}) \nonumber \\
\textbf{a}^r(\textbf{x}^* \oplus \textbf{1}^{i_1} \oplus \dots \oplus \textbf{1}^{i_{2j}}) &:=& \textbf{a}^r(\textbf{x}^* \oplus \textbf{1}^{i_1} \oplus \dots \oplus \textbf{1}^{i_{2j-1}}) \oplus \textbf{1}^{i_{2j}}. \nonumber \\
\end{eqnarray}
\State $\forall \textbf{x}$, set 
\begin{eqnarray}
&&\mathcal{P}_{\textbf{A} | \textbf{X}}(\textbf{a}^l(\textbf{x}) | \textbf{x}) = \mathcal{P}_{\textbf{A} | \textbf{X}}(\textbf{a}^r(\textbf{x}) | \textbf{x}) = \frac{1}{2}, \nonumber \\
&&\mathcal{P}_{\textbf{A} | \textbf{X}}(\textbf{a} | \textbf{x}) = 0, \; \; \; \text{otherwise}.
\end{eqnarray}
\EndProcedure
\end{algorithmic}
\end{algorithm}

The proof that the binary tree algorithm yields a box with the desired properties is done by induction. First consider the base steps. We are given $\textbf{x}^* \in \mathcal{X}^{n,1}_{\text{Merm}}$, and a fixed output bit string of even parity $\textbf{a}^*$ (the all-$0$ string). 
The algorithm sets $\mathcal{P}_{\textbf{A} | \textbf{X}}(\textbf{a}^*|\textbf{x}^*) = 1$ and $\mathcal{P}_{\textbf{A} | \textbf{X}}(\textbf{a} | \textbf{x}^*) = 0$ for $\textbf{a} \neq \textbf{a}^*$. This assignment clearly satisfies both Eqs. (\ref{eq:det-hash}) and the Mermin constraint $\langle x_1^* \dots x_n^* \rangle_{\mathcal{P}} = 1$.
In particular, since the output distribution is deterministic for the chosen input $\textbf{x}^*$, no hash function chosen by the parties can extract any randomness from the outputs for this input. 

The next base step in the construction of box $\mathcal{P}_{\textbf{A} | \textbf{X}}$ is to define the conditional probability distribution of the outputs for the set of inputs $\textbf{x}$ that are at Hamming distance $1$ from $\textbf{x}^*$, i.e., the inputs obtained from $\textbf{x}^*$ by a bit flip at position $i$. Explicitly, we consider the set of inputs $\textbf{x} \in \mathcal{X}_{\text{Ham}(\textbf{x}^*, \textbf{x}) = 1}$ where $\mathcal{X}_{\text{Ham}(\textbf{x}^*, \textbf{x}) = 1} := \{\textbf{x}^i | \textbf{x}^i = \textbf{x}^* \oplus \textbf{1}^i \}$, $\textbf{1}^i$ denotes the all-0 bit string with a single bit value $1$ at position $i$, and $\oplus$ denotes the bit-wise $\textsc{xor}$ of the two bit strings. For $\textbf{x}^i \in \mathcal{X}_{\text{Ham}(\textbf{x}^*, \textbf{x}) = 1}$, the algorithm sets 
$\mathcal{P}(_{\textbf{A} | \textbf{X}}\textbf{a}^*| \textbf{x}^i) = 
\mathcal{P}_{\textbf{A} | \textbf{X}}(\textbf{a}^* \oplus \textbf{1}^i| \textbf{x}^i) = \frac{1}{2}$, 
and $\mathcal{P}_{\textbf{A} | \textbf{X}}(\textbf{a} | \textbf{x}^i) = 0$ otherwise. 
Observe that inputs from $\mathcal{X}_{\text{Ham}(\textbf{x}^*, \textbf{x}) = 1}$ have even parity and do not enter the Mermin inequality, i.e., $\mathcal{X}_{\text{Ham}(\textbf{x}^*, \textbf{x}) = 1} \cap \mathcal{X}^n_{\text{Merm}} = \emptyset$. From the above assignment, we have $\langle x_1^i \dots x_n^i \rangle_{\mathcal{P}_{\textbf{A} | \textbf{X}}^*} = 0$ for $\textbf{x}^i \in \mathcal{X}_{\text{Ham}(\textbf{x}^*, \textbf{x}) = 1}$. Furthermore, the assignment satisfies the relativistically causal constraints for the contiguous subsets to the left ($<i := S_{1,i-1}^{n}$) and right ($>i := S_{i+1,n}^{n}$) of bit $i$, i.e., 
\begin{widetext}
\begin{eqnarray}
\label{eq:rel-caus-i-*}
\mathcal{P}_{\textbf{A} | \textbf{X}}(\textbf{a}^*_{<i} | \textbf{x}^*_{< i}) &=& \sum_{\textbf{a'}_{> i}} \mathcal{P}_{\textbf{A} | \textbf{X}}(\textbf{a'}|\textbf{x}^i) = \sum_{\textbf{a''}_{>i}} \mathcal{P}_{\textbf{A} | \textbf{X}}(\textbf{a''}|\textbf{x}^*) = 1, \; \; \;  \forall \textbf{a'}, \textbf{a''} \; \text{with} \; \textbf{a'}_{< i} = \textbf{a''}_{< i} = \textbf{a}^{*}_{< i}\nonumber \\
\mathcal{P}_{\textbf{A} | \textbf{X}}(\textbf{a}^*_{>i} | \textbf{x}^*_{> i}) &=& \sum_{\textbf{a'}_{< i}} \mathcal{P}_{\textbf{A} | \textbf{X}}(\textbf{a'}|\textbf{x}^i) = \sum_{\textbf{a''}_{<i}} \mathcal{P}_{\textbf{A} | \textbf{X}}(\textbf{a''}|\textbf{x}^*) = 1, \; \; \;  \forall \textbf{a'}, \textbf{a''} \; \text{with} \; \textbf{a'}_{> i} = \textbf{a''}_{> i} = \textbf{a}^{*}_{> i}. \nonumber \\
\end{eqnarray}
\end{widetext}

Having shown the base steps, we now prove by induction that the above algorithmic construction works as advertised, i.e., that the box $\mathcal{P}$ it produces satisfies all the constraints in Eqs.(\ref{eq:rel-caus}, \ref{eq:det-hash}, \ref{eq:constraint-1}, \ref{eq:constraint-2}, \ref{eq:norm-nonneg}).  Clearly, Eq.(\ref{eq:det-hash}) and the normalization constraints (\ref{eq:norm-nonneg}) are satisfied by construction. So it remains to show that the box satisfies the relativistic causality constraints (\ref{eq:rel-caus}) and the Mermin inequality constraints (\ref{eq:constraint-1}) and (\ref{eq:constraint-2}).

Let us firstly consider the Mermin inequality constraints (\ref{eq:constraint-1}) and (\ref{eq:constraint-2}). These constraints appear at the $2j$-th (even) steps with $1 \leq j \leq \frac{n-1}{2}$ of the Algorithm. Assume as the inductive hypothesis that the Mermin constraints are satisfied at the $2(k-1)$-th step, i.e., that for all $1 \leq i_1 \leq \dots \leq i_{2k-2} \leq n$, we have
\begin{eqnarray}
\label{eq:induct-Mermin-corr}
\langle \textbf{x}^* \oplus \textbf{1}^{i_1} \oplus \dots  \textbf{1}^{i_{2k-2}} \rangle &=& 1 \;   \forall \textbf{x}^* \oplus \textbf{1}^{i_1} \oplus \dots \textbf{1}^{i_{2k-2}} \in \mathcal{X}_{\text{Merm}}^{n,1} \nonumber \\
\langle \textbf{x}^* \oplus \textbf{1}^{i_1} \oplus \dots  \textbf{1}^{i_{2k-2}} \rangle &=& -1 \; \forall \textbf{x}^* \oplus \textbf{1}^{i_1} \oplus \dots \textbf{1}^{i_{2k-2}} \in \mathcal{X}_{\text{Merm}}^{n,-1} \nonumber \\
\end{eqnarray}
Now, consider the scenario when $\textbf{x}^* \oplus \textbf{1}^{i_1} \oplus \dots \oplus \textbf{1}^{i_{2k}} \in \mathcal{X}^{n,1}_{\text{Merm}}$. We have two cases: either 
\begin{enumerate}
\item $\textbf{x}^* \oplus \textbf{1}^{i_1} \oplus \dots \oplus \textbf{1}^{i_{2k-2}} \in \mathcal{X}^{n,1}_{\text{Merm}}$ and $(\textbf{x}^*_{i_{2k-1}}, \textbf{x}^*_{i_{2k}}) \in \{(0,1), (1,0)\}$ or 
\item $\textbf{x}^* \oplus \textbf{1}^{i_1} \oplus \dots \oplus \textbf{1}^{i_{2k-2}} \in \mathcal{X}^{n,-1}_{\text{Merm}}$ and $(\textbf{x}^*_{i_{2k-1}}, \textbf{x}^*_{i_{2k}}) \in \{(0,0), (1,1)\}$.
\end{enumerate} 
In the case (1), we have by the inductive hypothesis Eq.(\ref{eq:induct-Mermin-corr}) that $\langle \textbf{x}^* \oplus \textbf{1}^{i_1} \oplus \dots \textbf{1}^{i_{2k-2}}  \rangle = 1$, i.e, both $\textbf{a}^l(\textbf{x}^* \oplus \textbf{1}^{i_1} \oplus \dots \textbf{1}^{i_{2k-2}})$ and $\textbf{a}^r(\textbf{x}^* \oplus \textbf{1}^{i_1} \oplus \dots \textbf{1}^{i_{2k-2}})$ have even parity. Then when $(\textbf{x}^*_{i_{2k-1}}, \textbf{x}^*_{i_{2k}}) = (0,1)$ we get by first applying Eq.(\ref{eq:odd-step-0}) and then Eq.(\ref{eq:even-step-1}) that 
\begin{eqnarray}
\textbf{a}^l(\textbf{x}^* \oplus \textbf{1}^{i_1} \oplus \dots \textbf{1}^{i_{2k}}) &=& \textbf{a}^l(\textbf{x}^* \oplus \textbf{1}^{i_1} \oplus \dots \textbf{1}^{i_{2k-2}}), \nonumber \\
\textbf{a}^r(\textbf{x}^* \oplus \textbf{1}^{i_1} \oplus \dots \textbf{1}^{i_{2k}}) &=& \textbf{a}^r(\textbf{x}^* \oplus \textbf{1}^{i_1} \oplus \dots \textbf{1}^{i_{2k-2}})  \oplus \textbf{1}^{2k-1} \oplus \textbf{1}^{2k}, \nonumber \\
\end{eqnarray}    
so that both outputs $\textbf{a}^{l/r}(\textbf{x}^* \oplus \textbf{1}^{i_1} \oplus \dots \textbf{1}^{i_{2k}})$ have even parity and we obtain $\langle \textbf{x}^* \oplus \textbf{1}^{i_1} \oplus \dots \textbf{1}^{i_{2k}}  \rangle = 1$ as required. Similarly when $(\textbf{x}^*_{i_{2k-1}}, \textbf{x}^*_{i_{2k}}) = (1,0)$ we get by first applying Eq.(\ref{eq:odd-step-1}) and then Eq.(\ref{eq:even-step-0}) that 
\begin{eqnarray}
\textbf{a}^l(\textbf{x}^* \oplus \textbf{1}^{i_1} \oplus \dots \textbf{1}^{i_{2k}}) &=& \textbf{a}^l(\textbf{x}^* \oplus \textbf{1}^{i_1} \oplus \dots \textbf{1}^{i_{2k-2}}) \oplus \textbf{1}^{2k-1} \oplus \textbf{1}^{2k}, \nonumber \\
\textbf{a}^r(\textbf{x}^* \oplus \textbf{1}^{i_1} \oplus \dots \textbf{1}^{i_{2k}}) &=& \textbf{a}^r(\textbf{x}^* \oplus \textbf{1}^{i_1} \oplus \dots \textbf{1}^{i_{2k-2}}), \nonumber \\
\end{eqnarray} 
so that again both outputs have even parity. 

In the case (2), we have by the inductive hypothesis Eq.(\ref{eq:induct-Mermin-corr}) that $\langle \textbf{x}^* \oplus \textbf{1}^{i_1} \oplus \dots \textbf{1}^{i_{2k-2}}  \rangle = -1$, i.e, both $\textbf{a}^l(\textbf{x}^* \oplus \textbf{1}^{i_1} \oplus \dots \textbf{1}^{i_{2k-2}})$ and $\textbf{a}^r(\textbf{x}^* \oplus \textbf{1}^{i_1} \oplus \dots \textbf{1}^{i_{2k-2}})$ have odd parity. Then when $(\textbf{x}^*_{i_{2k-1}}, \textbf{x}^*_{i_{2k}}) = (0,0)$ we get by first applying Eq.(\ref{eq:odd-step-0}) and then Eq.(\ref{eq:even-step-0}) that 
\begin{eqnarray}
\textbf{a}^l(\textbf{x}^* \oplus \textbf{1}^{i_1} \oplus \dots \textbf{1}^{i_{2k}}) &=& \textbf{a}^l(\textbf{x}^* \oplus \textbf{1}^{i_1} \oplus \dots \textbf{1}^{i_{2k-2}})  \oplus \textbf{1}^{2k}, \nonumber \\
\textbf{a}^r(\textbf{x}^* \oplus \textbf{1}^{i_1} \oplus \dots \textbf{1}^{i_{2k}}) &=& \textbf{a}^r(\textbf{x}^* \oplus \textbf{1}^{i_1} \oplus \dots \textbf{1}^{i_{2k-2}}) \oplus \textbf{1}^{2k-1}, \nonumber \\
\end{eqnarray}    
so that a bit flip occurs on both outputs $\textbf{a}^{l/r}(\textbf{x}^* \oplus \textbf{1}^{i_1} \oplus \dots \textbf{1}^{i_{2k}})$ which now have even parity and we obtain $\langle \textbf{x}^* \oplus \textbf{1}^{i_1} \oplus \dots \textbf{1}^{i_{2k}}  \rangle = 1$ as required. Finally, when $(\textbf{x}^*_{i_{2k-1}}, \textbf{x}^*_{i_{2k}}) = (1,1)$ we get by first applying Eq.(\ref{eq:odd-step-1}) and then Eq.(\ref{eq:even-step-1}) that 
\begin{eqnarray}
\textbf{a}^l(\textbf{x}^* \oplus \textbf{1}^{i_1} \oplus \dots \textbf{1}^{i_{2k}}) &=& \textbf{a}^l(\textbf{x}^* \oplus \textbf{1}^{i_1} \oplus \dots \textbf{1}^{i_{2k-2}})  \oplus \textbf{1}^{2k-1}, \nonumber \\
\textbf{a}^r(\textbf{x}^* \oplus \textbf{1}^{i_1} \oplus \dots \textbf{1}^{i_{2k}}) &=& \textbf{a}^r(\textbf{x}^* \oplus \textbf{1}^{i_1} \oplus \dots \textbf{1}^{i_{2k-2}}) \oplus \textbf{1}^{2k}, \nonumber \\
\end{eqnarray}  
so that again the outputs have even parity as required. 

Now, consider the scenario when $\textbf{x}^* \oplus \textbf{1}^{i_1} \oplus \dots \oplus \textbf{1}^{i_{2k}} \in \mathcal{X}^{n,-1}_{\text{Merm}}$. We again have two cases: either 
\begin{enumerate}
\item $\textbf{x}^* \oplus \textbf{1}^{i_1} \oplus \dots \oplus \textbf{1}^{i_{2k-2}} \in \mathcal{X}^{n,1}_{\text{Merm}}$ and $(\textbf{x}^*_{i_{2k-1}}, \textbf{x}^*_{i_{2k}}) \in \{(0,0), (1,1)\}$ or 
\item $\textbf{x}^* \oplus \textbf{1}^{i_1} \oplus \dots \oplus \textbf{1}^{i_{2k-2}} \in \mathcal{X}^{n,-1}_{\text{Merm}}$ and $(\textbf{x}^*_{i_{2k-1}}, \textbf{x}^*_{i_{2k}}) \in \{(0,1), (1,0)\}$.
\end{enumerate} 
As in the previous scenario, an exhaustive check of the output parities in the case (1) and (2) shows that the Mermin constraints are satisfied. For clarity, we explicitly provide these here. In the case (1), we have by the inductive hypothesis Eq.(\ref{eq:induct-Mermin-corr}) that both $\textbf{a}^l(\textbf{x}^* \oplus \textbf{1}^{i_1} \oplus \dots \textbf{1}^{i_{2k-2}})$ and $\textbf{a}^r(\textbf{x}^* \oplus \textbf{1}^{i_1} \oplus \dots \textbf{1}^{i_{2k-2}})$ have even parity. Then when $(\textbf{x}^*_{i_{2k-1}}, \textbf{x}^*_{i_{2k}}) = (0,0)$, applying first Eq.(\ref{eq:odd-step-0}) and then (\ref{eq:even-step-0}) we get that  
\begin{eqnarray}
\textbf{a}^l(\textbf{x}^* \oplus \textbf{1}^{i_1} \oplus \dots \textbf{1}^{i_{2k}}) &=& \textbf{a}^l(\textbf{x}^* \oplus \textbf{1}^{i_1} \oplus \dots \textbf{1}^{i_{2k-2}})  \oplus \textbf{1}^{2k}, \nonumber \\
\textbf{a}^r(\textbf{x}^* \oplus \textbf{1}^{i_1} \oplus \dots \textbf{1}^{i_{2k}}) &=& \textbf{a}^r(\textbf{x}^* \oplus \textbf{1}^{i_1} \oplus \dots \textbf{1}^{i_{2k-2}}) \oplus \textbf{1}^{2k-1}, \nonumber \\
\end{eqnarray}  
so that now both the outputs $\textbf{a}^{l/r}(\textbf{x}^* \oplus \textbf{1}^{i_1} \oplus \dots \textbf{1}^{i_{2k}})$ have odd parity. Analogous results are obtained when $(\textbf{x}^*_{i_{2k-1}}, \textbf{x}^*_{i_{2k}}) = (1,1)$, i.e., applying first Eq.(\ref{eq:odd-step-1}) and then (\ref{eq:even-step-1}) we get that 
\begin{eqnarray}
\textbf{a}^l(\textbf{x}^* \oplus \textbf{1}^{i_1} \oplus \dots \textbf{1}^{i_{2k}}) &=& \textbf{a}^l(\textbf{x}^* \oplus \textbf{1}^{i_1} \oplus \dots \textbf{1}^{i_{2k-2}})  \oplus \textbf{1}^{2k-1}, \nonumber \\
\textbf{a}^r(\textbf{x}^* \oplus \textbf{1}^{i_1} \oplus \dots \textbf{1}^{i_{2k}}) &=& \textbf{a}^r(\textbf{x}^* \oplus \textbf{1}^{i_1} \oplus \dots \textbf{1}^{i_{2k-2}}) \oplus \textbf{1}^{2k}, \nonumber \\
\end{eqnarray}  
so that the resulting outputs have again odd parity. 

In the case (2), we have by the inductive hypothesis Eq.(\ref{eq:induct-Mermin-corr}) that both $\textbf{a}^l(\textbf{x}^* \oplus \textbf{1}^{i_1} \oplus \dots \textbf{1}^{i_{2k-2}})$ and $\textbf{a}^r(\textbf{x}^* \oplus \textbf{1}^{i_1} \oplus \dots \textbf{1}^{i_{2k-2}})$ have odd parity. Then when $(\textbf{x}^*_{i_{2k-1}}, \textbf{x}^*_{i_{2k}}) = (0,1)$, applying first Eq.(\ref{eq:odd-step-0}) and then (\ref{eq:even-step-1}) we get that 
\begin{eqnarray}
\textbf{a}^l(\textbf{x}^* \oplus \textbf{1}^{i_1} \oplus \dots \textbf{1}^{i_{2k}}) &=& \textbf{a}^l(\textbf{x}^* \oplus \textbf{1}^{i_1} \oplus \dots \textbf{1}^{i_{2k-2}}), \nonumber \\
\textbf{a}^r(\textbf{x}^* \oplus \textbf{1}^{i_1} \oplus \dots \textbf{1}^{i_{2k}}) &=& \textbf{a}^r(\textbf{x}^* \oplus \textbf{1}^{i_1} \oplus \dots \textbf{1}^{i_{2k-2}}) \oplus \textbf{1}^{2k-1} \oplus \textbf{1}^{2k}, \nonumber \\
\end{eqnarray}  
so that the resulting outputs have again odd parity. Similarly, when $(\textbf{x}^*_{i_{2k-1}}, \textbf{x}^*_{i_{2k}}) = (1,0)$, applying first Eq.(\ref{eq:odd-step-1}) and then (\ref{eq:even-step-0}) we get that 
\begin{eqnarray}
\textbf{a}^l(\textbf{x}^* \oplus \textbf{1}^{i_1} \oplus \dots \textbf{1}^{i_{2k}}) &=& \textbf{a}^l(\textbf{x}^* \oplus \textbf{1}^{i_1} \oplus \dots \textbf{1}^{i_{2k-2}}) \oplus \textbf{1}^{2k-1} \oplus \textbf{1}^{2k}, \nonumber \\
\textbf{a}^r(\textbf{x}^* \oplus \textbf{1}^{i_1} \oplus \dots \textbf{1}^{i_{2k}}) &=& \textbf{a}^r(\textbf{x}^* \oplus \textbf{1}^{i_1} \oplus \dots \textbf{1}^{i_{2k-2}}), \nonumber \\
\end{eqnarray}  
so that the resulting outputs have again odd parity. Having exhaustively checked all the cases, we conclude by the inductive hypothesis that the Mermin constraints (\ref{eq:constraint-1}) and (\ref{eq:constraint-2}) are satisfied at all the even steps of the Algorithm. It remains to show that the box $\mathcal{P}$ satisfies the relativistic causality constraints (\ref{eq:rel-caus}). 

The relativistic causality constraints (\ref{eq:rel-caus}) effectively say that the output probability distribution for any contiguous set of parties is independent of the complementary set of parties' inputs. 
\begin{eqnarray}
\label{eq:rel-caus-pf}
\sum_{\textbf{a'}_{(S_{m,k}^{n})^{c}}} P_{\textbf{A} | \textbf{X}}(\textbf{a'} | \textbf{x'}) = \sum_{\textbf{a''}_{(S_{m,k}^{n})^{c}}} P_{\textbf{A} | \textbf{X}}(\textbf{a''} | \textbf{x''})
\end{eqnarray}
The structure of the algorithm is such that from step $m$ to step $m+1$ ($m$ here is arbitrary), in going from input $\textbf{x}^* \oplus \textbf{1}^{i_1} \oplus \dots \oplus \textbf{1}^{i_{m}}$ to input $\textbf{x}^* \oplus \textbf{1}^{i_1} \oplus \dots \oplus \textbf{1}^{i_{m+1}}$, a bit flip is applied to either $\textbf{a}^l$ or $\textbf{a}^r$ \textit{at position $m +1$} and depends only on $\textbf{x}^*_{i_{m+1}}$. This implies that the probability distribution of the outputs to the left and right of position $m+1$ are unaffected, i.e., the probability distribution for the parties in the set $S_{1,m}^{n}$ and those in the set $S^{n}_{m+2,n}$ remain fixed under this operation (recall that $S_{m,k}^{n} = \{m, m+1, \dots, m+k-1\}$). This ensures that the relativistic causality constraints are satisfied, since for any two input strings $\textbf{x}'$ and $\textbf{x}''$ that share a common contiguous bit sequence $\textbf{x}'_{S^{n}_{m,k}} = \textbf{x}''_{S^{n}_{m,k}}$,  
both $\sum_{\textbf{a'}_{(S_{m,k}^{n})^{c}}} P_{\textbf{A} | \textbf{X}}(\textbf{a'} | \textbf{x'})$ and $\sum_{\textbf{a''}_{(S_{m,k}^{n})^{c}}} P_{\textbf{A} | \textbf{X}}(\textbf{a''} | \textbf{x''})$ only depend explicitly on the sequence of bits $\textbf{x}^*_{i}$ for $i \in S_{m,k}$.     
We have thus shown the box $\mathcal{P}$ that satisfies the Mermin constraints and the relativistic causality constraints in addition to returning a deterministic output $\textbf{a}^*$ for input $\textbf{x}^*$.   
\end{proof}

\textbf{Proposition 18.} \textit{Consider a three-party Bell scenario, with Alice, Bob and Charlie each performing two measurements $x, y, z \in \{0,1\}$ of two outcomes $a, b, c \in \{0,1\}$ respectively. Suppose that in some inertial reference frame, the three space-like separated parties are arranged in $1$-D, with $r_A < r_B < r_C$ and perform their measurements simultaneously, i.e., $t_A = t_B = t_C$. Then, there exists a three-party relativistically causal box $\mathcal{P}(a,b,c|x,y,z)$ such that
\begin{equation}
\label{eq:rel-caus-mono}
\langle CHSH \rangle_{AB} + \langle CHSH \rangle_{BC} = 8.
\end{equation}  } 
\begin{proof}
The relativistic causality constraints in this three-party scenario for the specific measurement configuration are given explicitly from Proposition \ref{prop:rel-cau-constraints} as
\begin{eqnarray}
\label{eq:rel-caus-3-party}
\sum_{a=0,1} P_{A,B,C|X,Y,Z}(a,b,c|x,y, z) &=& \sum_{a=0,1} P_{A,B,C|X,Y,Z}(a,b,c|x',y,z) \; \; \forall x,x',y,z,b,c \nonumber \\
\sum_{c=0,1} P_{A,B,C|X,Y,Z}(a,b,c|x,y,z) &=& \sum_{c=0,1} P_{A,B,C|X,Y,Z}(a,b,c|x,y,z') \; \; \forall z,z',x,y,a,b \nonumber \\
\sum_{b,c=0,1} P_{A,B,C|X,Y,Z}(a,b,c|x,y,z) &=& \sum_{b,c=0,1} P_{A,B,C|X,Y,Z}(a,b,c|x,y',z') \; \; \forall y,y',z,z',x,a \nonumber \\
\sum_{a,b=0,1} P_{A,B,C|X,Y,Z}(a,b,c|x,y,z) &=& \sum_{a,b=0,1} P_{A,B,C|X,Y,Z}(a,b,c|x',y',z) \; \; \forall x,x',y,y',z,c. \nonumber \\
\end{eqnarray}
\begin{eqnarray}
\mathcal{P}_{A,B,C|X,Y,Z}(0,0,1|0,1,1) &=& \mathcal{P}_{A,B,C|X,Y,Z}(1,1,0|0,1,1) = \frac{1}{2} \nonumber \\
\mathcal{P}_{A,B,C|X,Y,Z}(0,1,1|1,1,0) &=& \mathcal{P}_{A,B,C|X,Y,Z}(1,0,0|1,1,0) = \frac{1}{2} \nonumber \\
\mathcal{P}_{A,B,C|X,Y,Z}(0,1,0|1,1,1) &=& \mathcal{P}_{A,B,C|X,Y,Z}(1,0,1|1,1,1) = \frac{1}{2} \nonumber \\
\mathcal{P}_{A,B,C|X,Y,Z}(0,0,0|x,y,z) &=& \mathcal{P}_{A,B,C|X,Y,Z}(1,1,1|x,y,z) = \frac{1}{2} \; \; \;\;\text{else}
\end{eqnarray}

It can be readily verified that the box $\mathcal{P}$ as defined above satisfies all the relativistic causality constraints in Eq.(\ref{eq:rel-caus-3-party}). Furthermore, we also see that when Alice-Bob or Bob-Charlie measure input $(11)$, they get anti-correlated outputs, while for other inputs, they get correlated outcomes. This gives that the values of the CHSH Bell expressions $\langle CHSH \rangle_{AB}$ and $\langle CHSH \rangle_{AC}$ are both maximal (equal to the algebraic value of $4$). In other words, the phenomenon of monogamy for the CHSH violating correlations does not hold under the relativistic causality conditions in certain measurement configurations. 
\end{proof}



\begin{thebibliography}{99}

\bibitem{EPR}
A. Einstein, B. Podolsky and N. Rosen,
Physical Review \textbf{47} (10): 777 (1935). 

\bibitem{Bell}
J. S. Bell, 
Physics (Long Island City, N.Y.), \textbf{1}, 195 (1964).

\bibitem{Hensen}
B. Hensen et al., 
Nature \textbf{526}, 682686 (2015).

\bibitem{Giustina}
M. Giustina et al., 
Phys. Rev. Lett. \textbf{115}, 250401 (2015).

\bibitem{Shalm}
L. K. Shalm et al., 
Phys. Rev. Lett. \textbf{115}, 250402 (2015).


\bibitem{Pironio}
S. Pironio, A.~Acin, S.~Massar, A.~Boyer de~la~Giroday, D.N.~Matsukevich, P.~Maunz, S.~Olmschenk, D.~Hayes, L.~Luo, T.A.~Manning and C.~Monroe, 
\textit{Random numbers certified by Bell's theorem}, 
Nature \textbf{464}, 1021 (2010).

\bibitem{VV}
U. Vazirani and T. Vidick,
Phys. Rev. Lett. \textbf{113}, 140501 (2014).

\bibitem{our}
F.G.S.L. Brand\~{a}o, R. Ramanathan, A. Grudka, K. Horodecki, M. Horodecki, P. Horodecki, T.~Szarek and H. Wojew\'{o}dka,
Nat. Comm. \textbf{7}, 11345 (2016).

\bibitem{Hall}
M. J. W. Hall,
arXiv:1511.00729 (2015).

\bibitem{Valdenebro}
A. G. Valdenebro,
\textit{Assumptions Underlying Bell's Inequalities},
Eur.J.Phys. \textbf{23}, 1 (2002). 

\bibitem{Tumulka}
R. Tumulka,
\textit{The Assumptions of Bell's Proof},
In M. Bell and S. Gao (eds.), \textit{Quantum Nonlocality and Reality - 50 Years of Bell's Theorem}, Cambridge University Press (2016).

\bibitem{Norsen}
T. Norsen,
\textit{Local Causality and Completeness: Bell vs. Jarrett},
Found. of Phys., Vol. 39, Issue 3, 273 (2009).

\bibitem{Bell3}
 J. S. Bell, 
\textit{La nouvelle cuisine}, in 
\textit{Speakable and Unspeakable in Quantum Mechanics}, 2nd ed., Cambridge University Press (2004).

\bibitem{ZB14}
M. Zukowski and C. Brukner, 
\textit{Quantum non-locality - It ain't necessarily so...},
J. Phys. A: Math. Theor. \textbf{47}, 424009 (2014).

\bibitem{Shim83}
A. Shimony, in 
\textit{Foundations of Quantum Mechanics in Light of the New Technology}, S.
Kamefuchi et al., eds. (Tokyo, Japan Physical Society) (1983).


\bibitem{CC}
H. Buhrman, R. Cleve, S. Massar and R. de Wolf, 
\textit{Non-locality and Communication Complexity},
arXiv:0907.3584 (2009).

\bibitem{WKD00}
L. J. Wang, A. Kuzmich and A. Dogariu,
\textit{Gain-assisted superluminal light propagation},
Nature \textbf{406}, 277 (2000). 

\bibitem{PR}
S. Popescu and D. Rohrlich, 
Found. Phys. \textbf{24}, 379 (1994).

\bibitem{GPR}
J. Grunhaus, S. Popescu, and D. Rohrlich,
\textit{Jamming nonlocal quantum correlations}.,
Phys. Rev. A \textbf{53}, 3781 (1996).

\bibitem{GR}
J. B. Hartle,
\textit{Gravity: An Introduction to Einstein's General Relativity},
Addison-Wesley (2003).

\bibitem{Reich}
H. Reichenbach, 
\textit{The Direction of Time}, Berkeley, University of Los Angeles Press (1956).

\bibitem{Fritz}
T. Fritz,
\textit{Beyond Bell's Theorem II: Scenarios with arbitrary causal structure},
Comm. Math. Phys. \textbf{341}(2), 391-434 (2016).

\bibitem{Coecke}
B. Coecke and R. Lal,
\textit{Time-asymmetry of probabilities versus relativistic causal structure: an arrow of time},
Phys. Rev. Lett. \textbf{108}, 200403 (2012).

\bibitem{KT85}
L. Khalfin and B. Tsirelson, in \textit{Symposium on the Foundations of Modern Physics ’85}, P.
Lahti et al., eds. (World-Scientific, Singapore, 1985), p. 441; 

\bibitem{Rastall}
P. Rastall, 
Found. Phys. \textbf{15}, 963 (1985).

\bibitem{SW87}
S.J. Summers and R.F. Werner, 
J. Math. Phys. \textbf{28}, 2440 (1987).

\bibitem{KS94}
G. Krenn and K. Svozil, arXiv:quant-ph/9503010 (1994).

\bibitem{Principles2}
M. Pawlowski, T. Paterek, D. Kaszlikowski, V. Scarani, A. Winter and M. Zukowski,
Nature \textbf{461}, 1101 (2009). 

\bibitem{Principles3}
M. Navascu\'{e}s and H. Wunderlich, 
Proc. Royal Soc. A \textbf{466}: 881 (2009). 

\bibitem{PR97}
S. Popescu and D. Rohrlich,
\textit{Causality and Nonlocality as Axioms for Quantum Mechanics},
In the \textit{Proceedings of the symposium on Causality and Locality in Modern Physics and Astronomy: Open Questions and Possible Solutions} (York University, Toronto, 1997).

\bibitem{CHSH}
J.F. Clauser, M.A. Horne, A. Shimony and R.A. Holt, 
\textit{Proposed experiment to test local hidden-variable theories}, 
Phys. Rev. Lett. \textbf{23} (15): 880 (1969).

\bibitem{Toner}
B. Toner, 
Proc. R. Soc. A \textbf{465}, 59 (2009).

\bibitem{BDS62}
O. M. P. Bilaniuk, V. K. Deshpande and E. C. G. Sudarshan,
\textit{"Meta" Relativity},
Amer. J. Phys. \textbf{30}, 718 (1962).

\bibitem{Sve87}
G. Svetlichny, 
Phys. Rev. D \textbf{35}, 3066 (1987).

\bibitem{Haag}
R. Haag, 
\textit{Local Quantum Physics}, Springer Verlag, New York (1992).

\bibitem{BBGP13}
J-D. Bancal, J. Barrett, N. Gisin and S. Pironio
\textit{The definition of multipartite nonlocality},
Phys. Rev. A \textbf{88}, 014102 (2013).

\bibitem{GWAN12}
R. Gallego, L. E. W\"{u}rflinger, Antonio Ac\'{i}n, Miguel Navascu\'{e}s,
\textit{An operational framework for nonlocality},
Phys. Rev. Lett. \textbf{109}, 070401 (2012).



\bibitem{Terletskii68}
Ya. P. Terletskii, 
\textit{Paradoxes in the Theory of Relativity},
Plenum Press, New York (1968).

\bibitem{BFKL71}
A. Bers, R. Fox, C. G. Kuper and S. G. Lipson, 
\textit{The impossibility of free tachyons}, in 
\textit{Relativity and Gravitation}, eds. C. G. Kuper and Asher Peres, New York, Gordon and Breach Science Publishers, 41 (1971). 

\bibitem{Wei89} 
S. Weinberg, Ann.Phys. (N.Y.) \textbf{194}, 336 (1989).

\bibitem{Gis90}
N. Gisin, 
Phys.Lett. A \textbf{143}, 1 (1990).

\bibitem{Cza91}
M. Czachor, Found.Phys.Lett. \textbf{4}, 351 (1991).

\bibitem{Pol91}
J. Polchinski,
\textit{Weinberg's Nonlinear Quantum Mechanics and the Einstein-Podolsky-Rosen Paradox},
Phys. Rev. Lett. \textbf{66} 4, 397 (1991). 

\bibitem{H-H09}
C. Hewitt-Horsman,
Found Phys, \textbf{39}: 869 (2009).

\bibitem{R15}
D. Rohrlich,
arXiv:1507.01588, in
\textit{Quantum Nonlocality and Reality: 50 Years of Bell's theorem}, eds. S. Gao and M. Bell (Cambridge U. Press), (2015). 

\bibitem{R14}
D. Rohrlich,
arXiv:1407.8530, in \textit{Quantum Theory: A Two-Time Success Story} (Yakir Aharonov Festschrift), eds. D. C. Struppa and J. M. Tollaksen (New York: Springer), pp. 205 (2013). 


\bibitem{self-cons}
J. Friedman, M. Morris, I. Novikov, F. Echeverria, G. Klinkhammer, K. Thorne, U. Yurtsever,
\textit{Cauchy problem in spacetimes with closed timelike curves}, 
Phys. Rev. D. \textbf{42} (6): 1915 (1990).

\bibitem{SWH}
S. W. Hawking,
\textit{Chronology protection conjecture}, 
Phys. Rev. D \textbf{46}, 603 (1992). 

\bibitem{GRT86}
G. C. Ghirardi, A. Rimini, T. Weber, Phys. Rev. D 34 
470 (1986).

\bibitem{BPAL+12}
J.-D. Bancal, S. Pironio, A. Ac\'{i}n, Y.-C. Liang, V. Scarani and N. Gisin,
\textit{Quantum nonlocality based on finite-speed causal influences leads to superluminal signaling}, 
Nature Physics \textbf{8}, 867 (2012). 

\bibitem{BBLG13}
T. J. Barnea, J.-D. Bancal, Y.-C. Liang and N. Gisin,
\textit{Tripartite quantum state violating the hidden influence constraints},
Phys. Rev. A \textbf{88}, 022123 (2013).


\bibitem{AL}
Y. Aharonov, L. Vaidman,
\textit{Sending Signals to Space-Like Separated Regions},
arXiv:quant-ph/0102083 (2001). 


\bibitem{BHK}
J. Barrett, L. Hardy and A. Kent,
\textit{No Signalling and Quantum Key Distribution},
Phys. Rev. Lett. \textbf{95}, 010503 (2005).

\bibitem{BCK}
J. Barrett, R. Colbeck and A. Kent,
\textit{Unconditionally secure device-independent quantum key distribution with only two devices},
Phys. Rev. A \textbf{86}, 062326 (2012).


\bibitem{BKP}
J. Barrett, A. Kent and S. Pironio,
\textit{Maximally Non-Local and Monogamous Quantum Correlations},
Phys. Rev. Lett. \textbf{97}, 170409 (2006).

\bibitem{Renner-Colbeck}
R. Colbeck and R. Renner,
\textit{A short note on the concept of free choice},
arXiv:1302.4446 (2013).

\bibitem{Bell2}
J. S. Bell, 
\textit{Free variables and local causality}. In \textit{Speakable and unspeakable in quantum mechanics}, chap. 12
(Cambridge University Press, 1987).

\bibitem{CR}
R. Colbeck and R. Renner,
\textit{Free randomness can be amplified},
Nature Physics \textbf{8}, 450-454 (2012).
\bibitem{Acin}
R. Gallego, Ll. Masanes, G. de la Torre, C. Dhara, L. Aolita and A. Ac\'{i}n,
\textit{Full randomness from arbitrarily deterministic events},
Nature Communications \textbf{4}, 2654 (2013).

\bibitem{Mas06}
L. Masanes, A. Acin and N. Gisin,
Phys. Rev. A. \textbf{73}, 012112 (2006). 

\bibitem{DTA14}
C. Dhara, G. de la Torre, A. Ac\'{i}n,
\textit{Can observed randomness be certified to be fully intrinsic?},
Phys. Rev. Lett. \textbf{112}, 100402 (2014).

\bibitem{GHZ}
D. M. Greenberger, M. A. Horne, and A. Zeilinger, in 
\textit{Bell's Theorem, Quantum Theory, and Conceptions of the Universe}
(Kluwer, Dordrecht), 69 (1989).

\bibitem{Pearle}
P. M. Pearle, 
\textit{Hidden-variable example based upon data rejection},
Physical Review D \textbf{2}, 1418 (1970).

\bibitem{BC}
S. L. Braunstein and C. M. Caves, Annals of Physics {\bf 202}, 22 (1990).


\bibitem{OCB13}
O. Oreshkov, F. Costa and C. Brukner,
\textit{Quantum correlations with no causal order},
Nature Communications \textbf{3}, 1092 (2012).

\bibitem{our2}
R. Ramanathan, F.G.S.L. Brand\~{a}o, K. Horodecki, M. Horodecki, P. Horodecki and H. Wojew\'{o}dka,
\textit{Randomness amplification against no-signaling adversaries using two devices},
arXiv:1504.06313 (2015).


\bibitem{ER89}
P. H. Eberhard and R. R. Ross,
\textit{Quantum Field Theory cannot provide Faster-Than-Light Communication},
Foundations of Physics Letters, Vol. 2, No. 2 (1989). 

\bibitem{PH99}
K. A. Peacock, B. S. Hepburn,
\textit{Begging the Signalling Question: Quantum Signalling and the Dynamics of Multiparticle Systems},
arXiv:quant-ph/9906036 (1999). 

\bibitem{Hol93}
Peter R. Holland, 
\textit{The Quantum Theory of Motion: An Account of the De Broglie-Bohm Causal Interpretation of Quantum Mechanics}, Cambridge University Press, Cambridge (1993). 

\bibitem{BH93}
D. Bohm and B. J. Hiley, 
\textit{The Undivided Universe: An Ontological Interpretation of Quantum Theory},
Routledge (1993). 

\bibitem{CR2}
R. Colbeck and R. Renner,
\textit{No extension of quantum theory can have improved predictive power},
Nature Communications \textbf{2}, 411 (2011). 

\bibitem{Pearl09}
J. Pearl, \textit{Causality}, Cambridge University Press (2009).

\bibitem{FKL69}
R. Fox, C. G. Kuper and S. G. Lipson,
\textit{Do Faster-than-Light Group Velocities imply Violation of Causality?},
Nature \textbf{223}, 597 (1969). 

\bibitem{Brillouin}
L. Brillouin, 
\textit{Wave Propagation and Group Velocity},
Academic Press, New York, (1960).

\bibitem{HC12}
J. M. Hill and B. J. Cox,
\textit{Einstein's special relativity beyond the speed of light},
Proc. R. Soc. A, \textbf{468}, 2148 (2012). 

\bibitem{SS86}
R. I. Sutherland and J. R. Shepanski,
\textit{Superluminal reference frames and generalized Lorentz transformations},
Phys. Rev. D \textbf{33}, 2896 (1986).



\bibitem{Feinberg67}
G. Feinberg,
\textit{On the possibility of faster-than-light particles},
Phys. Rev. \textbf{159}, 1089 (1967). 

\bibitem{Eckstein1}
M. Eckstein and T. Miller,
\textit{Causality for nonlocal phenomena},
 arXiv:1510.06386 (2015). 


\bibitem{Miller}
T. Miller,
\textit{Polish spaces of causal curves}, 
arXiv:1609.09488 (2016).


\bibitem{Eckstein-Miller-models}
M. Eckstein and T. Miller,
\textit{Causal evolution of wave packets},
arXiv:1610.00764 (2016). 


\bibitem{our3}
P. Horodecki, R. Ramanathan, \textit{in preparation}. \\


\end{thebibliography}
\end{document}